\documentclass[sigconf]{acmart} 

\def\EPARA{{\rm E}\kern-.03em{\rm\sc\large p\kern-.065em a\kern-.03em r\kern-.01em a}\xspace} 
\def\EPARAbf{{\rm\bf E}\kern-.05em{\sc\bfseries\large p\kern-.065em a\kern-.03em r\kern-.01em a}\xspace} 
\def\EPARAit{{\rm\itshape E}\kern-.05em{\sc\itshape\large p\kern-.065em a\kern-.03em r\kern-.01em a}\xspace} 
\setcopyright{none} 
\renewcommand\footnotetextcopyrightpermission[1]{} 
\usepackage{verbatim}
\usepackage{nicefrac}
\usepackage{siunitx}
\usepackage{amsfonts}
\usepackage{amsmath}
\usepackage{amsthm}
\usepackage{amssymb}
\usepackage{array,framed}
\usepackage{enumitem}
\usepackage{makecell}
\usepackage{booktabs}
\usepackage{
  color,
  float,
  epsfig,
  wrapfig,
  graphics,
  graphicx,
  subcaption
}
\usepackage{textcomp,amssymb}
\usepackage{setspace}
\usepackage{latexsym,fancyhdr,url}
\usepackage{enumerate}
\usepackage{algpseudocode}
\usepackage{graphics}
\usepackage{subcaption}
\usepackage{xparse} 
\usepackage{xspace}
\usepackage{multirow}
\usepackage{csvsimple}
\usepackage{balance}
\usepackage{
  tikz,
  pgfplots,
  pgfplotstable
}
\usepackage{url}

\usepackage{xurl} 
\usepackage{xcolor}
\usepackage{hyperref}
\usepackage{cleveref}
\crefformat{section}{\S#2#1#3} 
\crefformat{subsection}{\S#2#1#3}
\crefformat{subsubsection}{\S#2#1#3}
\usetikzlibrary{
  shapes.geometric,
  arrows,
  external,
  pgfplots.groupplots,
  matrix
}
\pgfplotsset{compat=1.9}
\acmPrice{}
\usepackage{mathtools,}

\DeclareMathAlphabet{\mathcal}{OMS}{cmsy}{m}{n}

\usepackage{graphicx}
\usepackage{subcaption}
\graphicspath{{graph/}}
\usepackage[ruled,linesnumbered,vlined]{algorithm2e}
\SetAlgoSkip{smallskip}
\SetKwRepeat{Do}{do}{while}
\usepackage{wrapfig}
\usepackage{titlesec}
\titlespacing{\section}{0pt}{1\baselineskip plus 0.1\baselineskip minus 0.1\baselineskip}{0.2\baselineskip}
\titlespacing{\subsection}{0pt}{0.5\baselineskip plus 0.05\baselineskip minus 0.05\baselineskip}{0.2\baselineskip}
\titlespacing{\subsubsection}{0pt}{0.2\baselineskip}{0pt}

\setcopyright{acmlicensed}
\copyrightyear{2018}
\acmYear{2018}
\acmDOI{XXXXXXX.XXXXXXX}
\acmConference[]{}{}{}
\acmISBN{978-1-4503-XXXX-X/18/06}

\DeclareGraphicsExtensions{%
    .png,.PNG,%
    .pdf,.PDF,%
    .jpg,.mps,.jpeg,.jbig2,.jb2,.JPG,.JPEG,.JBIG2,.JB2}
\usepackage{xparse}
\newcommand{\bnm}{\begin{newmath}}
\newcommand{\enm}{\end{newmath}}

\newcommand{\bea}{\begin{eqnarray*}}%
\newcommand{\eea}{\end{eqnarray*}}%

\newcommand{\bne}{\begin{newequation}}
\newcommand{\ene}{\end{newequation}}

\newcommand{\bal}{\begin{newalign}}
\newcommand{\eal}{\end{newalign}}

\newenvironment{newalign}{\begin{align}%
\setlength{\abovedisplayskip}{4pt}%
\setlength{\belowdisplayskip}{4pt}%
\setlength{\abovedisplayshortskip}{6pt}%
\setlength{\belowdisplayshortskip}{6pt} }{\end{align}}

\newenvironment{newmath}{\begin{displaymath}%
\setlength{\abovedisplayskip}{4pt}%
\setlength{\belowdisplayskip}{4pt}%
\setlength{\abovedisplayshortskip}{6pt}%
\setlength{\belowdisplayshortskip}{6pt} }{\end{displaymath}}

\newenvironment{newequation}{\begin{equation}%
\setlength{\abovedisplayskip}{4pt}%
\setlength{\belowdisplayskip}{4pt}%
\setlength{\abovedisplayshortskip}{6pt}%
\setlength{\belowdisplayshortskip}{6pt} }{\end{equation}}

\newcounter{ctr}

\newcounter{mytable}
\def\mytable{\begin{centering}\refstepcounter{mytable}}
\def\endmytable{\end{centering}}

\newcounter{myfig}
\def\myfig{\begin{centering}\refstepcounter{myfig}}
\def\endmyfig{\end{centering}}

\newlength{\saveparindent}
\newlength{\saveparskip}
\newcommand{\E}{{\rm I\kern-.3em E}}

\renewcommand{\eqref}[1]{\mbox{Equation~(\ref{#1})}}

\def \part {part}

\renewcommand{\paragraph}[1]{\vspace*{6pt}\noindent\textbf{#1}\;}

\def \blackslug{\hbox{\hskip 1pt \vrule width 4pt height 8pt
    depth 1.5pt \hskip 1pt}}
\def \qed{\quad\blackslug\lower 8.5pt\null\par}

\newcounter{mynote}[section]

\newcommand\ignore[1]{}

\newcounter{rcnote}[section]

\newcounter{mrnote}[section]

\newcounter{fknote}[section]

\newcounter{anote}[section]

\DeclareMathSymbol{\mlq}{\mathord}{operators}{``}
\DeclareMathSymbol{\mrq}{\mathord}{operators}{`'}

\newcommand{\rhf}[2]{R_{f, \gamma}}

\DeclareDocumentCommand{\edist}{o o}{
  \ensuremath{
    \IfNoValueTF{#1}{{d}}{{\sf d}(#1,#2)}
  }
}

\newcommand{\olrk}[1]{\ifx\nursymbol#1\else\!\!\mskip4.5mu plus 0.5mu\left(\mskip0.5mu plus0.5mu #1\mskip1.5mu plus0.5mu \right)\fi}

\NewDocumentCommand{\indseq}{ O{1} O{r} }{{#1}\ldots {#2}}

\title{{EPARA}: Parallelizing Categorized AI Inference in Edge Clouds}

\author{Yubo Wang$^{1}$, Yubo Cui$^{1}$, Tuo Shi$^{2}$, Danyang Li$^{3}$, Wenxin Li$^{1\ast}$, Lide Suo$^{1}$, Tao Wang$^{1}$, Xin Xie$^{1}$}

\authornote{Corresponding author.}

\affiliation{
  \institution{Tianjin Key Laboratory of Advanced Networking, Tianjin University$^{1}$}
  \country{}
}
\affiliation{
  \institution{Department of Computer Science, Aalto University$^{2}$}
  \country{}
}
\affiliation{
  \institution{College of Computer Science, Nankai University$^{3}$}
  \country{}
}

\renewcommand{\shortauthors}{Yubo Wang et al.}

\renewcommand{\shorttitle}{{EPARA}: Parallelizing Categorized AI Inference in Edge Clouds}
\begin{document}

\begin{abstract}
\noindent With the increasing adoption of AI applications such as large language models and computer vision AI, the computational demands on AI inference systems are continuously rising, making the enhancement of task processing capacity using existing hardware a primary objective in edge clouds. 
We propose \EPARA, an end-to-end AI parallel inference framework in edge, aimed at enhancing the edge AI serving capability. 
Our key idea is to categorize tasks based on their sensitivity to latency/frequency and requirement for GPU resources, thereby achieving both request-level and service-level task-resource allocation. 
\EPARA consists of three core components: 1) a task-categorized parallelism allocator that decides the parallel mode of each task, 2) a distributed request handler that performs the calculation for the specific request, and 3) a state-aware scheduler that periodically updates service placement in edge clouds. 
We implement a \EPARA prototype and conduct a case study on the \EPARA operation for LLMs and segmentation tasks. 
Evaluation through testbed experiments involving edge servers, embedded devices, and microcomputers shows that \EPARA achieves up to 2.1$\times$ higher goodput in production workloads compared to prior frameworks, while adapting to various edge AI inference tasks.
\end{abstract}




\maketitle


\section{Introduction}

\noindent With the rapid development of AI applications, primarily including generative~\cite{grattafioriLlama3Herd2024,linOpenSoraPlanOpenSource2024} and analytical~\cite{chengMaskedAttentionMaskTransformer2022,khanamYOLOv11OverviewKey2024} models, the resource demands for AI services continue to escalate. This trend motivates AI service providers to enhance inference system capabilities using existing hardware. The key objective lies in fulfilling service level objectives (SLOs) by implementing \textit{task-resource allocation}~\cite{zhangSHEPHERDServingDNNs2023}.

Existing centralized AI serving schemes primarily fall into three categories: 
1) serving specific AI tasks such as LLMs~\cite{agrawalTamingThroughputLatencyTradeoff2024,linParrotEfficientServing2024,zhongDistServeDisaggregatingPrefill2024}, 
2) serving AI tasks without parallelism considerations~\cite{soiferDeepLearningInference2019,hazelwoodAppliedMachineLearning2018,parkDeepLearningInference2018}, and 
3) serving AI tasks with parallelism adaptability~\cite{shubhaUSHERHolisticInterference2024,choiServingHeterogeneousMachine2022,zhangSHEPHERDServingDNNs2023}.

However, these centralized architectures still encounter service overload due to intensive requests and insufficient resources~\cite{brownLanguageModelsAre2020,deepseek-aiDeepSeekV3TechnicalReport2024}. \textit{Edge computing} with idle resources closer to end users has become an extension option~\cite{wenSurveyIntegratedSensing2024}. Though edge servers and devices exhibit scattered and individually limited resources~\cite{brownArchitectureEdgeNetworking2024}, such fragmented but cumulatively abundant resources can collaborate to meet SLOs requirements~\cite{shengyuanyeJupiterFastResourceEfficient2024}. 
Additionally, the edge system obtains more specific request patterns and is compatible with customized services.
Coupled with rising privacy concerns~\cite{yeGalaxyResourceEfficientCollaborative2024a}, \textit{AI is moving closer to the edge}~\cite{rajagopalanRunningDistilledDeepSeek2025}.
Unfortunately, compared to datacenters, edge servers and devices have limited resources and cannot deploy multiple AI services separately to fully meet SLOs~\cite{kongAccuMOAccuracyCentricMultitask2023}. Meanwhile, edge clouds and services exhibit decentralization, heterogeneity, and eruption~\cite{dongLinkLab20Multitenant2023}.
Therefore, conventional centralized AI serving systems cannot be directly applied to edge clouds.

Existing edge-based serving systems mainly encompass three directions: 
1) AI inference on multiple wireless or resource-limited edge devices~\cite{yeGalaxyResourceEfficientCollaborative2024a,shengyuanyeJupiterFastResourceEfficient2024,zhangEdgeShardEfficientLLM2024}, which encounters challenges in resource sharing and heterogeneous hardware, 
2) offloading and coordination between edge servers and devices~\cite{huangCLIOEnablingAutomatic2020,wangDependentTaskOffloading2022}, typically requiring model-specific adaptations, and 
3) inter-edge-cloud scheduling without parallelism~\cite{liMultiHopTaskOffloading2025,farhadiServicePlacementRequest2021,ioannidisAdaptiveCachingNetworks2018}, often formulated as NP-hard problems that hinder real-time serving.

However, these edge-based schemes solely focus on universal non-AI tasks, or apply partial existing datacenter AI inference strategies to edge, and they fail to achieve end-to-end edge AI inference through multi-server and multi-device. Finer-grained task-resource allocation is required to cope with diverse or resource-limited hardware and abrupt or uneven requests in edge, thereby enhancing the serving capability of edge cloud for AI inference tasks.


Based on the above discussion, we try to investigate task-resource allocation from \textit{request-level} and \textit{service-level}, two finer-grained perspectives. 
First, at service-level, similar to datacenter practices, we categorize tasks based on their requirements for GPU resources. This primarily involves coordinating multiple GPUs for each heavy task and packing multiple lightweight tasks to each GPU~\cite{shenNexusGPUCluster2019}.
Second, at request-level, we further categorize tasks into \textit{frequency-sensitive} and \textit{latency-sensitive} types based on their requirements (smoother or quicker) and request patterns.

For example, Fig.~\ref{fig:introduction-frameDP} illustrates a heavy frequency-sensitive video task in resource-limited edge environments, where single GPU fails to meet its frame rate demands. Multi-GPU service-level parallelism introduces additional communication overhead and may still fall short of target fps, but simple request-level round-robin processing can achieve linear scalability in frame rates.
%
Notably, integrating both service-level and request-level allocation necessitates additional design efforts. Implementing such allocation mechanisms in edge further requires a pragmatic end-to-end scheduling system.



{\bf The \EPARAbf approach:} 
Our core innovation lies in applying both request-level and service-level allocation strategies to different tasks categorized by latency/frequency and required resources, coupled with real-time request handling and periodic service placement in edge clouds. 
\EPARA not only maximizes edge system goodput through task-resource allocation, but also guarantees end-to-end edge cloud AI serving through pragmatic scheduling. 

For implementation, \EPARA incorporates adaptive strategies for parameter configuration and establishes hierarchical management protocols. Case studies on LLMs and segmentation models demonstrate \EPARA's framework validity and operational workflow.

We conduct extensive experiments to evaluate \EPARA: 1)~Testbed experiments evaluate \EPARA's performance with embedded devices and microcomputers. 2)~Large-scale co-simulation measures \EPARA's goodput and consumed resources. 3)~Deepdive component-level analyses further prove \EPARA's effect and stability.

Our work makes the following technical contributions:
\begin{itemize}[leftmargin=*]
	\item In edge clouds, we adopt different allocation strategies at request-level and service-level tailored to distinct task categories, while achieving optimized task-resource allocation between AI inference tasks and edge resources.
	\item To perform allocation strategies and enable end-to-end services, we propose a pragmatic edge cloud serving system that includes distributed request handling, submodular service placement, and information synchronization, simultaneously maintaining decoupling, simplicity, and flexibility.
	\item We implement a \EPARA prototype with adaptive strategies and pragmatic management policies. Extensive testbed experiments show that \EPARA achieves up to 2.1$\times$ goodput improvement over InterEdge~\cite{brownArchitectureEdgeNetworking2024} and AlpaServe~\cite{liAlpaServeStatisticalMultiplexing2023}.
\end{itemize}

\section{Motivation}
\noindent In this section, we motivate the requirements for applying task-resource allocation for categorized edge AI inference tasks in both request-level and service-level (shown in Fig.~\ref{fig:introduction-frameDP}), through background~(\cref{section:motivation-background}), desirable situation~(\cref{section:motivation-desirable}), and design space~(\cref{section:motivation-insight}). 
\subsection{Background and Limitations}\label{section:motivation-background}
\noindent When a user \emph{request} specifies a \emph{service} as its target, such combination constitutes a \emph{task}. Our final goal is to serve the most tasks bounded by SLOs with existing hardware. This essentially involves determining an optimal task-resource allocation. We discuss the limitations of three types of related schemes as follows.

\noindent{\bf Centralized AI inference system:}
Existing task-resource allocation strategies for improving centralized AI Inference system goodput mainly includes parallelism, batching, and multi-task~\cite{zhongDistServeDisaggregatingPrefill2024,shenNexusGPUCluster2019,shubhaUSHERHolisticInterference2024}. Parallelism primarily includes Data Parallelism (DP) and Model Parallelism (MP)~\cite{liAlpaServeStatisticalMultiplexing2023}. DP is utilized in training scenarios~\cite{shoeybiMegatronLMTrainingMultiBillion2020}, while MP comprises intra-operator and inter-operator parallelism, also known as Tensor Parallelism (TP) and Pipeline Parallelism (PP)~\cite{zhongDistServeDisaggregatingPrefill2024,liAlpaServeStatisticalMultiplexing2023}, respectively, as illustrated in Fig.~\ref{fig:motivation-parallelism-DPMP}. Batching refers to setting appropriate batch size to enhance GPU efficiency~\cite{choiServingHeterogeneousMachine2022}. Multi-task involves partitioning GPU resources for distinct tasks through frameworks such as Nvidia MPS~\cite{nvidiaNVIDIAMultiProcessService2025} or virtual GPU technologies.

{\bf \emph{Limitations:}}
Centralized serving systems like USHER~\cite{shubhaUSHERHolisticInterference2024}, AlpaServe~\cite{liAlpaServeStatisticalMultiplexing2023}, and gpulet~\cite{choiServingHeterogeneousMachine2022} leverage MP, batching, and multi-task (or replication degree~\cite{shubhaUSHERHolisticInterference2024}/bin-packing~\cite{choiServingHeterogeneousMachine2022}) to meet SLOs. 

\emph{Their allocation strategies are all service-level rather than request-level, because:}
1)~In datacenter, when a single GPU can sufficiently meet the frame rate SLOs of frequency-sensitive tasks, scaling the number of concurrent frames processed can maximize GPU goodput. Even though AlpaServe~\cite{liAlpaServeStatisticalMultiplexing2023} suggests that MP should be deployed for lightweight models to enhance goodput stability, it does not provide enough motivation for request-level allocation. 
2)~It is difficult to maintain a complete flow when the requests reach different GPUs in datacenter, due to the layered encapsulation of APIs and packet disorder caused by complex custom load balancing~\cite{zhangResilientDatacenterLoad2017}.

\emph{Meanwhile, these schemes lack adaptation to edge cloud environments.}
1)~Single edge node often has limited resources and cannot solely handle multiple or abrupt requests with stringent SLO, for which evenly distributing requests and utilizing idle resources at request-level are required.
2)~Edge clouds cannot easily gather information from each other, rendering centralized datacenter schedulers ineffective.
3)~The heterogeneity of hardware and services in edge increases the complexity of customized datacenter policies such as topology optimization~\cite{miaoSpecInferAcceleratingLarge2024} and graph merging~\cite{shubhaUSHERHolisticInterference2024,chenMAGISMemoryOptimization2024}, making their deployment in edge clouds unrealistic.
\begin{figure}[!t]
    \vspace{-0.15cm} 
	\centering
	\abovecaptionskip=-2pt 
	\includegraphics[width=1\columnwidth,trim=0cm 0.2cm 0cm 0cm, clip]{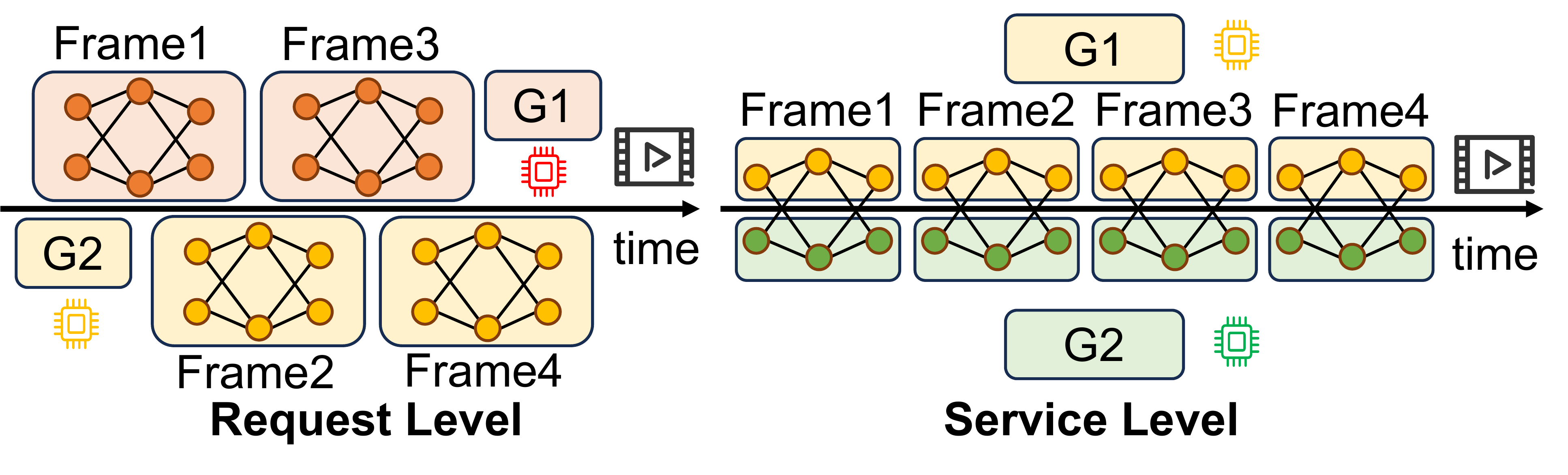}
	\caption{Request-level and service-level allocation.}\label{fig:introduction-frameDP}
	\vspace{2pt} 
\end{figure}
\begin{figure}[!t]
    \centering
    \begin{subfigure}[thbp]{0.49\columnwidth}
        \centering
        \includegraphics[width=1\textwidth,trim=0cm 0cm 0cm 0cm, clip]{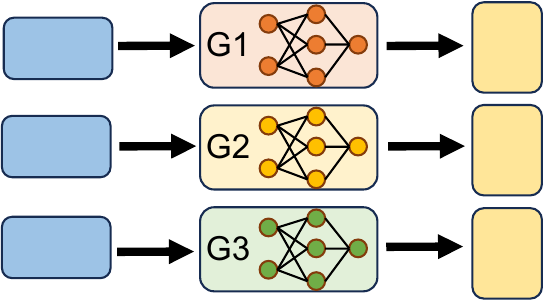}
        \captionsetup{font=footnotesize, justification=centering}
        \vspace{-0.49cm} 
        \caption{Data parallelism for training.}\label{fig:motivation-parallelism-DP}
    \end{subfigure}
    ~~
    \begin{subfigure}[thbp]{0.49\columnwidth}
        \centering
		\includegraphics[width=0.96\textwidth,trim=0.1cm 0cm 0cm 0.3cm, clip]{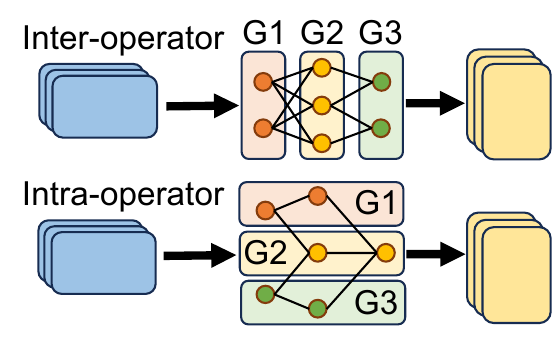}
        \captionsetup{font=footnotesize,justification=centering}
        \vspace{-0.19cm} 
        \caption{Two kinds of model parallelism.}
        \label{fig:motivation-parallelism-MP}
    \end{subfigure}
    \caption{Examples of service-level AI parallelism.}\label{fig:motivation-parallelism-DPMP}
    \vspace{9pt} 
\end{figure}
\begin{figure*}[thbp]
    \centering
    \begin{subfigure}[b]{0.162\textwidth}
        \centering
        \includegraphics[width=\textwidth,trim=0.57cm 17.5cm 10.8cm 1.5cm, clip]{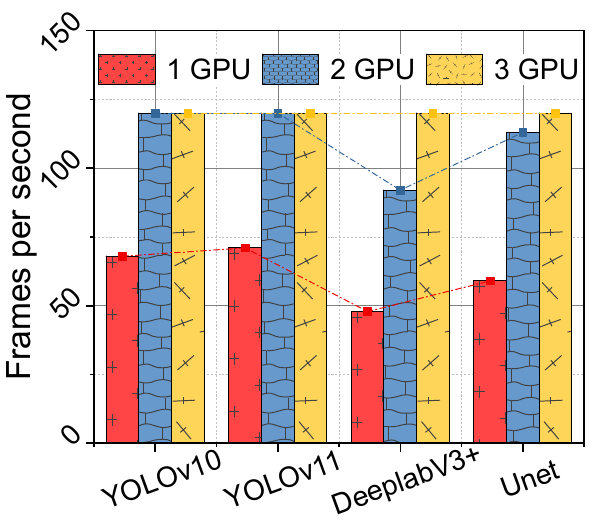}
        \captionsetup{font=footnotesize, justification=centering}
        \vspace{-0.6cm}
        \caption{DP for inference.}
        \label{fig:motivation-DP}
    \end{subfigure}
    ~~
    \begin{subfigure}[b]{0.162\textwidth}
        \centering
        \includegraphics[width=\textwidth,trim=0.47cm 17.5cm 10.9cm 1.5cm, clip]{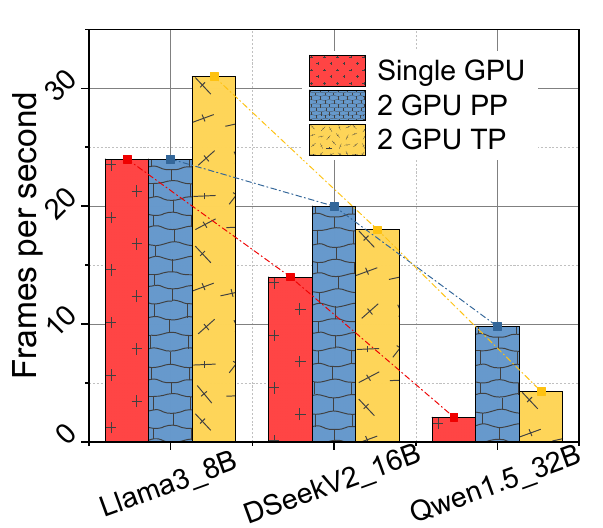}
        \captionsetup{font=footnotesize, justification=centering}
        \vspace{-0.6cm}
        \caption{MP for inference.}
        \label{fig:motivation-PPTP}
    \end{subfigure}
    ~~
    \begin{subfigure}[b]{0.162\textwidth}
        \centering
        \includegraphics[width=\textwidth,trim=0.47cm 17.5cm 10.9cm 1.5cm, clip]{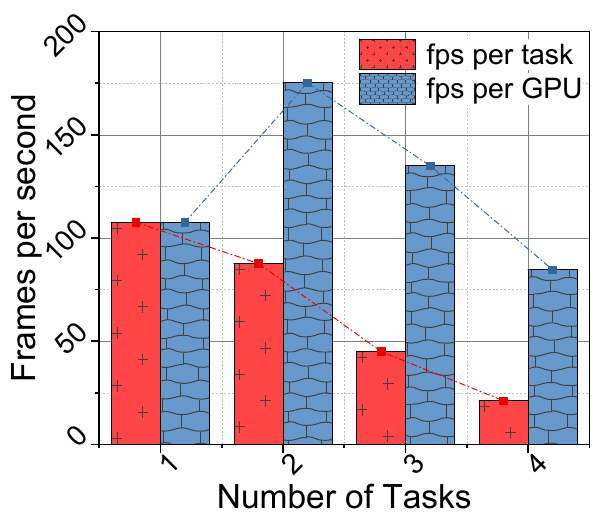}
        \captionsetup{font=footnotesize, justification=centering}
        \vspace{-0.6cm}
        \caption{Multi task in one GPU.}
        \label{fig:motivation-task}
    \end{subfigure}
    ~~
    \begin{subfigure}[b]{0.162\textwidth}
        \centering
        \includegraphics[width=\textwidth,trim=0cm 0.05cm 4cm 2cm, clip]{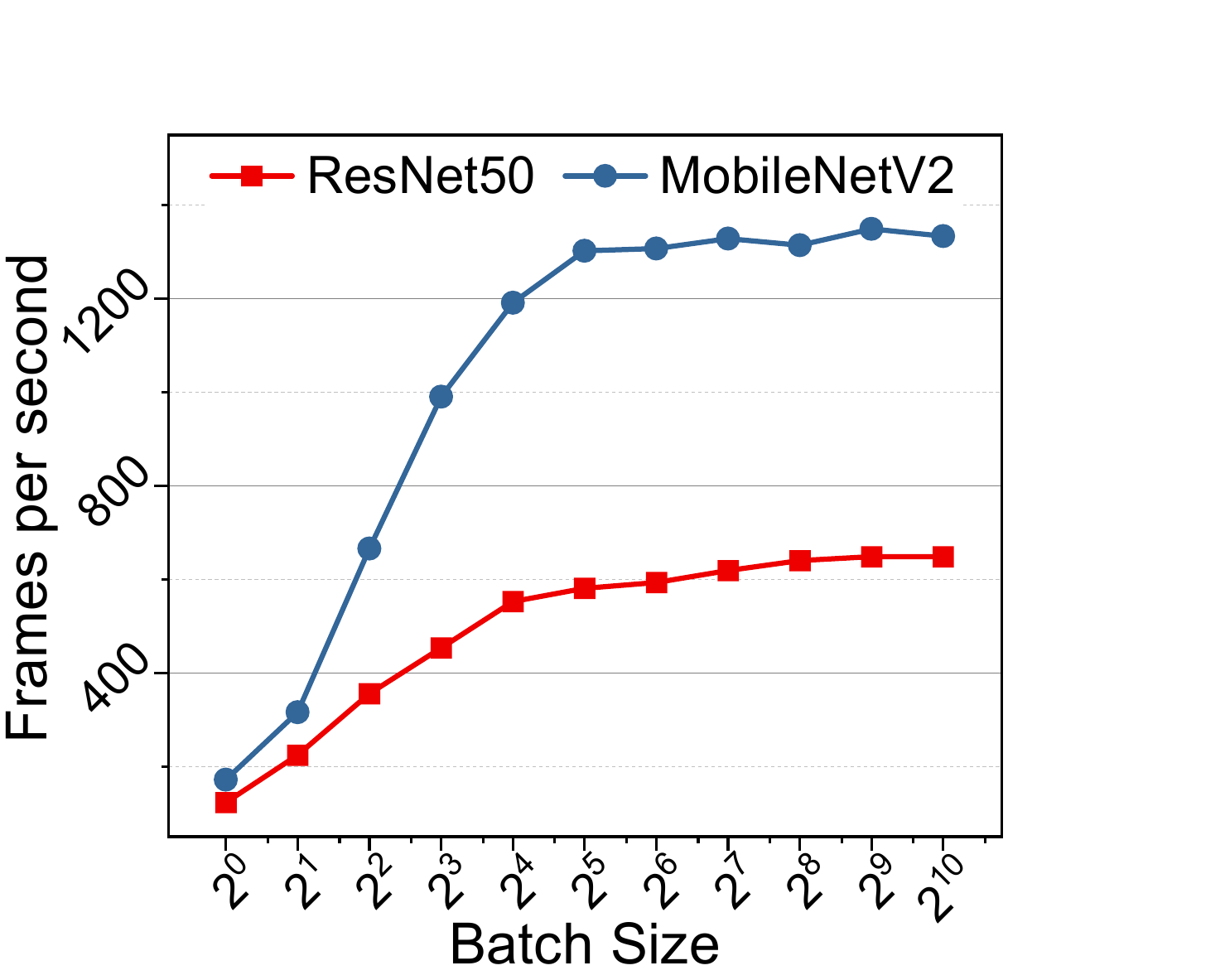}
        \captionsetup{font=footnotesize, justification=centering}
        \vspace{-0.6cm}
        \caption{Batch processing.}
        \label{fig:motivation-batch}
    \end{subfigure}
    ~~
    \begin{subfigure}[b]{0.162\textwidth}
        \centering
        \includegraphics[width=\textwidth,trim=0cm 0cm 4cm 2cm, clip]{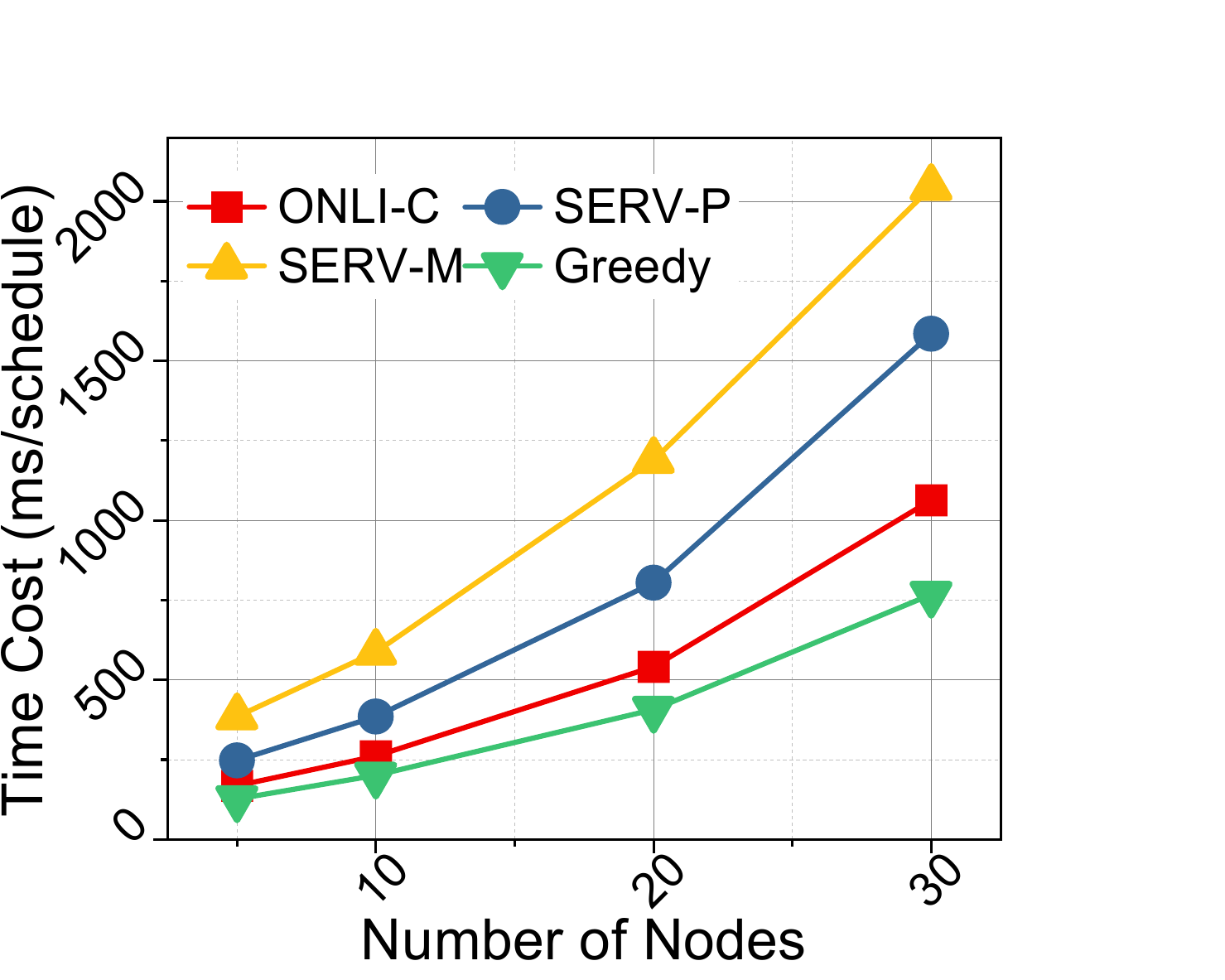}
        \captionsetup{font=footnotesize, justification=centering}
        \vspace{-0.6cm}
        \caption{Centralized schedule.}
        \label{fig:motivation-service_placement}
    \end{subfigure}
    ~~
    \begin{subfigure}[b]{0.162\textwidth}
        \centering
        \includegraphics[width=\textwidth,trim=0cm 0cm 0.35cm 0.7cm, clip]{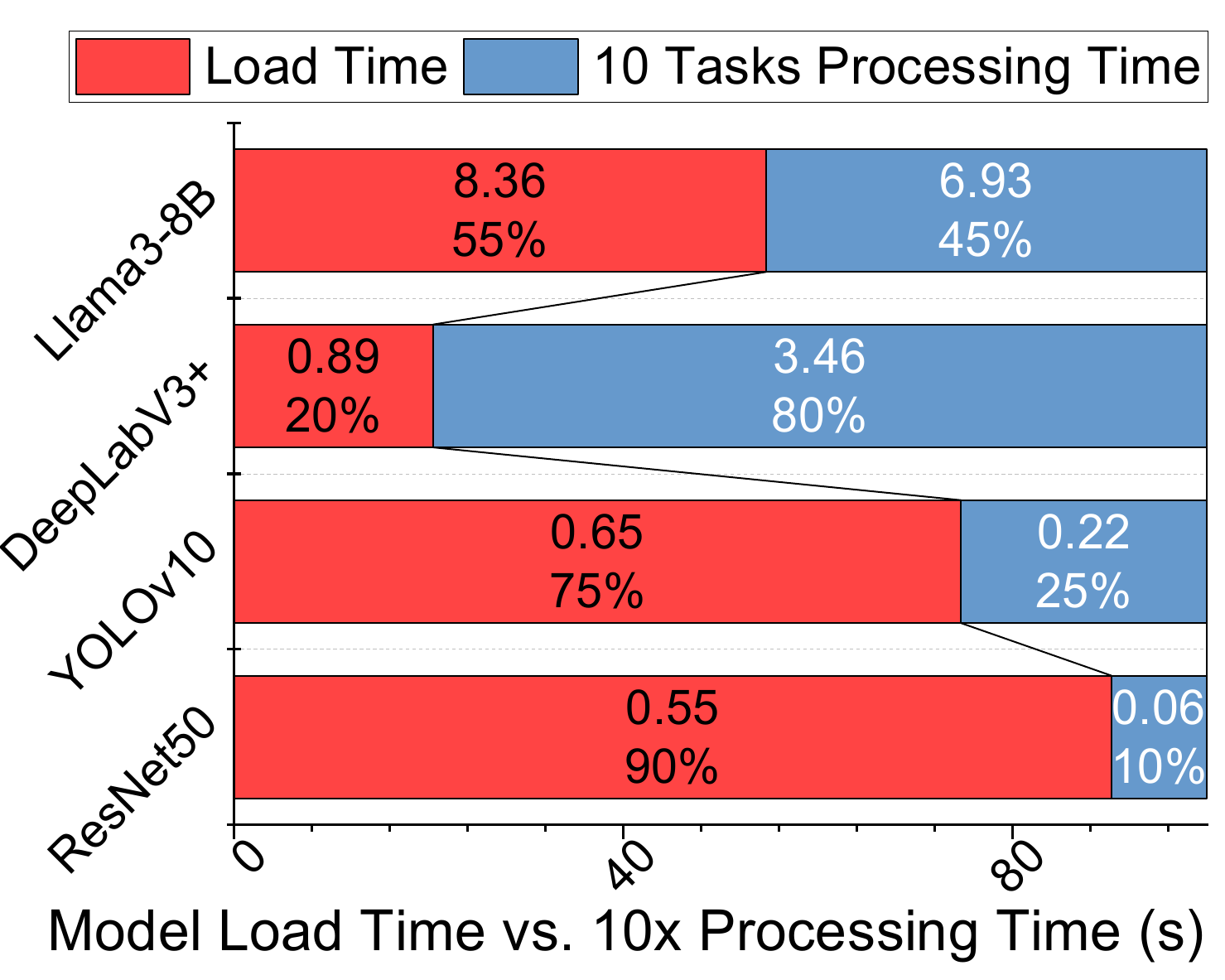}
        \captionsetup{font=footnotesize, justification=centering}
        \vspace{-0.6cm}
        \caption{Load models to GPU.}
        \label{fig:motivation-model_load}
    \end{subfigure}
    \vspace{-0.01cm}
    \caption{\EPARAbf is committed to serving real-time categorized AI parallel inference in edge clouds.}\label{fig:motivation}
    \vspace{-0.1cm}
\end{figure*}

\noindent{\bf Universal edge system:}
In edge computing systems, the edge cloud is generally composed of edge servers and edge devices~\cite{brownArchitectureEdgeNetworking2024}. Typically, edge devices submit requests to edge servers or engage in collaborative device-server processing. Edge servers are often physically distant or without high-bandwidth links, and exhibit resource limitation or hardware heterogeneity~\cite{brownArchitectureEdgeNetworking2024,wenSurveyIntegratedSensing2024,devotoAdaptiveSemanticToken2025}.

{\bf \emph{Limitations:}}
InterEdge~\cite{brownArchitectureEdgeNetworking2024}, SERV-P~\cite{farhadiServicePlacementRequest2021}, and TORS~\cite{liMultiHopTaskOffloading2025} focus on serving universal tasks. 
1)~These approaches are designed for non-AI tasks, and do not employ task-GPU allocation tailored for AI inference. 
2)~Centralized edge schemes~\cite{farhadiServicePlacementRequest2021,liMultiHopTaskOffloading2025,liOnlineCachingNetworks2021} exhibit significant scheduling complexity and latency, which poses challenges for real-time solving, as our experiment results shown in Fig.~\ref{fig:motivation-service_placement}.

\noindent{\bf Edge $+$ AI:}
Current edge AI inference systems typically adopt one of three paradigms~\cite{wenSurveyIntegratedSensing2024}: leveraging inter-device communication to achieve MP inference, enabling device-to-cloud communication for cloud-assisted inference, or focusing on specific services. They aim to facilitate the execution of computationally intensive services within resource-limited edge environments.

{\bf \emph{Limitations:}}
Galaxy~\cite{yeGalaxyResourceEfficientCollaborative2024a}, Detransformer~\cite{weiCommunicationEfficientModelParallelism2024}, and EdgeShar\-ed~\cite{zhangEdgeShardEfficientLLM2024} primarily focus on MP inference where edge devices belong to a unified and centralized edge cloud. CLIO~\cite{huangCLIOEnablingAutomatic2020} focuses on PP inference in single-device and single-cloud systems. REMIX~\cite{jiangFlexibleHighresolutionObject2021}, NeuLens~\cite{houNeuLensSpatialbasedDynamic2022}, and AccuMO~\cite{kongAccuMOAccuracyCentricMultitask2023} focus on specific tasks in edge. 
However, they still exhibit limitations. 1)~They insufficiently consider edge AI systems with multiple servers and devices. 2)~They incompletely implement the service-level strategies of datacenters, lacking consideration for batching~\cite{zhangEdgeShardEfficientLLM2024} or multi-task~\cite{yeGalaxyResourceEfficientCollaborative2024a,weiCommunicationEfficientModelParallelism2024}. 3)~They still regard a frequency task as multiple latency tasks and focus on the processing of each single service like in datacenters, while ignoring the allocation at request-level, so they cannot truly adapt to abrupt requests and scattered resources in edge.

\subsection{What the AI Inference Tasks Really Need in Edge Clouds?}\label{section:motivation-desirable}
\noindent In order to achieve better task-resource allocation, we list three key requirements that AI inference tasks have in edge clouds:

\noindent{\bf Smoother or Quicker? Request-level allocation customized for edge:}
Certain tasks cannot meet frequency SLOs with a single resource-limited edge server~\cite{kongAccuMOAccuracyCentricMultitask2023}, and the requests arrive with abruptness or non-uniformity in edge clouds. They all make SLOs more difficult to fulfill in edge clouds than in datacenters. 
Fortunately, edge is closer to users and can obtain more detailed request patterns. This gives us the opportunity to distinguish the tasks' real needs for frequency and latency, and apply request-level allocation.
As shown in Fig.~\ref{fig:introduction-frameDP} and Fig.~\ref{fig:motivation-DP}, when sending video frames to different GPUs alternately, a simple 2-GPU DP can increase the frame rate from 49 fps to 97 fps. It is workable to fully perform customized request-level allocation for AI inference in edge clouds.

\noindent{\bf One GPU or more? Service-level datacenter-like allocation:}
Although we consider request-level allocation, one frame or a single latency-sensitive request still needs to be satisfied. As shown in Fig.~\ref{fig:motivation-PPTP}-\ref{fig:motivation-batch}, our experiments demonstrate that optimized MP strategies can enhance specific task performance by up to 4.8$\times$ in fps, and superior multi-task and batching strategies can increase GPU throughput by 1.7$\times$ and 6.9$\times$ respectively. These service-level allocation strategies adopted by datacenter should still be considered.

\noindent{\bf Pragmatic end-to-end scheduling:}
In order to achieve the above allocation, an advanced edge cloud scheduling system is essentially needed. This system should enable end-to-end real-time AI serving for edge servers and devices, and involve service placement and request handling to actually meet SLOs~\cite{farhadiServicePlacementRequest2021,liMultiHopTaskOffloading2025,chenDeepLearningEdge2019}.

\begin{figure*}[thbp]
	\centering
	\addtolength{\abovecaptionskip}{-10pt}
	\addtolength{\belowcaptionskip}{-5pt}
	\includegraphics[width=0.95\textwidth]{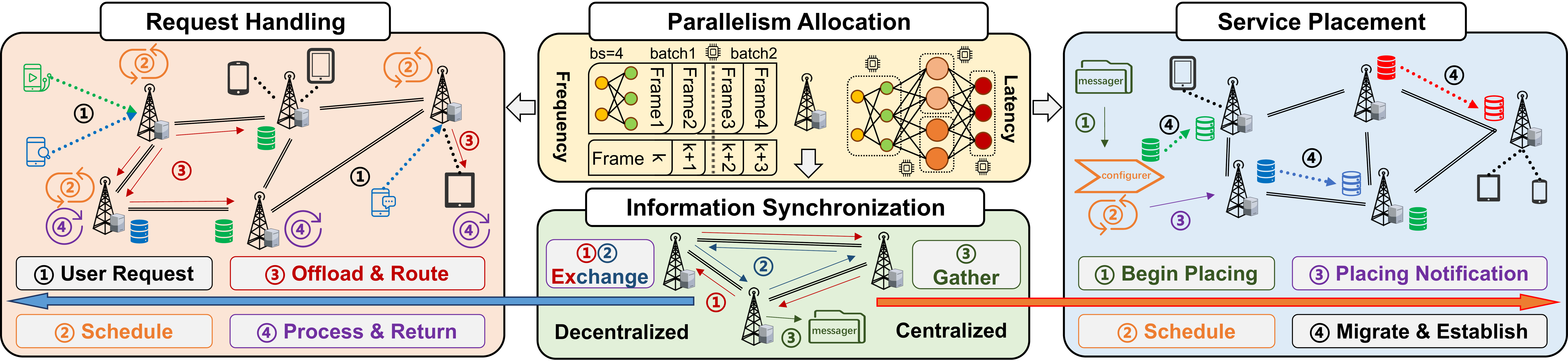}
	\caption{\EPARAbf system overview.}\label{fig:design-overview}
\end{figure*}

\subsection{The Design Space and Challenges}\label{section:motivation-insight}
\noindent{\bf MP+multi task+batching is a good start:}
As illustrated in Fig.~\ref{fig:motivation-PPTP}-\ref{fig:motivation-batch} our findings demonstrate that adjustments in MP, multi-task, and batching all play a pivotal role in enhancing edge cloud system goodput. Based on these service-level observations, we propose that large-scale computational tasks should be distributed across multiple GPUs through systematic parallelism strategies~\cite{liAlpaServeStatisticalMultiplexing2023}, while all AI models/slices should adopt batching and be executed via multi-task on each single GPU~\cite{shubhaUSHERHolisticInterference2024}.

\noindent{\bf DP+multi frame for frequency-sensitive tasks:}
As illustrated in Fig.~\ref{fig:motivation-parallelism-DP}, we evaluate the effectiveness of DP using a 120-fps video input. By distributing frames in a round-robin manner across $k$ GPUs, we achieve a nearly $k$-fold improvement in frame rate. For frequency-sensitive applications with continuous and periodic request arrivals, DP enables efficient and scalable inference to meet frame rate SLOs. Furthermore, we define multi-frame as the scenario where an identical number of frames from homogeneous tasks are grouped within a batch. When the edge cloud serves different tasks at the same time, multi-frame can increase the batch size and better fill GPU resources, which just copes with the request abruptness in edge. The challenge for us lies in strategically combining request-level and service-level approaches to maximize system goodput.

\noindent{\bf Co-design of request handling and service placement:}
To achieve the aforementioned allocation while realizing end-to-end edge cloud AI inference, applying coordinated request handling and service placement is a must. As illustrated in Fig.~\ref{fig:motivation-service_placement}, our implementation of several centralized collaborative placement and handling approaches in edge~\cite{liOnlineCachingNetworks2021,farhadiServicePlacementRequest2021} reveals that scheduling latency exceeds 100 ms when server counts surpass 10 nodes, and escalates beyond 750 ms at 30+ nodes, which is a critical violation of SLOs in most service scenarios~\cite{shubhaUSHERHolisticInterference2024}. Fig.~\ref{fig:motivation-model_load} further demonstrates that AI model placement time constitutes at least 2.5$\times$ the single-task processing duration, necessitating proactive service pre-placement.
Therefore, designing an effective and flexible approach to achieve end-to-end AI inference in edge clouds is still a significant challenge.

\noindent{\bf Putting it all together:}
We propose three core components: task allocation, request handling, and service placement. However, the mere combination of these components fails to achieve end-to-end functionality. This limitation arises because edge servers lack real-time visibility~\cite{autosarAUTOSARAutomotiveOpen2025} into the operational states of neighboring edge servers--even basic service placement scenarios remain unobservable. Consequently, we face the critical challenge of information synchronization, which requires provisioning time-sensitive scheduling data at varying granularities to support these components.

\section{Design}
\noindent{\bf \EPARAbf Overview}:
\EPARA is an edge cloud computing system consisting of multiple edge servers, each populated by numerous edge devices, collectively providing heterogeneous services, as illustrated in Fig.~\ref{fig:design-overview}. 
The primary functions of \EPARA include task-categorized allocation (\cref{section:design-parallelism}), request handling (\cref{section:design-handling}), and service placement (\cref{section:design-placement}). 
The main resources considered in \EPARA include GPU computational resource and GPU VRAM (video RAM) resource. \EPARA reduces the communication resource pressure faced by edge clouds through resource-saving information synchronization and temporal granularity design, and maintain practicality and flexibility (\cref{section:design-decentre}).

\subsection{Task-Categorized Parallelism Allocator}\label{section:design-parallelism}
\noindent In order to achieve higher goodput when serving edge AI inference, finer-grained task-resource allocation is needed.
\EPARA employs different allocation strategies at request/service level for different types of tasks, as illustrated in Fig.~\ref{fig:design-task categorized}.

\noindent{\bf Allocation operators:} 
\EPARA adopts the following five operators.
\begin{itemize}[leftmargin=*]
	\item {\it Batching (BS):} \EPARA employs batching to group multiple tasks of the same service into a single batch for processing, improving GPU utilization and goodput at service-level.
	\item {\it Multi-task (MT):} \EPARA supports multi-task by allowing multiple tasks of different services to be processed simultaneously on the same GPU, further enhancing resource utilization at service-level.
	\item {\it Model Parallelism (MP):} For large models, \EPARA adopts model parallelism to distribute the model across multiple GPUs, enabling efficient processing at service-level (Fig.~\ref{fig:motivation-parallelism-DPMP}).
	\item {\it Multi-frame (MF):} For multi-frame, \EPARA batches an identical number of frames from homogeneous tasks into a unified batch, thereby improving GPU goodput at request-level.
	\item {\it Data Parallelism (DP):} \EPARA employs data parallelism to process frames of a task in different GPUs in a round robin manner, which better meets fps SLOs at request-level. (Fig.~\ref{fig:introduction-frameDP})
	\vspace{2pt}
\end{itemize}

\noindent{\bf Task categories:} 
\EPARA categorizes tasks based on their sensitivity to frequency/latency and their resource consumption (requiring single GPU or multiple GPUs, marked as <1 GPU and >1 GPU). 
We first clarify how latency/frequency and <1/>1 GPU are used as two criteria to categorize tasks. 
1) For frequency-sensitive tasks, such as video AI and LLM human-computer interaction (HCI) systems, requests are continuous and frequent. While meeting latency requirements is a baseline expectation, frame rate becomes the bottleneck for meeting their SLOs~\cite{kongAccuMOAccuracyCentricMultitask2023}. 
In contrast, latency-sensitive tasks, such as LLM dialogue and image AI services, handle non-continuous requests where latency is the sole SLO.
2) For tasks that do not need multi-GPU collaboration, packing allocators like MT~\cite{choiServingHeterogeneousMachine2022,shubhaUSHERHolisticInterference2024,shenNexusGPUCluster2019} and BS suffice. 
In contrast, tasks demanding multi-GPU collaboration require parallelism allocators such as MP. 

\noindent{\bf Performing operators to categories:} 
\EPARA allocator assigns tailored allocation strategies to four distinct categories.

First, for <1 GPU latency-sensitive tasks, such as single-image classification/detection or miniaturized LLM-based text dialogues, \EPARA utilizes BS and MT co-location on a single GPU as service-level strategies. This approach improves GPU utilization and service capacity by concurrently processing heterogeneous tasks.

Second, for >1 GPU latency-sensitive tasks, such as 70B LLM web dialogues and high-resolution semantic segmentation, \EPARA additionally implements TP and PP as two widely used MP operators at service-level. By modeling AI models as weighted directed acyclic graphs (DAGs), TP accelerates parallelizable segments to reduce latency, while PP mitigates VRAM bottlenecks~\cite{liAlpaServeStatisticalMultiplexing2023,yeGalaxyResourceEfficientCollaborative2024a}.

Third, for <1 GPU frequency-sensitive tasks, such as image classification in video streams, \EPARA employs MF, MT, and BS operators. For MF, when a single GPU handles multiple concurrent video streams, \EPARA batches an identical number of frames from one task (or some homogeneous tasks) into a unified batch, thereby increasing batch size and improving overall goodput.

Fourth, for >1 GPU frequency-sensitive tasks, such as semantic segmentation in video analytics or LLM-driven HCI~\cite{kannanSMARTLLMSmartMultiAgent2024,karliAlchemistLLMAidedEndUser2024}, \EPARA additionally adopts DP and MP allocators. Specifically, the full model is replicated across multiple DP GPU groups, with each group handling the task through MP. DP enables alternating input (e.g., frames or interaction requests) assignments across GPU groups to enhance service frequency at request-level, while MP addresses resource constraints and accelerates inference at service-level.

Finally, for applications efficiently executable on CPUs, such as MobileNetV2~\cite{sandlerMobileNetV2InvertedResiduals2018}, inference typically occurs locally on edge devices. Deploying frameworks such as CLIO~\cite{huangCLIOEnablingAutomatic2020}, which coordinates device CPU and server GPU resources, often resembles <1 GPU inference scenarios. This device-server collaborative pattern inherently aligns with the allocation discussed above.

\subsection{Distributed Request Handler}\label{section:design-handling}
\noindent To apply allocation strategies, customized scheduling of user requests is necessary, meaning that \EPARA must allocate computational power in real-time for user requests. Based on periodically synchronized decentralized information, \EPARA employs a simple yet reliable greedy strategy for scheduling user requests, and strives to ensure end-to-end simplicity between users and edge servers, as well as between edge servers themselves~\cite{brownArchitectureEdgeNetworking2024}.
Specifically, as shown in Fig.~\ref{fig:design-request handling}, when a user sends a request $r$ to edge server $n$, $n$ immediately schedules $r$ through the following steps.
\begin{itemize}[leftmargin=*]
	\item If the request $r$ has timed out (SLO violation), processing is halted, and a timeout is returned.
    \item Edge server $n$ checks whether locally placed services are sufficient to handle request $r$; if so, it prioritizes resolving request $r$ locally.
	\item If it is determined that the local resources are insufficient, $n$ will check the task offload count. If the count has reached a certain number, it returns offloading exceed. If not, $n$ analyzes local information to find the optimal offloading strategy, offloads the request $r$, and increments the offload count of $r$.
	\item If based on $n$'s local information, $n$ cannot process request~$r$ in any way, it returns resource insufficiency.
\end{itemize}
\noindent{\bf Offloading strategy:}
When edge server $ n $ determines that request $ r $ requires offloading, it identifies all feasible destination servers $ \dot{n} $ capable of handling $ r $ by calculating an idle goodput $ \tilde{p}_{\dot{n}}^r(\ddot{t}) $, namely, the processing speed of additional $ r $ that $ \dot{n} $ can accommodate. Specifically, edge server $ n $ offloads the request to $ \dot{n} $ with a probability of 
\begin{equation}
\tilde{p}_{\dot{n}}^r(\ddot{t}_n) / \sum\nolimits_{m \in N} \tilde{p}_{m}^r(\ddot{t}_m)\text{, where } \tilde{p}_{n}^r(\ddot{t}_n) = \hat{p}_{n}^r(t_n)-{p}_{n}^r(\ddot{t}_n).
\end{equation}
$ \hat{p} $ and $p$ respectively represents the theoretical and actual goodput. $t_n$ denotes the state information sync delay of server $n$ (sync method is in~\cref{section:design-decentre}), and $ \ddot{t}_n $ denotes the time interval $[-2{t_n}, -{t_n}]$. 
The destination server $ \dot{n} $ is excluded from feasible candidates if the expected computation time of its queued requests exceeds $ t_n + \text{SLO}_r $, preventing offloading with latency SLO violation.

\noindent{\bf Offloading paths:}
When offloaded between edge servers, the request data packets record the path of edge servers they traverse. This ensures that no loop is formed among several edge servers during offloading, preventing waste of communication resources.

\noindent{\bf Request cross-server parallelism:} 
There may be DP or MP scenarios across servers. While these circumstances are not preferred during service placement (\cref{section:design-placement}), complete avoidance is impractical and wasteful. For edge servers containing such nodes, \EPARA considers the task as locally solving, and with a lower priority than solving solely on the local edge server.

\noindent{\bf Edge device participation:} 
In \EPARA, edge devices are allowed to register GPU computational power, which is centrally scheduled by the local edge server. 
For edge devices providing GPU computing capacity (e.g., Jetson Nano~\cite{nvidiaNVIDIAJetsonNano2025}), the \EPARA edge server dynamically allocates dedicated service model weights that can be independently processed by each registered computing-capable edge device, while forwarding service requests to them as needed.
\EPARA considers offloading requests to edge devices within the same edge cloud for inference as locally solving, and with a lower priority than solving through cross-server parallelism.

\begin{figure}[!t]
	\centering
	\vspace{-2pt}
	\addtolength{\abovecaptionskip}{-13pt}
	\addtolength{\belowcaptionskip}{5pt}
	\includegraphics[width=0.95\columnwidth]{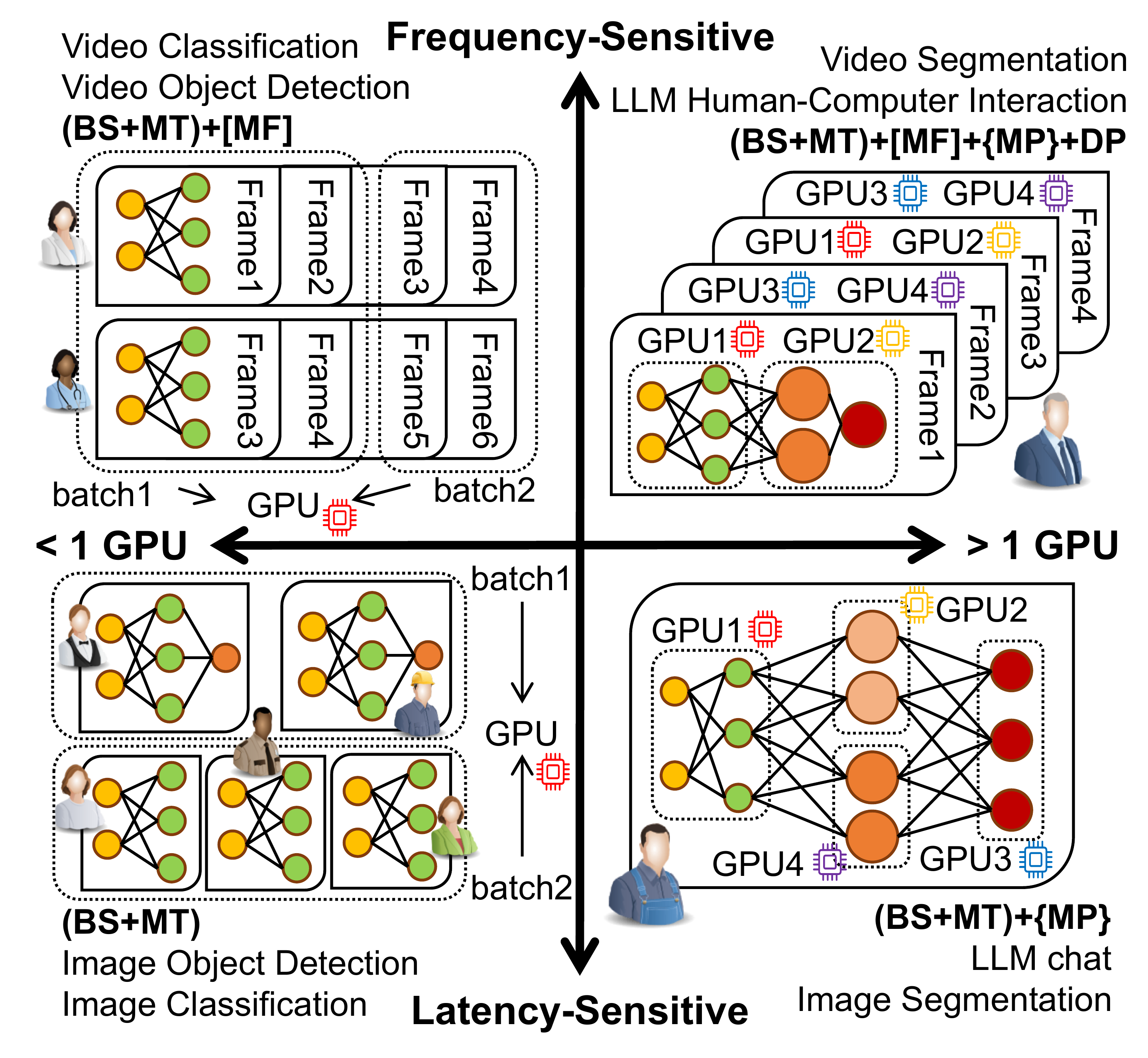}
	\vspace{-1pt}
	\caption{\EPARAbf task-categorized allocation.}\label{fig:design-task categorized}
	\vspace{4pt}
\end{figure}

\subsection{State-aware Submodular Service Placement}\label{section:design-placement}

\noindent \EPARA considers service placement as a critical step in serving edge AI inference due to two main reasons. First, edge typically provides a variety of services, but loading AI model into a GPU is time-consuming especially for LLMs~\cite{ashkboosQuaRotOutlierFree4Bit2024,khareSuperServeFineGrainedInference2025}, which far exceeds the time needed for processing a single task (550ms/60ms for ResNet50~\cite{heDeepResidualLearning2016}), as shown in Fig.~\ref{fig:motivation-model_load}. Second, periodic service placement effectively mitigates computational overhead in request handling. 
A confirmed placement can reduce the action space of request handling decisions, while relaxing the real-time requirements of state monitoring. 
Otherwise we would have to use heuristic algorithms with arbitrarily bad approximation ratio for real-time request handling~\cite{liAlpaServeStatisticalMultiplexing2023}.

During the service placement process, we obtain centralized information of the edge network (detailed in \cref{section:design-decentre}). This includes the resources of edge servers, the request arrivals of a period $T$, and processing situations. For edge service placements, existing works formulate the problem as an integer program, with some employing daunting machine learning methods like DRL\cite{liMultiHopTaskOffloading2025} for resolution, while others use approximation algorithms\cite{farhadiServicePlacementRequest2021} for solutions. This paper employs a submodular function definition approach to solve the problem with a guaranteed lower bound, ensuring it can obtain a relatively optimal solution with polynomial time.

\begin{figure}[!t]
	\centering
	\addtolength{\abovecaptionskip}{-8pt}
	\addtolength{\belowcaptionskip}{5pt}
	\includegraphics[width=1\columnwidth,trim=0cm 0.2cm 0cm 0cm,clip]{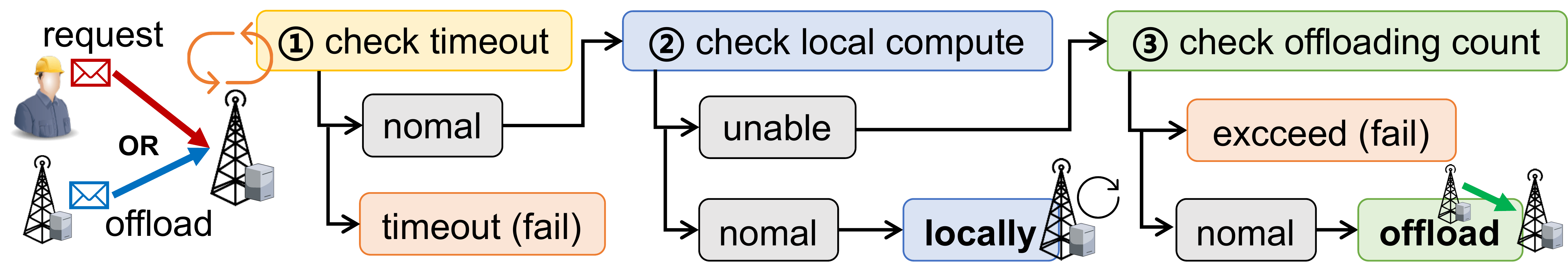}
	\caption{\EPARAbf request handling.}\label{fig:design-request handling}
	\vspace{2pt}
\end{figure}
We denote services in \EPARA as $l \in L$ and servers in \EPARA as $n \in N$. We define user requests as $r_{tln} \in R$, and satisfied user requests $r_{tln}$ by server $m\in N$ as $y_{tlnm}$. Latency-sensitive tasks are seen successfully completed if they meet the latency constraints SLO. While frequency-sensitive tasks need to satisfy the condition which is determined by the proportion of requests fulfilling the SLO-specified frequency threshold. For instance, consider a frequency-sensitive request containing 120 frames with an SLO requirement of 60 fps. If \EPARA achieves only 30 fps due to request saturation, the number of satisfied requests is calculated as $120\times 30/60=60$, reflecting the fraction of frames meeting the SLO target.
\begin{equation}\label{equation:optimal}
	{\max{\varphi(\Theta^T)=\sum\nolimits_{t \in T}\,}{\sum\nolimits_{l \in L}\,}{\sum\nolimits_{n \in N}\,}{\sum\nolimits_{m \in N}\,}y_{tlnm}}
\end{equation}

We define an AI inference service $l \in L$ placed in server $n \in N$ as a placement $x_{ln} \in X$. Our goal is to achieve relatively optimal placement list $\Theta^T\subseteq X^*$, which is, under a definite request handling strategy in Fig.~\ref{fig:design-request handling}, to handle the maximum number of user requests during a time period set $T$.

\vspace{-0.8em}
\begin{algorithm}[!h]
	{\fontsize{8.95}{9}\selectfont
		\SetKwFunction{full}{FullModelPlacement}
		\SetKwProg{Fn}{def}{\string :}{}
		\KwIn{parameters of state-aware submodular placement}
		\KwOut{submodular placement list $\Theta^T$}
	
		\text{Input} $R^T$, $X$, $\mathscr{X}$; $\Theta_0 \leftarrow \emptyset$\;
		S1:  $\Theta_1^T\leftarrow$\full{$R^T$, $\mathscr{X}$, $\Theta_0$}\; 
		S2:  $\Theta_2^T\leftarrow$\full{$R^T$, $X$, $\Theta_1^T$}\; 
		$\dot{X}\leftarrow \{x_{l\epsilon}\mid l \in L \text{, let } \epsilon \notin N \text{ be hypothetical}\}$\;
		$\text{All placements of }\dot{X}\text{ are considered on logical server }\epsilon$\;
		S3:  $\Theta_3^T\leftarrow$\full{$R^T$, $\dot{X}$, $\Theta_2^T$}\;
		
		\text{Output} $\Theta^T\leftarrow \Theta_3^T$ 
	}
	\caption{State-aware service placement~(SSSP)}\label{algorithm:greedy}
\end{algorithm}
\vspace{-1em}

\EPARA 's state-aware submodular service placement is shown in Algorithm~\ref{algorithm:greedy}. During the placement process, we have three placement steps. First~(\emph{S1}), we prioritize partial placement configurations in list $\mathscr{X}$, which applies to scenarios where service providers lease edge GPU resources under time-based pricing models. It can also be used to prioritize the allocation of services requiring intensive inter-GPU parallelism (e.g., DeepSeek-R1~\cite{deepseek-aiDeepSeekV2StrongEconomical2024}) to prevent resource preemption by smaller-scale services. Second~(\emph{S2}), \EPARA places multi-GPU parallel services in the same edge server to reduce communication overhead from DP and MP. Third~(\emph{S3}), we aggregate all GPUs to form a hypothetical server $\epsilon$, which is used to process the placement requiring multi-GPU parallelism across edge servers.

Specifically, once deploying service $l$ on server $n$, we locally maintain the computation time required to process request $l$ as $c_{ln}$, thus we can determine the completion status of all tasks within a given period under a known placement strategy based on the request handling strategy.
Based on the above methods, we always select the currently optimal next full-model placement to obtain a final output placement list, as the pseudo-code shown in Algorithm~\ref{algorithm:fullmodel}. 

\vspace{-1.0em}
\begin{algorithm}[!h]
	{\fontsize{8.95}{9}\selectfont
		\SetKwFunction{full}{FullModelPlacement}
		\SetKwProg{Fn}{def}{\string :}{}
		\KwIn{request handling strategy (\cref{section:design-handling})}

		\Fn{\full{$R$: list, X: set/list, $\Theta_0$: list}}{ 
			$\Tilde{\Theta}_0 \leftarrow \Theta_0$, $k \leftarrow 0$\;
			\Do{$\text{if S1: }\varphi(\Tilde{\Theta}_{k})\geq \varphi(\Tilde{\Theta}_{k-1})\text{, else: }\varphi(\Tilde{\Theta}_{k})> \varphi(\Tilde{\Theta}_{k-1})$}{
			$k \leftarrow k+1$, $\Tilde{\theta}_k \leftarrow \emptyset $\;
			\For{$\text{if \ {\rm typeof}}(X)\text{ is set: }\delta \in X \text{, else: }\delta \in X \setminus \Tilde{\Theta}_{k-1}$}{
				Solve $\varphi(\Tilde{\Theta}_{k-1}+\delta)$ in~\eqref{equation:optimal} with $R$, using handling strategy~\cref{section:design-handling}\;

				\If{$\varphi(\Tilde{\Theta}_{k-1}+\delta) > \varphi(\Tilde{\Theta}_{k-1}+\iota), \forall \iota \in \Tilde{\theta}_{k}$}{
					$\Tilde{\theta}_k \leftarrow \{\delta\}$\;
				}
				\ElseIf{$\varphi(\Tilde{\Theta}_{k-1}+\delta) = \varphi(\Tilde{\Theta}_{k-1}+\iota), \forall \iota \in \Tilde{\theta}_{k}$}{
						$\Tilde{\theta}_k \leftarrow \Tilde{\theta}_k \cup \{\delta\}$\;
				}    

			}
			\If{$\Tilde{\theta}_k \neq \emptyset$}{
				$\Tilde{\Theta}_{k} \leftarrow \Tilde{\Theta}_{k-1} + \iota \text{, let } \iota \in \Tilde{\theta}_{k} \text{ be arbitrary}$\;
			\lElse{break}
				}
			} 
			\textbf{return} $\Theta \leftarrow \Tilde{\Theta}_{k-1}$\;
		}
	}
	\caption{Submodular placement for full models (SPF)}\label{algorithm:fullmodel}
\end{algorithm}
\vspace{-1.2em} 

We define $a_l$ as the MPS computational resource consumption and $b_l$ as the VRAM consumption of a service.
When each server has at most one GPU that cannot fully utilize computational power or VRAM, our algorithm yields a $1/(1+P)$ approximation, where
\vspace{-0.15em}\begin{equation}\label{equation:approximation}
	P = \left\lceil \frac{\max a_{l}}{\min_{a_{l} > 0} a_{l}} \right\rceil + \left\lceil \frac{\max b_l}{\min_{ b_l > 0} b_l} \right\rceil .
\end{equation}

For the proof of the submodular function and the approximation lower bound, please refer to appendix~\ref{section:appendix-submodular}. According to our experiments, \EPARA can get far more AI serving benefits than this lower bound in most application scenarios.

\EPARA supports both offline and online placement strategies. When the number of servers is limited, we recommend adopting offline placement. Specifically, this entails deferring the transmission and loading of AI model until after the complete placement scheduling.
Under large-scale server deployments, we need to allocate computational resources and VRAM online across multiple GPUs when placing each service on a specific server, which falls into the category of online allocation of CPU resources and memory for virtual machines (VMs). \EPARA simply employs an optimized greedy approach~\cite{openstackOpenStackDocsScheduling2025} to deal with this online situation.

\subsection{Synchronization, Granularity and Flexibility}\label{section:design-decentre}
\noindent{\bf Information synchronization:}
To provide decision-critical metadata for the \emph{request handler} (\cref{section:design-handling}) and \emph{placement scheduler} (\cref{section:design-placement}), \EPARA implements an efficient yet lightweight bidirectional information synchronization mechanism~\cite{autosarAUTOSARAutomotiveOpen2025,wanBiCCBilateralCongestion2024,leitaoHyParViewMembershipProtocol2007,wangIntelligentLowOverhead2023}. Specifically, all servers form a ring topology where each edge server periodically transmits its local request arrival/processing status and cached system-wide state information to two adjacent peers through a ring-reduce-like~\cite{gibiansky2017bringing} protocol. Concurrently, each server synchronously receives updated status reports from both neighboring nodes at equivalent intervals. This topology-aware synchronization paradigm ensures minimal bandwidth consumption while maintaining sufficiently up-to-date system state for distributed coordination.

\noindent{\bf Temporal granularity:}
Fig.~\ref{fig:design-overview} shows that \EPARA is a pragmatic edge cloud system featuring three levels of temporal granularity. 
\textit{Request handling} is fine-grained (processed on-demand) and decentralized. 
When an edge server receives an offloading or user request, it immediately handles the request, processes locally or offloads it. 
\textit{Information synchronization} is medium-grained and semi-centra\-lized. 
At regular intervals, edge servers synchronize \textit{request handling} and \textit{service placement} data, and ultimately relay it to the \textit{messager}. 
\textit{Service placement} is coarse-grained and centralized. 
At set intervals, the \textit{configurer} executes a submodular algorithm based on \textit{messager}'s information to determine placement.
This enables \EPARA to adapt to the changing edge environment without affecting other granular events, while maintaining simplicity and practicality.

\noindent{\bf Flexibility of \EPARAbf components:} Except for task-categorized allocator, other modules are detachable from each other in \EPARA. Specifically, replacing \EPARA handler with a KubeEdge based centralized request handler~\cite{kubeedge} will not affect the approximation of \EPARA placement. Also, with parameter-server~\cite{liCommunicationEfficientDistributed2014} based sync method, the handler and placement of \EPARA remain valid.

\noindent{\bf Why this works:} \EPARA 's allocator solves service and GPU planning at service-level and addresses computing power shortages of edge nodes and abrupt requests of the edge environment at request-level. The handler ensures higher effectiveness than simple random offloading while not requiring precise global information. Placement can be solved in polynomial time with lower bound guarantee, while remaining in effect for categorized tasks and different handler strategies. These three main designs effectively address the challenges currently faced by edge cloud AI inference.

\section{Implementation}
\noindent 
We implement a \EPARA prototype and deploy it in a distributed manner across multiple servers. 
\EPARA resides between the network and application layers, introducing no modifications to network protocols or additional programming requirements for users. 
In this section, we provide detailed discussions on adaptive deployment strategies (\cref{section:implement-pretraining}) and edge cloud management (\cref{section:implement-management}). 
We validate \EPARA by deploying LLMs in real-world scenarios, analyzing its operational efficacy and end-to-end workflow for deployment (\cref{section:implement-case_study}).

\subsection{Adaptive Deployment}\label{section:implement-pretraining}
\noindent{\bf MP and BS configuration:}
\EPARA accepts user-specified MP and BS strategy. If MP is unspecified, \EPARA defaults to Deepspeed~\cite{rajbhandariZeROMemoryOptimizations2020} prescribed parallelism strategy. After MP configuration, \EPARA performs offline profiling on either the full task or its MP-sliced variants to determine the optimal BS (ranging from $2^0$ to $2^9$) when BS is not provided. Ultimately, \EPARA allocates memory-partitioned Nvidia MPS~\cite{nvidiaNVIDIAMultiProcessService2025} slices to tasks through coordinated MP+BS scheduling.

\noindent{\bf MT configuration:}
After determining the memory partitioning scheme for MPS, \EPARA performs offline profiling to identify the optimal replication degree for slice parallel execution (ranging from $2^0$ to $2^4$). The value of MT configuration is set to match this optimal replication degree. Although potential inconsistencies between memory and computational resource consumption may lead to suboptimal resource utilization, \EPARA's pricing standard based on GPU resource consumption, which is specifically determined by the MT configuration of MPS slice and service duration/processed requests count, creates intrinsic motivation for users to optimize MPS configurations. This MT strategy additionally prevents malicious users from monopolizing GPU resources through artificial replication degree inflation.

\noindent{\bf DP and MF configuration:}
First, for DP, multiple frames of the same task may be interleaved across different GPU groups. The allocation of GPU groups is determined by the frame rate achievable by a single GPU group and the user's frame rate SLOs for the task.
\begin{equation}\label{equation:inter-task-gpu-group}
    DP\ group\ count= \left\lceil\frac{frame\ rate\ requirement}{frame\ rate\ of\ one\ DP \ group}\right\rceil
    \vspace{2pt}
\end{equation}
Second, for MF, an identical number of frames from homogeneous tasks may be grouped into a single batch. Due to better filling BS with abrupt and uneven requests, MF can improve overall goodput, but also increase latency (from $1/\text{fps}$ to $[inter\ frame \ count]/\text{fps}$). Therefore, \EPARA firstly sets MF to the maximum inter-frame count defined by the task's basic latency requirement, then the inter-request count is calculated based on BS and MF.
\begin{equation}\label{equation:intra-request-batch-size}
	inter\ request\ count = \left\lfloor\frac{batch\ size\ (BS)}{\max (inter\ frame\ count\ [MF])}\right\rfloor
\end{equation}

\noindent{\bf Maximum offloading count:}
Increased offloading instances can improve the success rate of request handling, but concurrently incur additional network transmission overhead, thereby exerting pressure on communication resources between edge servers.
Under circumstances where robust and high-bandwidth networks exist between edge servers, inter-server communication is no longer a bottleneck, and request failure handling can be exclusively delegated to \emph{timeout}. 
In other cases, compared to common WAN~\cite{chuRFC6928Increasing2013,hongAchievingHighUtilization2013} routing, each offloading attempt in \EPARA has a high likelihood of being processed, thereby terminating offloading. 
The maximum offloading count is set to 5 by default in our experiments (\cref{section:evaluation-deepdive-effect}), and can be adjusted according to the application scenario.

\subsection{Management of Edge Clouds}\label{section:implement-management}
\noindent{\bf Edge servers management:}
Each edge server in \EPARA is considered a node that can both offload and resolve requests. The \emph{messager} holds stationary information such as IP and MAC addresses of all edge servers. During joining or exiting processes of edge servers, they transmit their stationary configuration metadata along with join/exit requests to the centralized messager. The messager subsequently updates the device list and broadcasts it to all registered servers. To maintain service continuity, joining or exiting will not take effect until current placement cycle completion.

\noindent{\bf Edge devices management:}
In \EPARA, edge devices often exhibit a certain degree of selfishness~\cite{brownArchitectureEdgeNetworking2024}. It is challenging to keep edge devices online and collaboratively perform AI inference tasks like Galaxy~\cite{yeGalaxyResourceEfficientCollaborative2024a} and Jupiter\cite{shengyuanyeJupiterFastResourceEfficient2024}. Therefore, we propose a compromise approach to maximize the utilization of edge devices GPU resources while ensuring stable service provision by edge servers.
Specifically, edge devices~\cite{nvidiaNVIDIAJetsonNano2025} can register to provide computational power to edge servers at any time, then edge servers rapidly allocate models to the GPUs of edge devices. These AI models need to be solvable with a single GPU without inter-device parallelism, due to the uncertain lifecycle of edge devices.

\subsection{Case Study: LLMs from Chats to Robots}\label{section:implement-case_study}
\noindent \EPARA applies operators to different AI tasks divided by <1/>1 GPU and frequency/latency. To elaborate on how \EPARA works, we present a case study employing LLMs from four distinct categories.

For LLMs, the categorical distinction of frequency and latency is determined by the temporal patterns of user request arrivals rather than specific model architectures. In \textit{latency-sensitive} scenarios exemplified by conversational interfaces (including ChatGPT-style voice-enabled applications~\cite{brownLanguageModelsAre2020}), users demand minimal response latency for input submission, interruption handling, and output text. Conversely, \textit{frequency-sensitive} paradigms manifest in interactive LLM applications such as AI virtual assistants (e.g., Manus, XiaoZhi Chatbot~\cite{xiaoxiaXiaoZhiAIChatbot2025}) and robot interactions (e.g., AIchemist~\cite{karliAlchemistLLMAidedEndUser2024}, Smart-LLM~\cite{kannanSMARTLLMSmartMultiAgent2024}), where users frequently inject additional inputs during LLMs processing phases (including both prefill and decoding phases). As shown in Fig.~\ref{fig:implementation-casestudy-llmframe}, this necessitates service providers to implement DP strategies through multi-replica deployments, enabling instantaneous switching to the most recent decoding-phase outputs when interaction interruptions occur or next action is needed.

\begin{figure}[!t]
	\centering
	\abovecaptionskip=0.5pt
	\includegraphics[width=1\columnwidth]{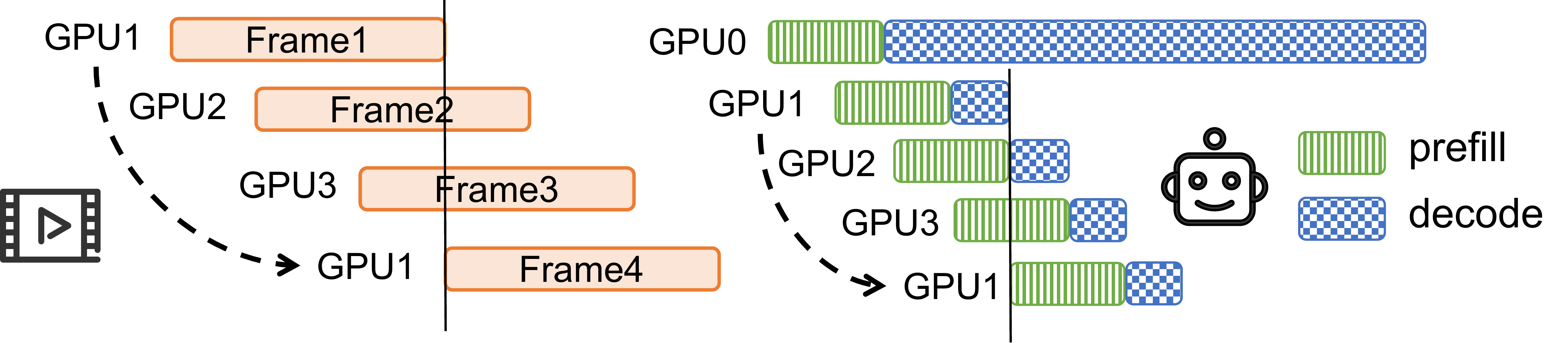}
	\caption{Frequency-sensitive videos and LLMs.}\label{fig:implementation-casestudy-llmframe}
	\vspace{-4pt}
\end{figure}
\begin{figure}[!t]
    \begin{subfigure}[t]{0.24\textwidth}
        \centering
        \includegraphics[width=0.49\textwidth,trim=0.5cm 14.3cm 12.2cm 1.1cm, clip]{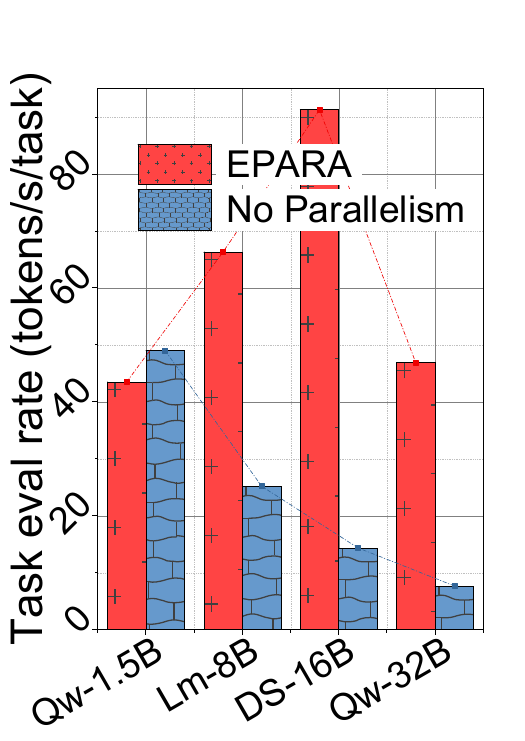}
        \includegraphics[width=0.49\textwidth,trim=0.5cm 14.3cm 12.2cm 1.1cm, clip]{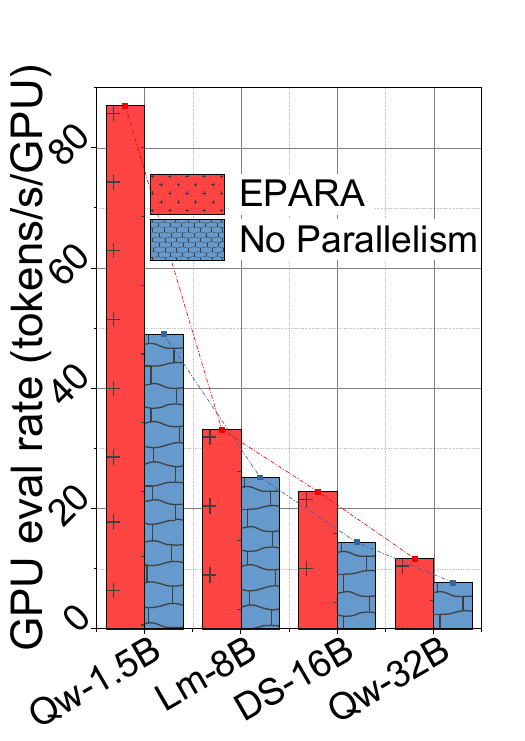}
        \captionsetup{font=footnotesize, justification=centering}
        \vspace{-0.56cm}
        \caption{LLMs HCI.}\label{fig:implementation-LLMHCI}
    \end{subfigure}
    ~~
    \centering
    \begin{subfigure}[t]{0.24\textwidth}
        \centering
        \includegraphics[width=0.49\textwidth,trim=0.5cm 14.3cm 12.2cm 1.1cm, clip]{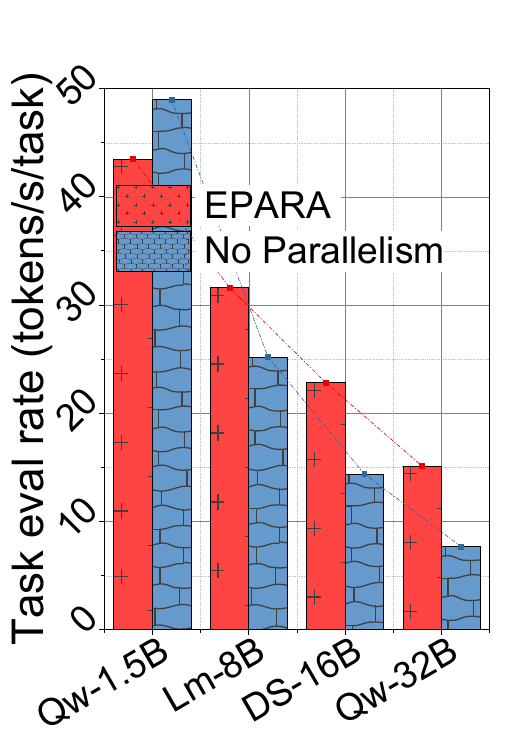}
        \includegraphics[width=0.49\textwidth,trim=0.5cm 14.3cm 12.2cm 1.1cm, clip]{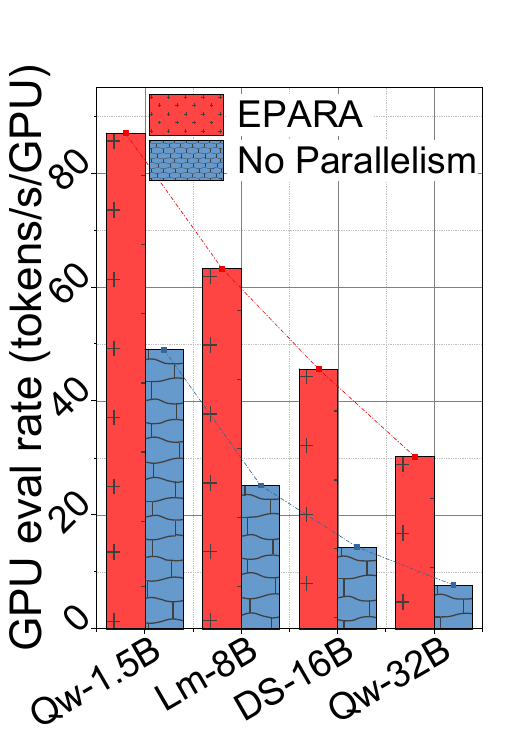}
        \captionsetup{font=footnotesize, justification=centering}
        \vspace{-0.56cm}
        \caption{LLMs chat.}\label{fig:implementation-LLMchat}
    \end{subfigure}
    \vspace{-0.43cm}
    \caption{Case study: LLMs in \EPARAbf.}\label{fig:implementation-LLMcase_study}
    \vspace{10pt}
\end{figure}

\textit{Remarks: DP fundamentally addresses the latency in LLMs prefill phase. Readers may question that if enhancing single-GPU computational capacity could resolve this issue. We posit that achieving equivalent performance (for example 0.1s) via hardware upgrades is a requirement that incurs prohibitively expensive expenditures in edge, particularly for chain-of-thought optimized LLMs~\cite{deepseek-aiDeepSeekV3TechnicalReport2024}.} 

For settings, we categorize mainstream LLMs within \EPARA, as illustrated in the \emph{Text} subsection of Table~\ref{table:evaluation-models}. Chat services and HCI systems are categorized as latency-sensitive and frequency-sensitive respectively. We use four servers each equipped with an Nvidia Tesla P100 GPU to deploy LLMs of four categories. 

We subsequently conduct adaptive deployment according to~\cref{section:implement-pretraining}.
First, we formulate MP strategies and batch size (BS) for these models. For Llama3-8B, DeepSeekV2-16B and Qwen2.5-32B, we adopt BS4+TP2, BS4+TP2, BS4+TP2+PP2 respectively in latency-sensitive scenarios, and adopt BS2, BS4+PP2, BS2+PP2 respectively in frequency-sensitive scenarios.
For Qwen2.5-1.5B, we only adopt BS2 to reach the highest token rate of 87 tokens/s.
Second, \EPARA measures the MT of these tasks, and sets MT to 2 for Qwen2.5-1.5B, remaining MT equal to 1 for other LLMs.
Third, we focus on MF and DP strategies for HCI applications. For Qwen2.5-1.5B with a tolerable inter-frame processing latency of 30 ms, \EPARA adopts MF1. For Llama3-8B, DeepSeekV2-16B, and Qwen2.5-32B, achieving 24 tokens/s, 46 tokens/s, and 24 tokens/s, at BS2, BS2+PP2, and BS2+PP2, respectively, \EPARA employs DP2 to fulfill HCI demands.  

We validate the above discussions in real testbed, with experimental results in Fig.~\ref{fig:implementation-LLMcase_study} demonstrating that \EPARA improves GPU efficiency while satisfying SLO requirements.

Crucially, image and video processing in visual AI correspond perfectly to chat and HCI in LLMs services. 
Another case for segmentation models is provided in \cref{appendix-implement-case_study-segementation}.

\section{Evaluation}\label{section:evaluation}
\noindent We evaluate our \EPARA through a combination of testbed experiments and large-scale simulations and show that:
\begin{itemize}[leftmargin=*]
	\item \EPARA achieves higher performance in practice (\cref{section:evaluation-testbed}).
	\item \EPARA works well in large edge networks (\cref{section:evaluation-simulation}).
	\item \EPARA's design components are effective (\cref{section:evaluation-deepdive}).
\end{itemize}

\begin{figure}[!t]
    \centering
    \begin{subfigure}[thbp]{0.23\textwidth}
        \centering
        \includegraphics[width=0.49\textwidth]{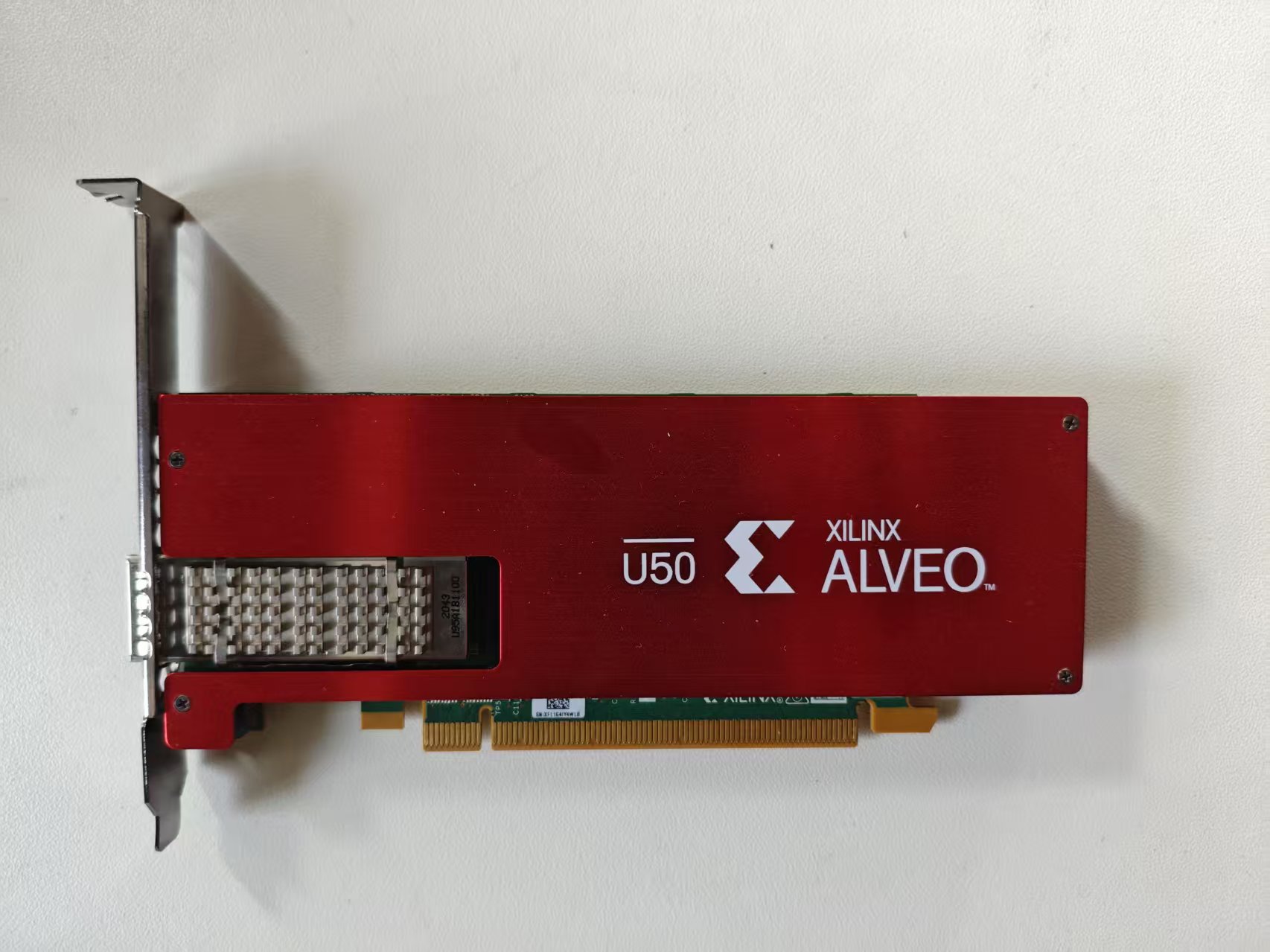}
		\includegraphics[width=0.49\textwidth]{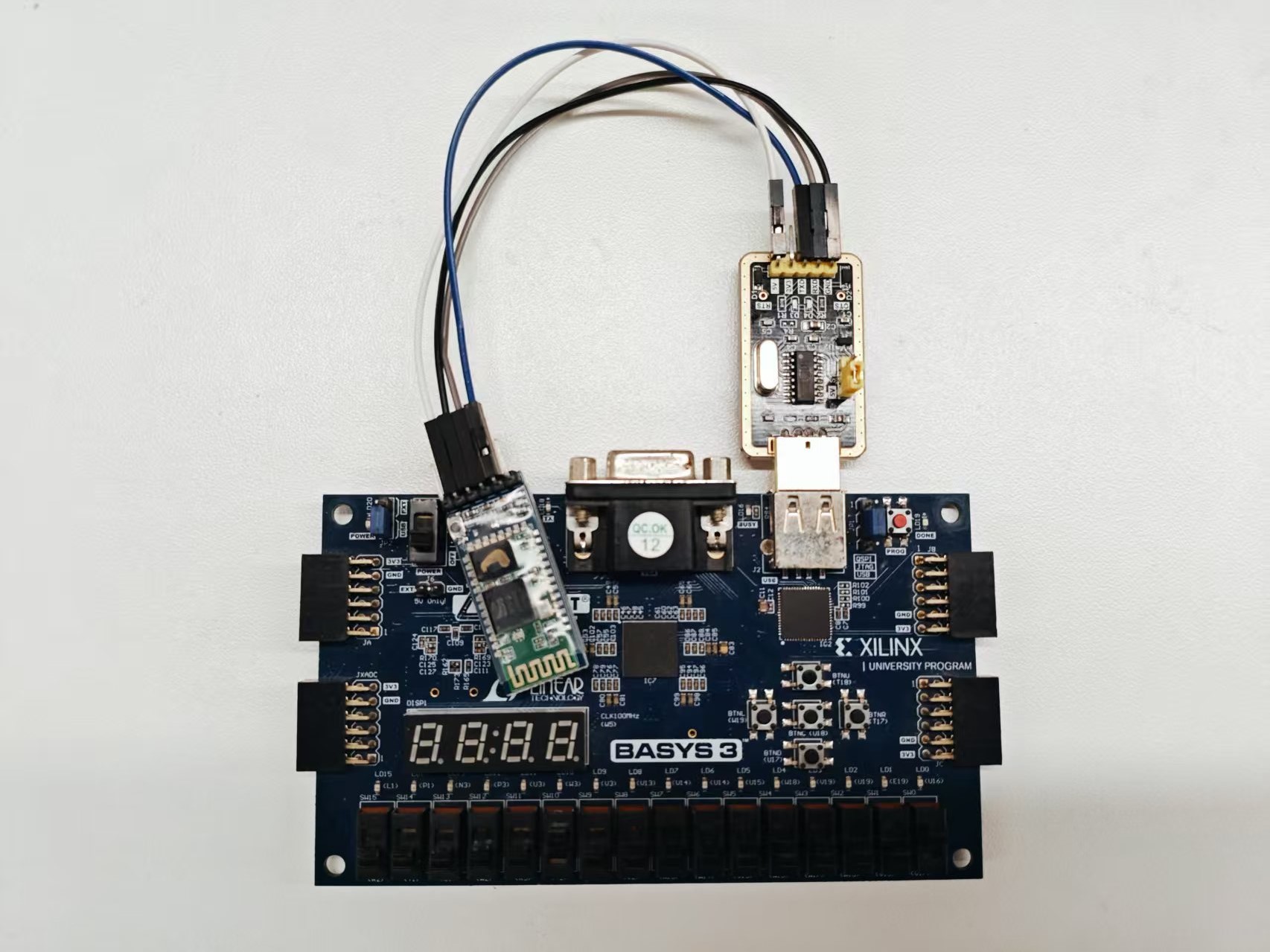}
        \captionsetup{font=footnotesize, justification=centering}
        \caption{Embedded devices.}
        \label{fig:evaluation-setup-device1}
    \end{subfigure}
    ~~
    \centering
    \begin{subfigure}[thbp]{0.23\textwidth}
        \centering
		\includegraphics[width=0.49\textwidth]{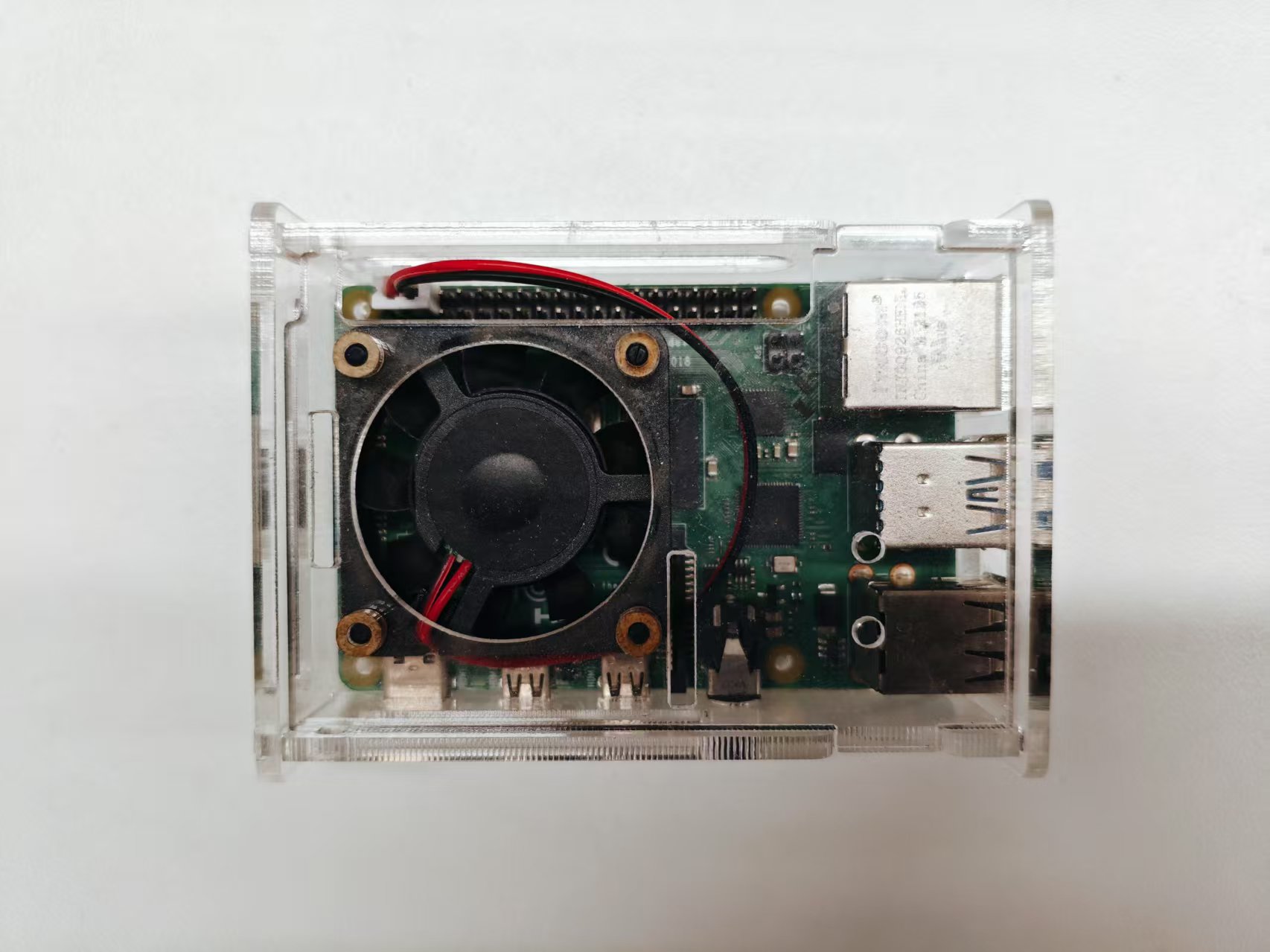}
		\includegraphics[width=0.49\textwidth]{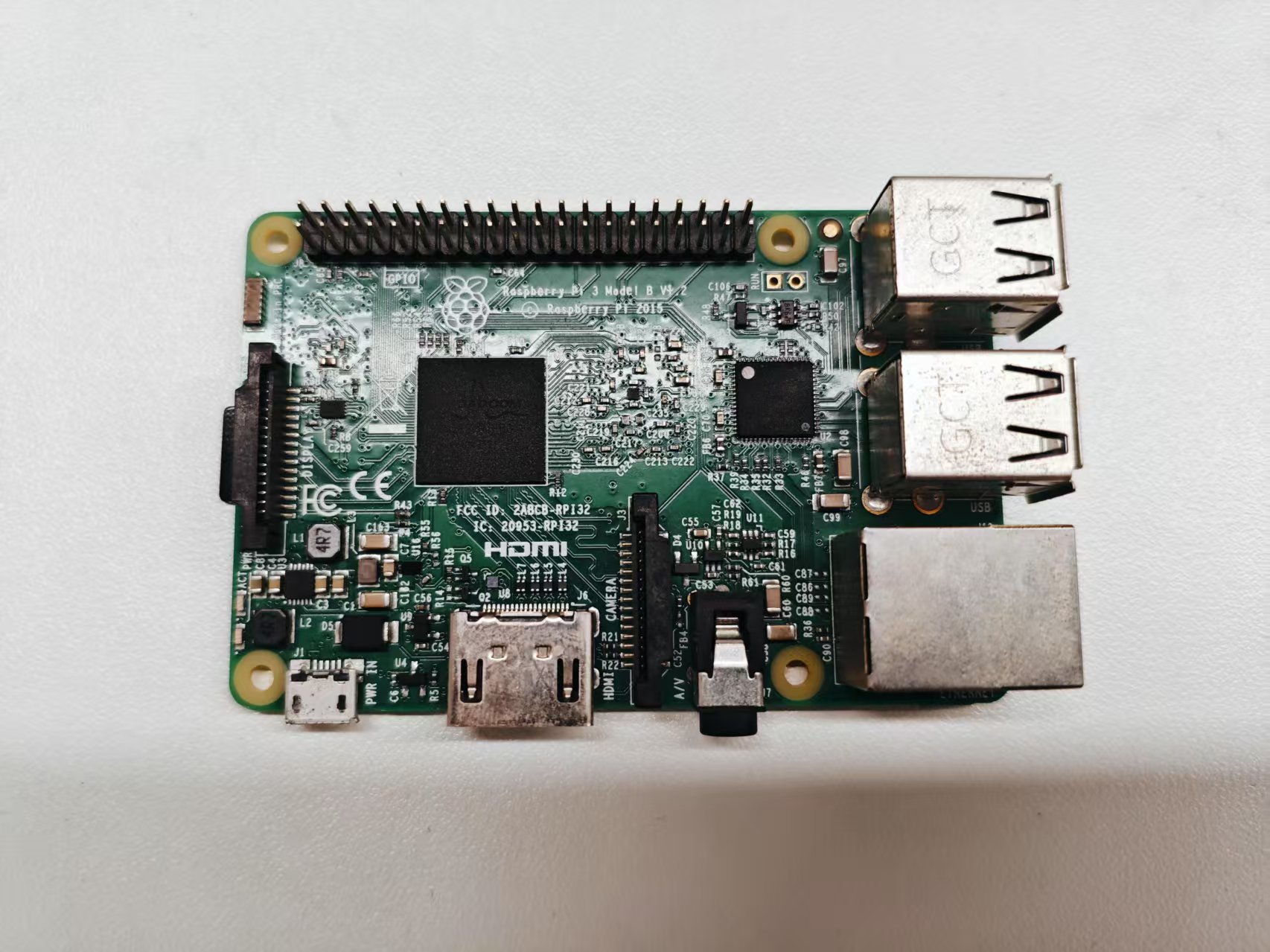}
        \captionsetup{font=footnotesize, justification=centering}
        \caption{Microcomputers.}\label{fig:evaluation-setup-device2}
    \end{subfigure}
    \caption{Edge devices used in testbed.}\label{fig:evaluation-setup-device}
    \vspace{1pt}
\end{figure}


\subsection{Testbed Experiments}\label{section:evaluation-testbed}

\noindent{\bf Setup}:
Our real-world experimental platform includes six Dell PowerEdge R750 servers. 
Each host is equipped with a 28-core Intel(R) Xeon(R) 3.10GHz CPU and 128GB memory.
Two hosts are linked to each other by Mellonax CX6 NIC, and the other four hosts equipped with Nvidia Tesla P100 (16GB) are linked with Edge-core AS4610-54T switch\footnote{\EPARA still performs well without high-performance networks, as shown in \cref{section:evaluation-deepdive}.}. 
For edge devices, as shown in Fig.~\ref{fig:evaluation-setup-device}, we use a Xilinx Alveo U50 and a Xilinx Basys 3 as embedded devices, and use a Raspberry Pi 3 Model B (1GB RAM) and a Raspberry Pi 4 Model B (3GB RAM) as microcomputers.
We provide a summary of the environment settings in Table~\ref{table:appendix-environment_settings} of appendix~\ref{section:appendix-testbed_setting}.

\noindent{\bf Workloads}:
As presented in Table~\ref{table:evaluation-models}, the AI models evaluated in our evaluation encompass LLMs, translation, classification of text/image, detection, and segmentation. As~\cite{shubhaUSHERHolisticInterference2024,zhangSHEPHERDServingDNNs2023,liAlpaServeStatisticalMultiplexing2023}, we use the Microsoft Azure Function Trace 2021~\cite{zhangFasterCheaperServerless2021} for the request rates and the Microsoft Azure LLM Inference Traces 2023~\cite{patelSplitwiseEfficientGenerative2024} for the token length. We assign 100,000 function streams from the function trace to the models mentioned above in a round-robin manner.

\noindent{\bf Comparisions}:
We compare \EPARA , InterEdge~\cite{brownArchitectureEdgeNetworking2024}, AlpaServe~\cite{liAlpaServeStatisticalMultiplexing2023}, Galaxy~\cite{yeGalaxyResourceEfficientCollaborative2024a} and SERV-P~\cite{farhadiServicePlacementRequest2021} in testbed experiments. For InterEdge, we implement a round-robin forwarding offloading strategy, where MP, BS and MT policies align with \EPARA. InterEdge can represent decentralized edge clouds system. For AlpaServe, by default, it refuses to process requests which need offloading or parallelism through multiple distributed edge servers. We use AlpaServe to represent datacenter schemes. For Galaxy, each GPU is treated as an edge device, facilitating collaborative operations that include cross-server GPU coordination. We use Galaxy to represent centralized edge devices AI inference works. For SERV-P, placement and handling across all servers are centrally managed by a single CPU process on one server. SERV-P can represent KubeEdge~\cite{kubeedge} based systems with complex (NP-hard) centralized handling strategies.
\begin{table}[t]
    \centering
    \vspace{-2pt}
    \resizebox{1\linewidth}{!} {
    \begin{tabular}{|c|c||c|c|c|c|}
    \hline 

    \multicolumn{2}{|c||}{} & 
    \multicolumn{2}{c|}{\textbf{< 1 GPU (Tesla P100)}} &  
    \multicolumn{2}{c|}{\textbf{> 1 GPU (Tesla P100)}} \\

    \hline
    \hline

    \multirow{6}{*}{\rotatebox{90}{\textbf{Frequency$\ \ $}}} & 
    \multirow{3}{*}{Vid} & 
    Classify &
    MobileNetV2~\cite{sandlerMobileNetV2InvertedResiduals2018}, ResNet~\cite{heDeepResidualLearning2016}  & 
    Segment & 
    DeeplabV3+~\cite{chenEncoderDecoderAtrousSeparable2018}\\

    \cline{3-5}

    & 
    & 
    Detect &
    YOLOv10~\cite{NEURIPS2024_c34ddd05}, YOLOv11~\cite{khanamYOLOv11OverviewKey2024}  & 
    \multicolumn{2}{c|}{SCTNet~\cite{xuSCTNetSingleBranchCNN2024}, MaskFormer~\cite{NEURIPS2021_950a4152},} \\
    
    \cline{3-4}

    & 
    & 
    Segment &
    Unet~\cite{ronnebergerUNetConvolutionalNetworks2015}  & 
    \multicolumn{2}{c|}{OMG-Seg~\cite{liOMGSegOneModel2024}}\\

    \cline{2-6}

     & 
    \multirow{3}{*}{HCI} & 
    Classify &
    BERT~\cite{devlinBERTPretrainingDeep2019}  & 
    Generate &
    Llama3-8B~\cite{grattafioriLlama3Herd2024}, \\

    \cline{3-5}

    & 
     & 
    Trans &
    GNMT~\cite{wuGooglesNeuralMachine2016}  & 
    \multicolumn{2}{c|}{DeepSeekV2-16B~\cite{deepseek-aiDeepSeekV3TechnicalReport2024}, Qwen2.5-32B~\cite{qwenQwen25TechnicalReport2025},}  \\

    \cline{3-4}

    & 
     & 
    Generate &
    Qwen2.5-1.5B~\cite{qwenQwen25TechnicalReport2025} & 
    \multicolumn{2}{c|}{llama3-70B~\cite{grattafioriLlama3Herd2024}}  \\

    \hline
    \hline

    \multirow{6}{*}{\rotatebox{90}{\textbf{Latency$\ \ $}}} & 
    \multirow{3}{*}{Pic} & 
    Classify &
    MobileNetV2, ResNet & 
    Segment & 
    \\

    \cline{3-5}

    & 
    & 
    Detect &
    YOLOv10, YOLOv11  & 
    \multicolumn{2}{c|}{MaskFormer,} \\
    
    \cline{3-4}

    & 
    & 
    Segment &
    Unet, DeepLabV3+, SCTNet  & 
    \multicolumn{2}{c|}{OMG-Seg}\\

    \cline{2-6}

     & 
    \multirow{3}{*}{Text} & 
    Classify &
    BERT  & 
    Generate &
    Llama3-8B, \\

    \cline{3-5}

    & 
     & 
    Trans &
    \makecell[c]{GNMT\\}  & 
    \multicolumn{2}{c|}{DeepSeekV2-16B, Qwen2.5-32B,}  \\

    \cline{3-4}

    & 
     & 
    Generate &
    Qwen2.5-1.5B  & 
    \multicolumn{2}{c|}{llama3-70B}  \\

    \hline

    \end{tabular}
    }
    \vspace{10pt}
    \caption{Models used in evaluation.}\label{table:evaluation-models}
    \vspace{-15pt}
\end{table}

\subsubsection{Overall Performance}
~\\
We employ two Raspberry Pi models as edge devices and six Dell R750 as edge servers to evaluate \EPARA's AI inference goodput.
 
\begin{wrapfigure}{r}{0.3\textwidth}
    \vspace{-12pt}
	\centering
	\abovecaptionskip=-5pt
	\belowcaptionskip=-8pt
    \leavevmode\hspace*{-0.04\textwidth}
	\includegraphics[width=0.33\textwidth,trim=4.8cm 3.5cm 3.5cm 0cm, clip]{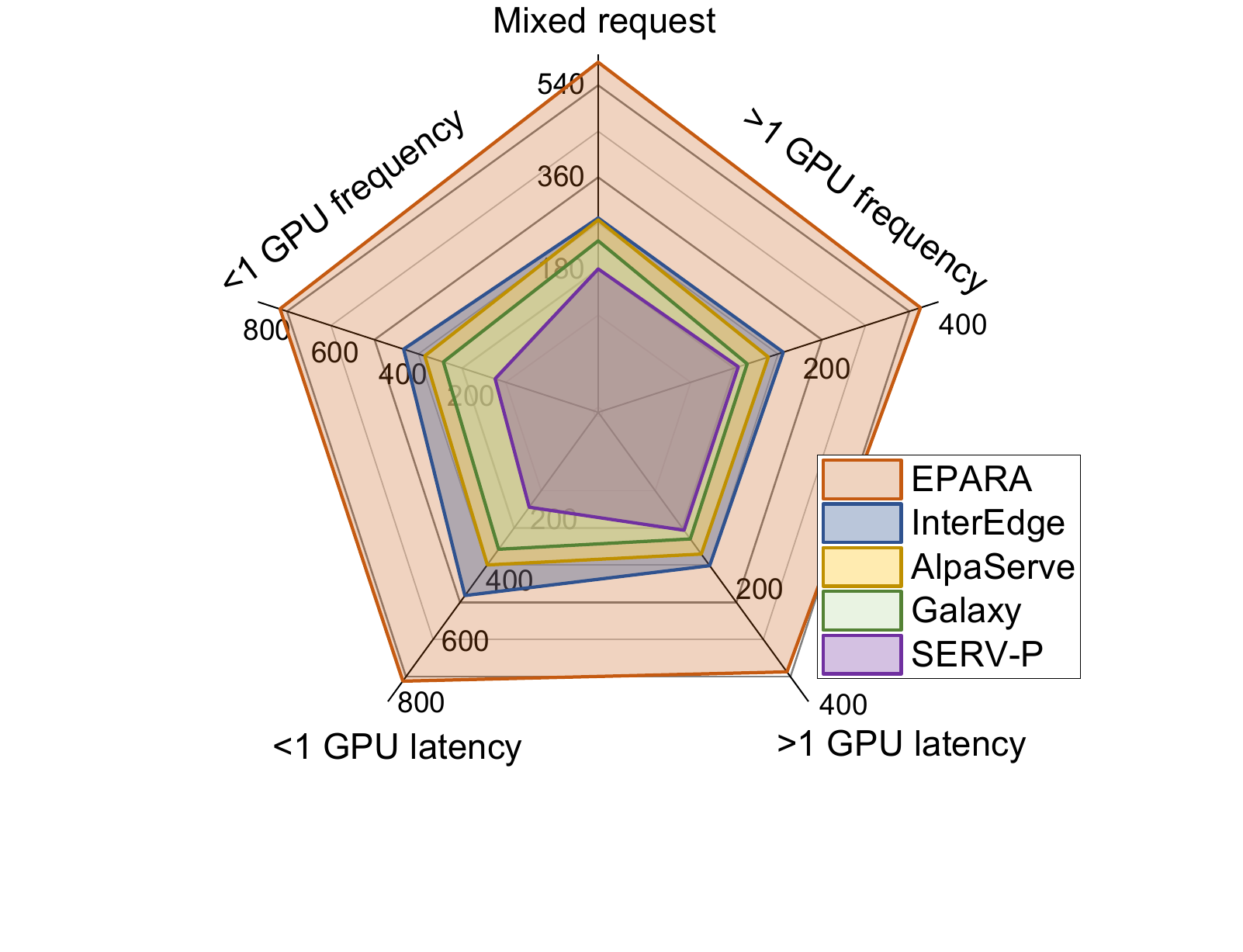}

    \captionsetup{singlelinecheck=false}
    \captionsetup{width=0.34\textwidth}
    \caption{Testbed goodput (reqs/sec).}\label{fig:evaluation-testbed-overall-overall}
    \vspace{-3pt}
\end{wrapfigure}
\begin{figure*}[thbp]
    \centering
    \begin{subfigure}[b]{0.2\textwidth}
        \centering
        \includegraphics[width=\textwidth,trim=0cm 0cm 5.5cm 0.5cm, clip]{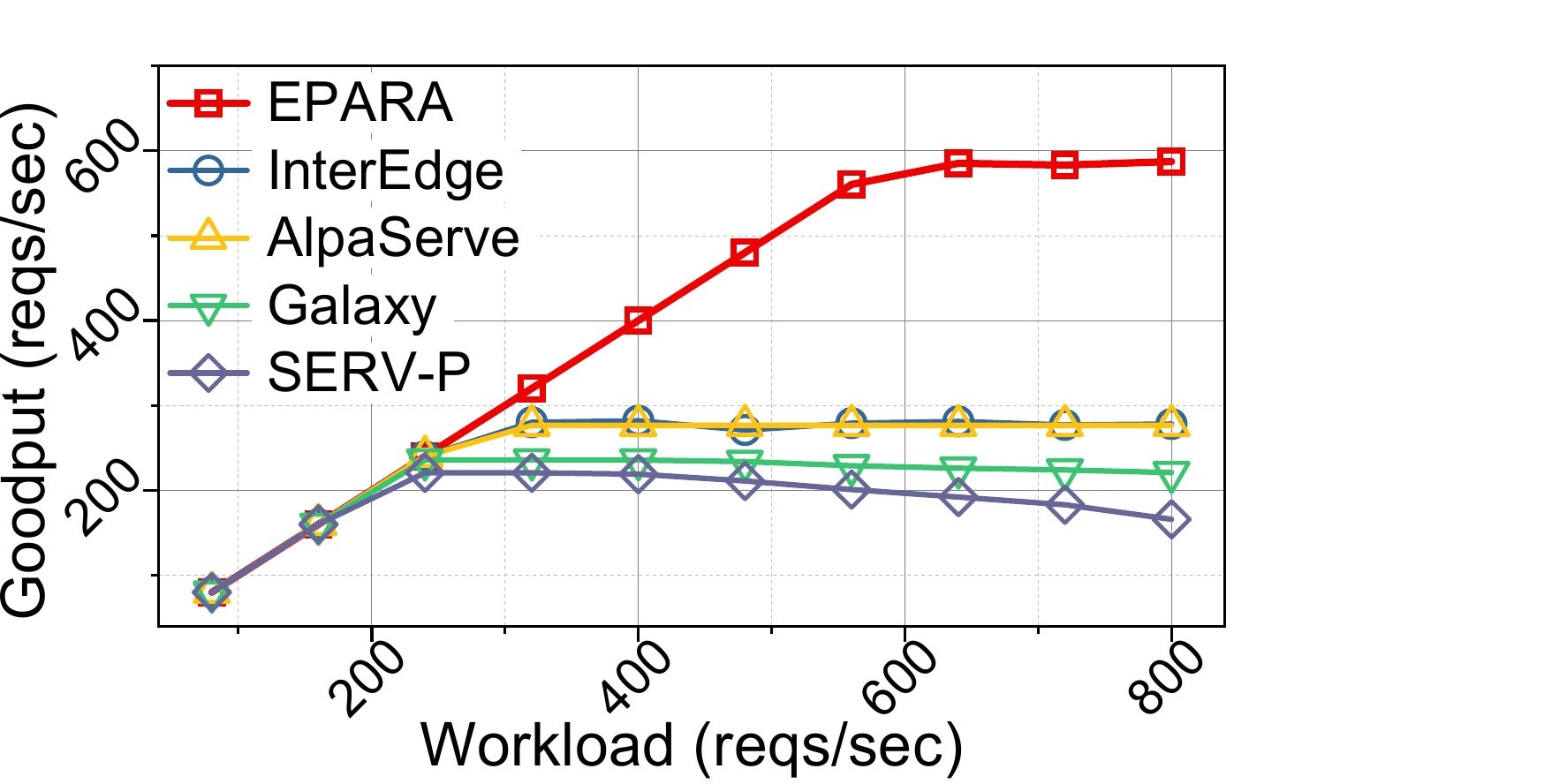}
        \captionsetup{font=footnotesize, justification=centering}
        \vspace{-0.5cm}
        \caption{Mixed request.}
        \label{fig:evaluation-testbed-mix}
    \end{subfigure}
    ~~
    \begin{subfigure}[b]{0.2\textwidth}
        \centering
        \includegraphics[width=\textwidth,trim=0cm 0cm 5.5cm 0.5cm, clip]{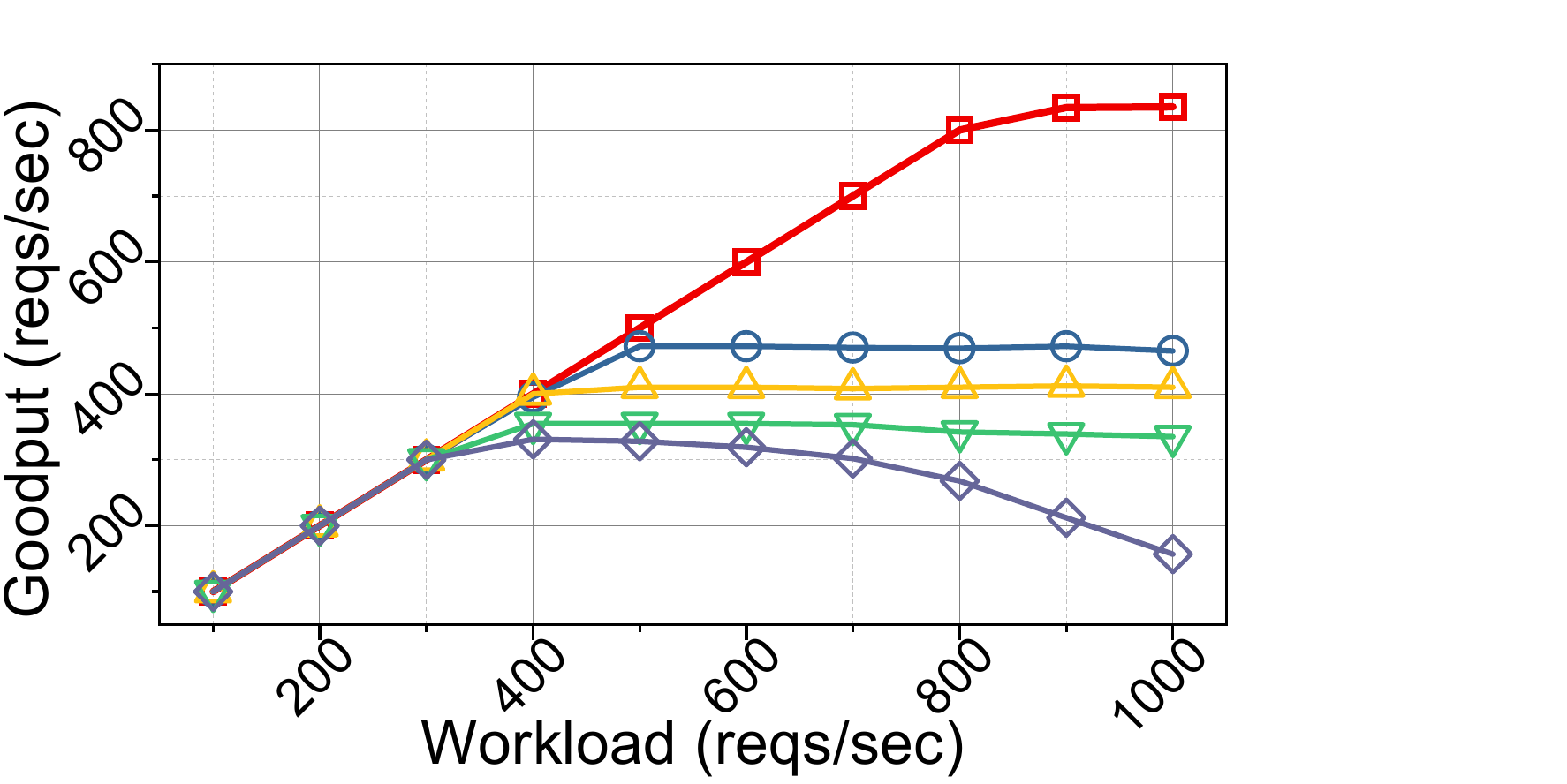}
        \captionsetup{font=footnotesize, justification=centering}
        \vspace{-0.5cm}
        \caption{<1 GPU Frequency-sensitive.}
        \label{fig:evaluation-testbed-lessfrequency}
    \end{subfigure}
    ~~
    \begin{subfigure}[b]{0.2\textwidth}
        \centering
        \includegraphics[width=\textwidth,trim=0cm 0cm 5.5cm 0.5cm, clip]{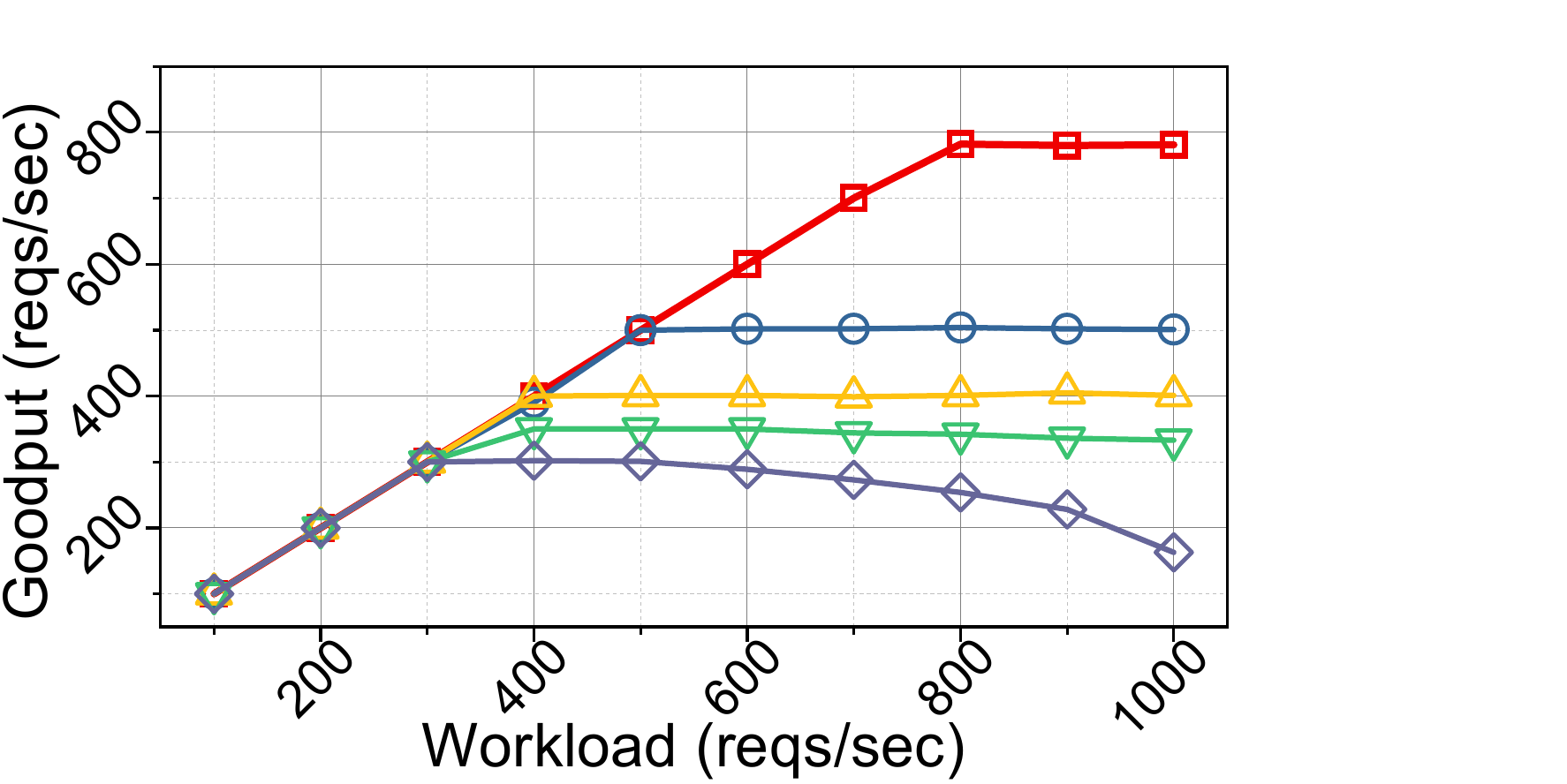}
        \captionsetup{font=footnotesize, justification=centering}
        \vspace{-0.5cm}
        \caption{<1 GPU Latency-sensitive.}
        \label{fig:evaluation-testbed-lesslatency}
    \end{subfigure}
    ~~
    \begin{subfigure}[b]{0.2\textwidth}
        \centering
        \includegraphics[width=\textwidth,trim=0cm 0cm 5.5cm 0.5cm, clip]{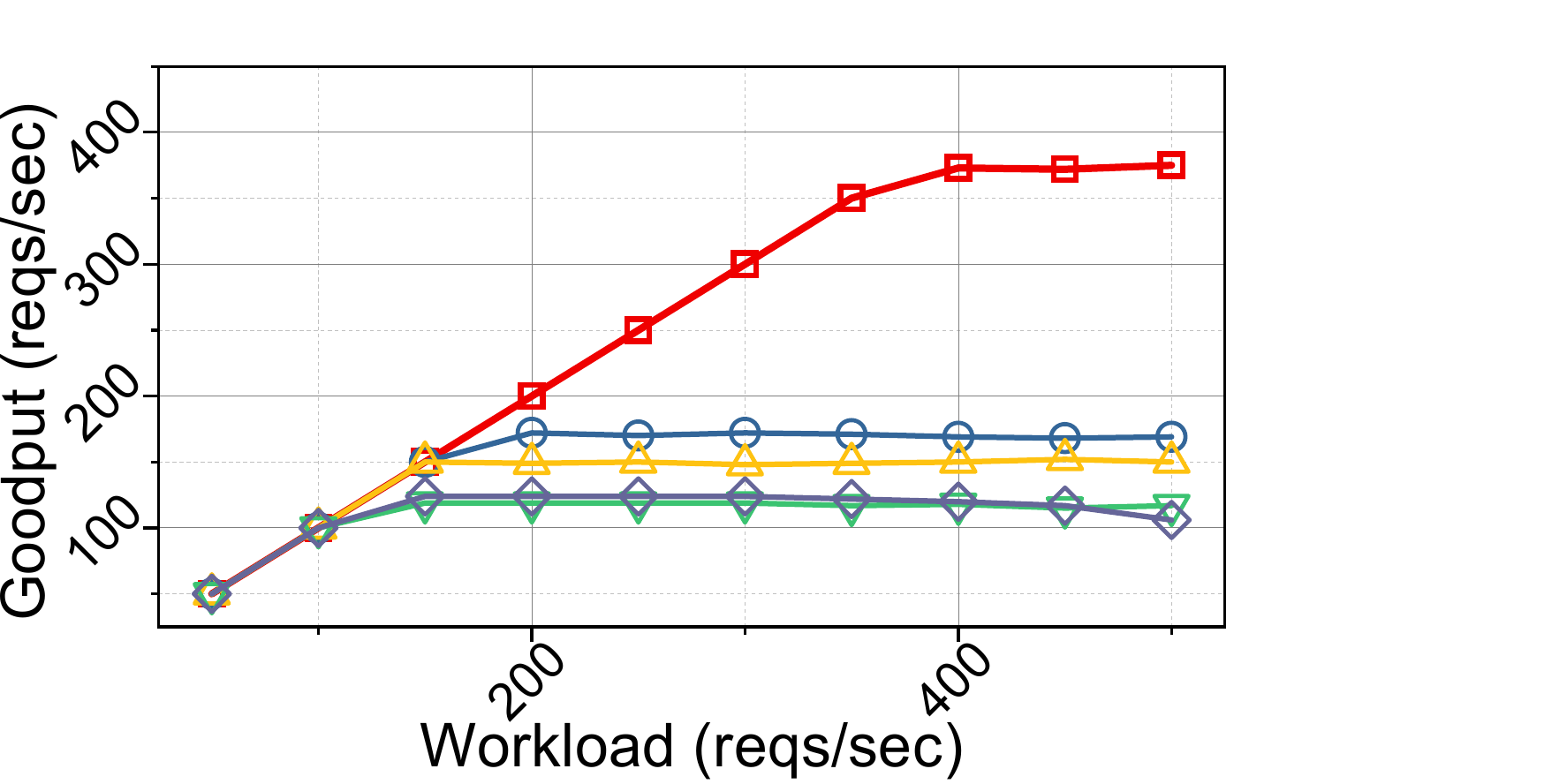}
        \captionsetup{font=footnotesize, justification=centering}
        \vspace{-0.5cm}
        \caption{>1 GPU Frequency-sensitive.}
        \label{fig:evaluation-testbed-morefrequency}
    \end{subfigure}
    ~~
    \begin{subfigure}[b]{0.2\textwidth}
        \centering
        \includegraphics[width=\textwidth,trim=0cm 0cm 5.5cm 0.5cm, clip]{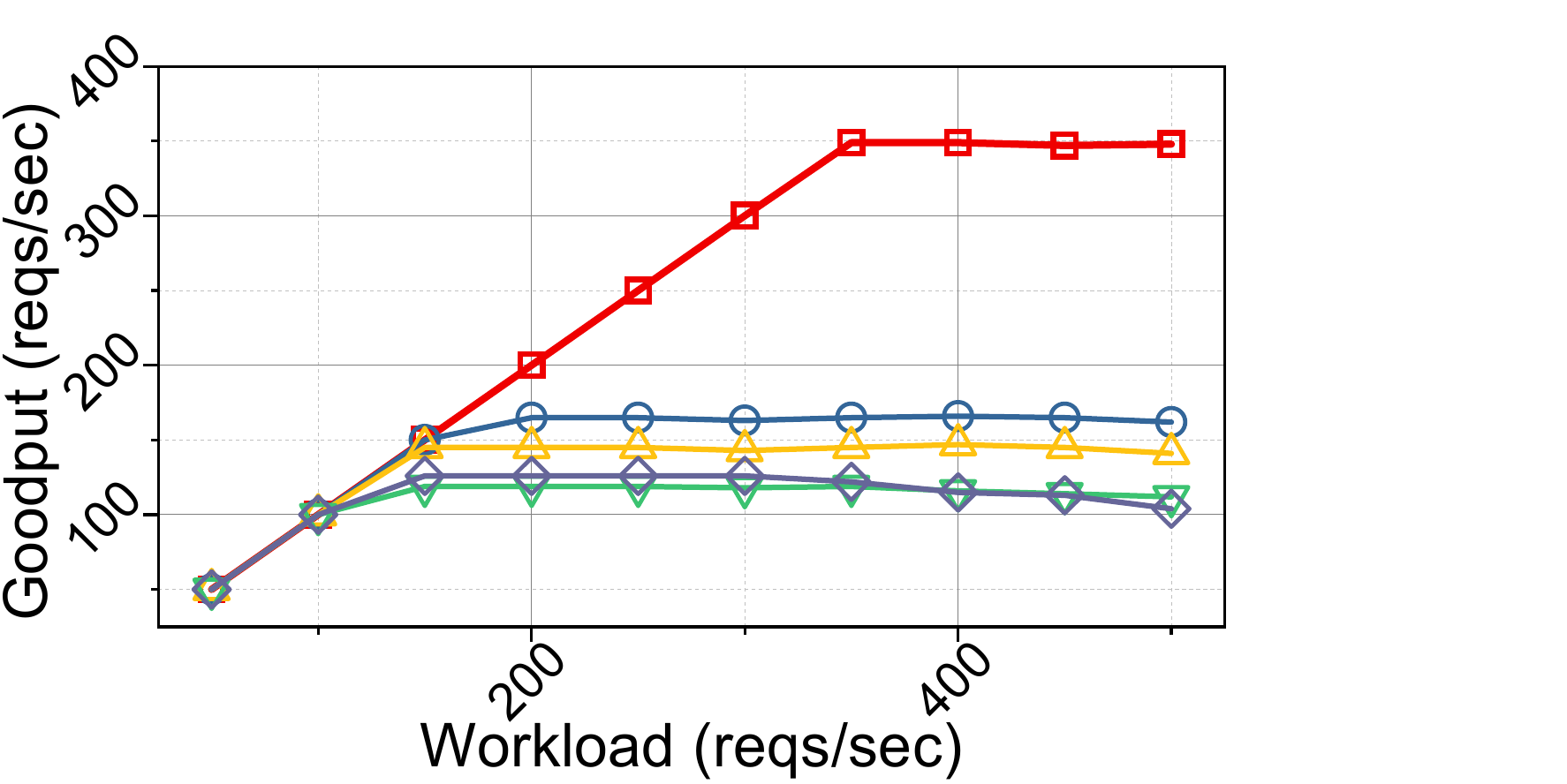}
        \captionsetup{font=footnotesize, justification=centering}
        \vspace{-0.5cm}
        \caption{>1 GPU Latency-sensitive.}
        \label{fig:evaluation-testbed-morelatency}
    \end{subfigure}
    \caption{\EPARAbf overall performance in real testbed.}
    \label{fig:evaluation-testbed-overall}
    \vspace{-2pt}
\end{figure*}

\begin{figure}[thbp]
    \centering
    \begin{subfigure}[thbp]{0.24\textwidth}
        \centering
        \includegraphics[width=\textwidth,trim=0.7cm 17cm 4.5cm 1.5cm, clip]{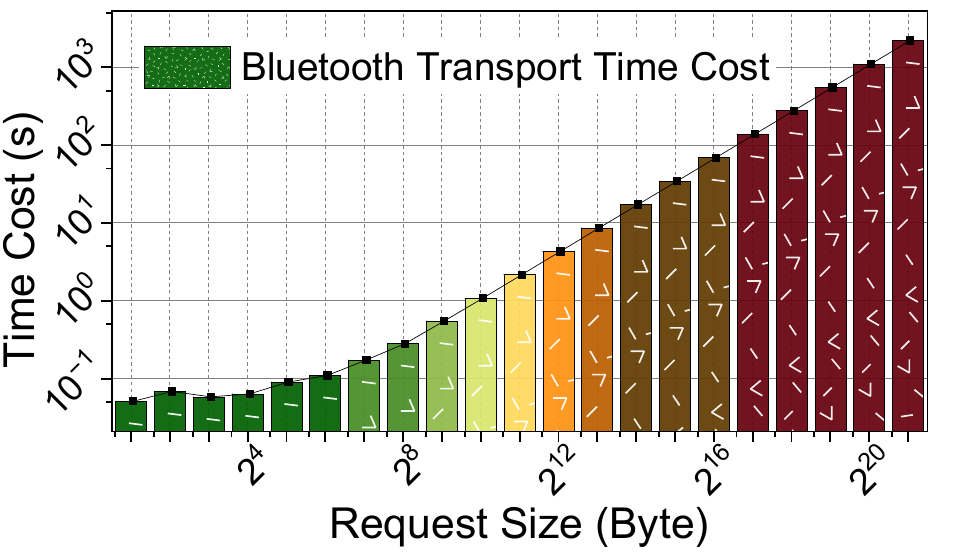}
        \captionsetup{font=footnotesize, justification=centering}
        \vspace{-0.5cm}
        \caption{Bluetooth transport.}
        \label{fig:evaluation-testbed-embedded devices1}
    \end{subfigure}
    ~~
    \centering
    \begin{subfigure}[thbp]{0.24\textwidth}
        \centering
        \includegraphics[width=\textwidth,trim=0.7cm 17cm 4.5cm 1.5cm, clip]{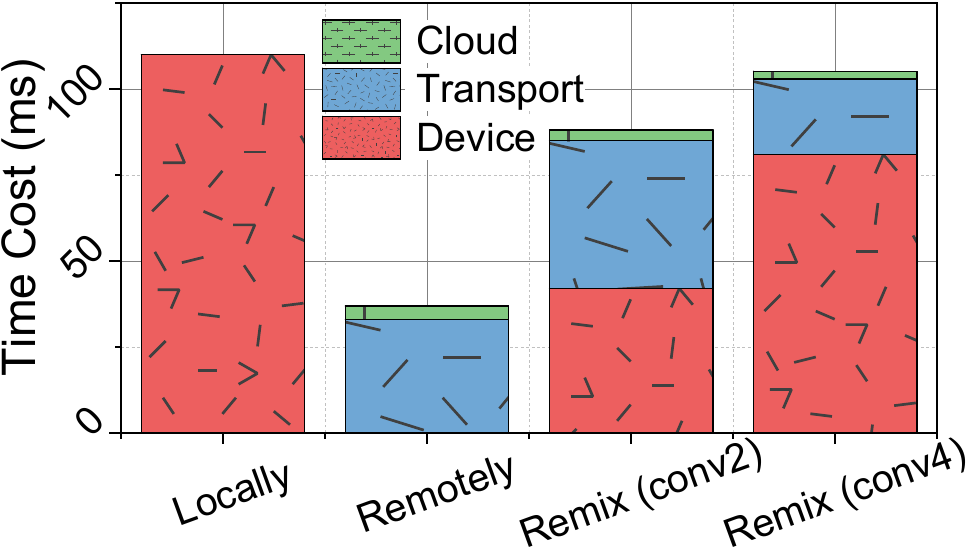}
        \captionsetup{font=footnotesize, justification=centering}
        \vspace{-0.5cm}
        \caption{Accelerator cards.}
        \label{fig:evaluation-testbed-embedded devices2}
    \end{subfigure}
    \caption{\EPARAbf with embedded devices in real testbed.}\label{fig:evaluation-testbed-embedded devices}
    \vspace{9pt}
\end{figure}


The testbed result is shown in Fig.~\ref{fig:evaluation-testbed-overall-overall} and Fig.~\ref{fig:evaluation-testbed-overall}, that respectively present the overall and detailed goodput under five workloads.
\EPARA achi\-eves the best performance in average goodput for all workloads. Compared to InterEdge, AlpaServe, Galaxy and SERV-P, \EPARA improves the overall goodput by up to 2.1$\times$, 2.2$\times$, 2.5$\times$ and 3.2$\times$, respectively, for the mixed request workload, and 1.9$\times$, 2.2$\times$, 2.6$\times$ and 3.9$\times$, respectively, for the frequency-sensitive workload. When the workload remains below \EPARA's maximum goodput, \EPARA fulfills requests with a probability exceeding 99.4\%. Conversely, for workloads exceeding this goodput, \EPARA maintains a stable goodput of at least 98.1\% of its maximum value. These results verify the effectiveness and stability of \EPARA in edge clouds AI serving.

\subsubsection{Other Testbed Performance}
~\\
\noindent{\bf Bluetooth devices:}
We leverage HC-05 bluetooth module with Xilinx Basys 3 as an independent edge device to connect to the edge server, and evaluate \EPARA's bluetooth processing capabilities.
The serial ports on both server and device are invoked to establish a bluetooth communication link. Subsequently, we measure the transmission time delay (from send beginning to receive ending) for files of varying sizes, as illustrated in the Fig.~\ref{fig:evaluation-testbed-embedded devices1}. Experimental results indicate that bluetooth data transmission exhibits a 105ms delay for 64B file and a 1039ms delay for 1KB file, making it suitable for text-related tasks with loose latency SLOs.

\noindent{\bf Accelerator cards:}
We employ Xilinx Alveo U50 as edge device accelerator cards integrated into \EPARA to evaluate \EPARA's collaborative computing capabilities~\cite{huangCLIOEnablingAutomatic2020}.
We deploy a part of VGG16~\cite{simonyanVeryDeepConvolutional2015} to U50~\cite{wangPipeCNNOpenCLbasedOpensource2017} based on different offloading points (conv2, conv4), and the remaining part is handled by edge server, as shown in Fig.~\ref{fig:evaluation-testbed-embedded devices2}. 
\EPARA considers this synergy as PP and is able to handle collaborative requests correctly. 

\noindent{\bf Resource monitor:}
We monitor \EPARA's resource utilization capability with Nvidia Nsight~\cite{nvidiaNVIDIANsightSystems2025} through two key metrics: memory resource (VRAM) and computing resource (\texttt{achieved\_occupancy} in Nsight)~\cite{shubhaUSHERHolisticInterference2024}. 
As shown in Fig.\ref{fig:evaluation-testbed-resource}, \EPARA achieves 95\%+ computing resource utilization and 98\%+ VRAM utilization.
These findings indicate that \EPARA can effectively utilize dispersed edge server resources.
This is mainly due to \EPARA using MF to increase batch size and leading datacenter work (AlpaServe), while significantly leading the work that does not use MT (Galaxy).

\subsection{Large-scale Simulations}\label{section:evaluation-simulation}

\noindent{\bf Comparisions}:
In addition to InterEdge~\cite{brownArchitectureEdgeNetworking2024}, AlpaServe~\cite{liAlpaServeStatisticalMultiplexing2023}, Galaxy~\cite{yeGalaxyResourceEfficientCollaborative2024a} and SERV-P~\cite{farhadiServicePlacementRequest2021}, we also compare \EPARA with two recent AI inference systems: USHER~\cite{shubhaUSHERHolisticInterference2024} and DeTransformer~\cite{weiCommunicationEfficientModelParallelism2024}. For SERV-P, we always combine servers into groups of 10 for scheduling, otherwise we cannot solve it within a feasible time.

\noindent{\bf Settings}:
We adopt an event-driven simulation architecture, which achieves up to 10$\times$ speed up compared to the discrete-time stepping approach.
Our simulator fully executes the request scheduling process but bypasses the actual execution of packet transmission, and model computations.
For specific strategies, transmission latency is simulated based on service-specific data volumes and network bandwidth, while computational latency is derived from lookup tables indexed by GPU and AI service, which are precomputed from our real-world experimental results. 

\noindent{\bf Large-scale performance:}
As illustrated in Fig.~\ref{fig:evaluation-simulation overall}, the simulation results demonstrate the goodput performance across varying numbers of servers equipped with eight P100 GPUs. For latency-sensitive requests, \EPARA achieves 1.5$\times$-2.0$\times$ higher goodput compared to baseline methods by simultaneously incorporating inter-server request handling and task-GPU allocation. For frequency-sensitive requests, \EPARA further enhances performance through multi-frame and DP, resulting in 2.8$\times$-3.1$\times$ superior goodput over conventional approaches. For mixed request scenarios, \EPARA maintains competitive advantages with 1.6$\times$-2.4$\times$ goodput improvements compared to existing methods.

\begin{figure}[!t]
	\centering
	\addtolength{\abovecaptionskip}{-13pt}
	\addtolength{\belowcaptionskip}{5pt}
    \includegraphics[width=0.95\columnwidth,trim=1.5cm 21.4cm 0.5cm 1.4cm, clip]{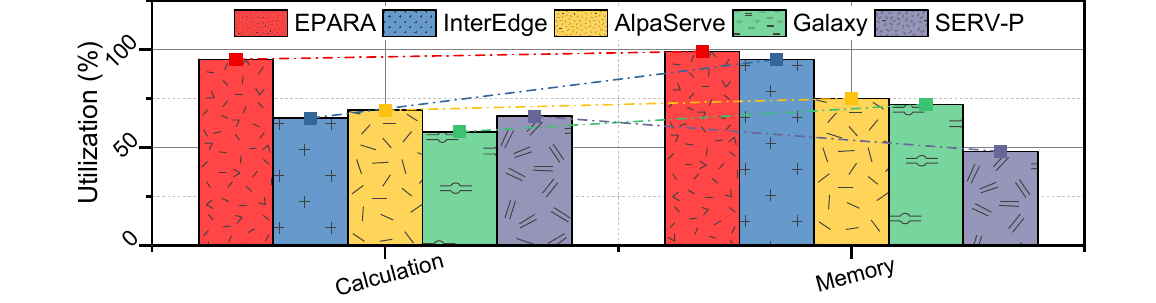}
	\caption{\EPARAbf resources monitor when serving mixed workloads with maximum goodput in real testbed.}\label{fig:evaluation-testbed-resource}
	\vspace{3.5pt}
\end{figure}

\begin{figure*}[thbp]
    \centering
    \begin{subfigure}[b]{0.2\textwidth}
        \centering
        \includegraphics[width=\textwidth,trim=0cm 0cm 5.5cm 0.7cm, clip]{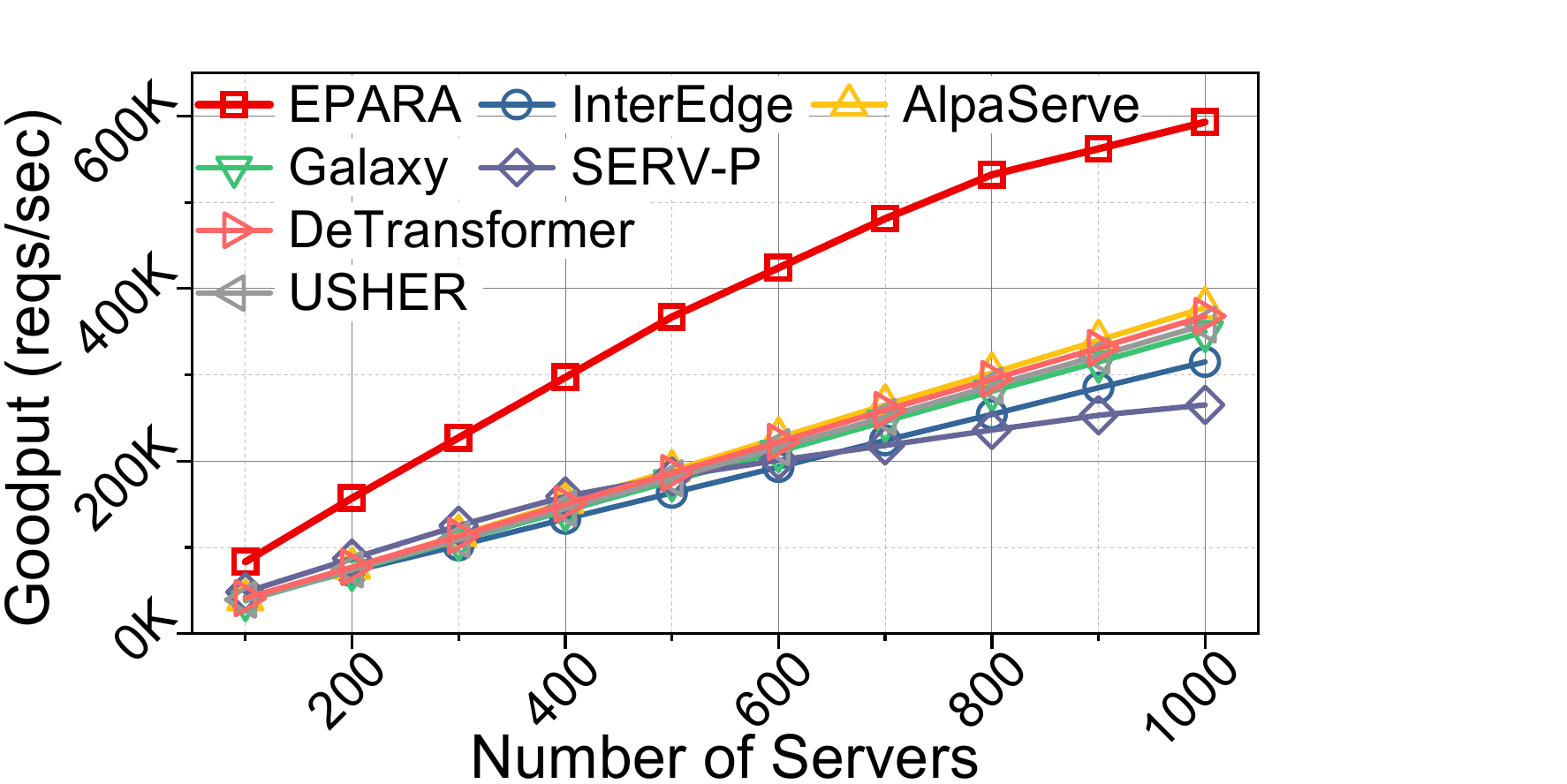}
        \captionsetup{font=footnotesize, justification=centering}
        \vspace{-0.5cm}
        \caption{Mixed request.}
        \label{fig:simulation-overall-mix}
    \end{subfigure}
    ~~
    \begin{subfigure}[b]{0.2\textwidth}
        \centering
        \includegraphics[width=\textwidth,trim=0cm 0cm 5.5cm 0.7cm, clip]{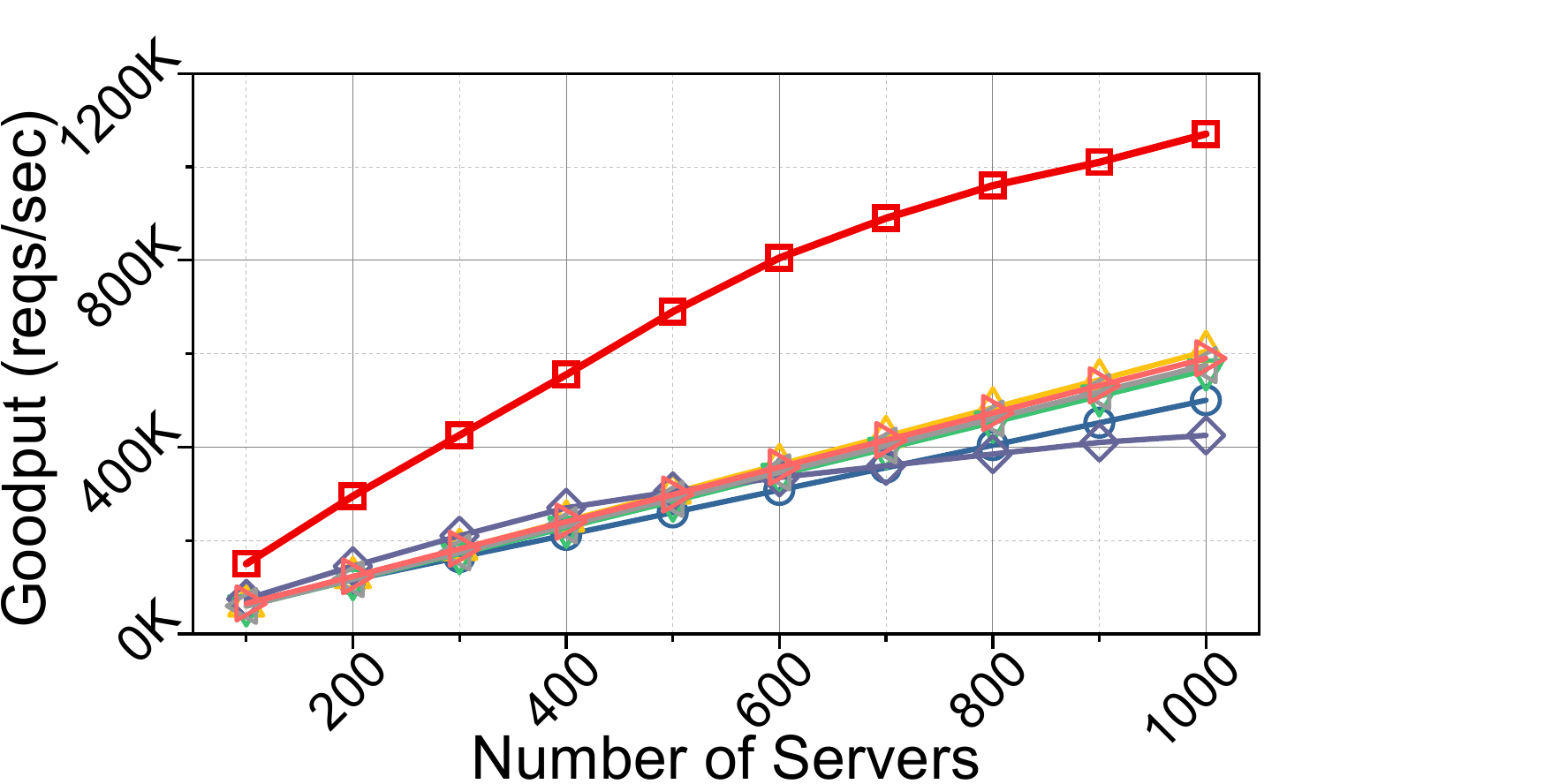}
        \captionsetup{font=footnotesize, justification=centering}
        \vspace{-0.5cm}
        \caption{<1 GPU frequency-sensitive.}
        \label{fig:simulation-overall-lessfrequency}
    \end{subfigure}
    ~~
    \begin{subfigure}[b]{0.2\textwidth}
        \centering
        \includegraphics[width=\textwidth,trim=0cm 0cm 5.5cm 0.7cm, clip]{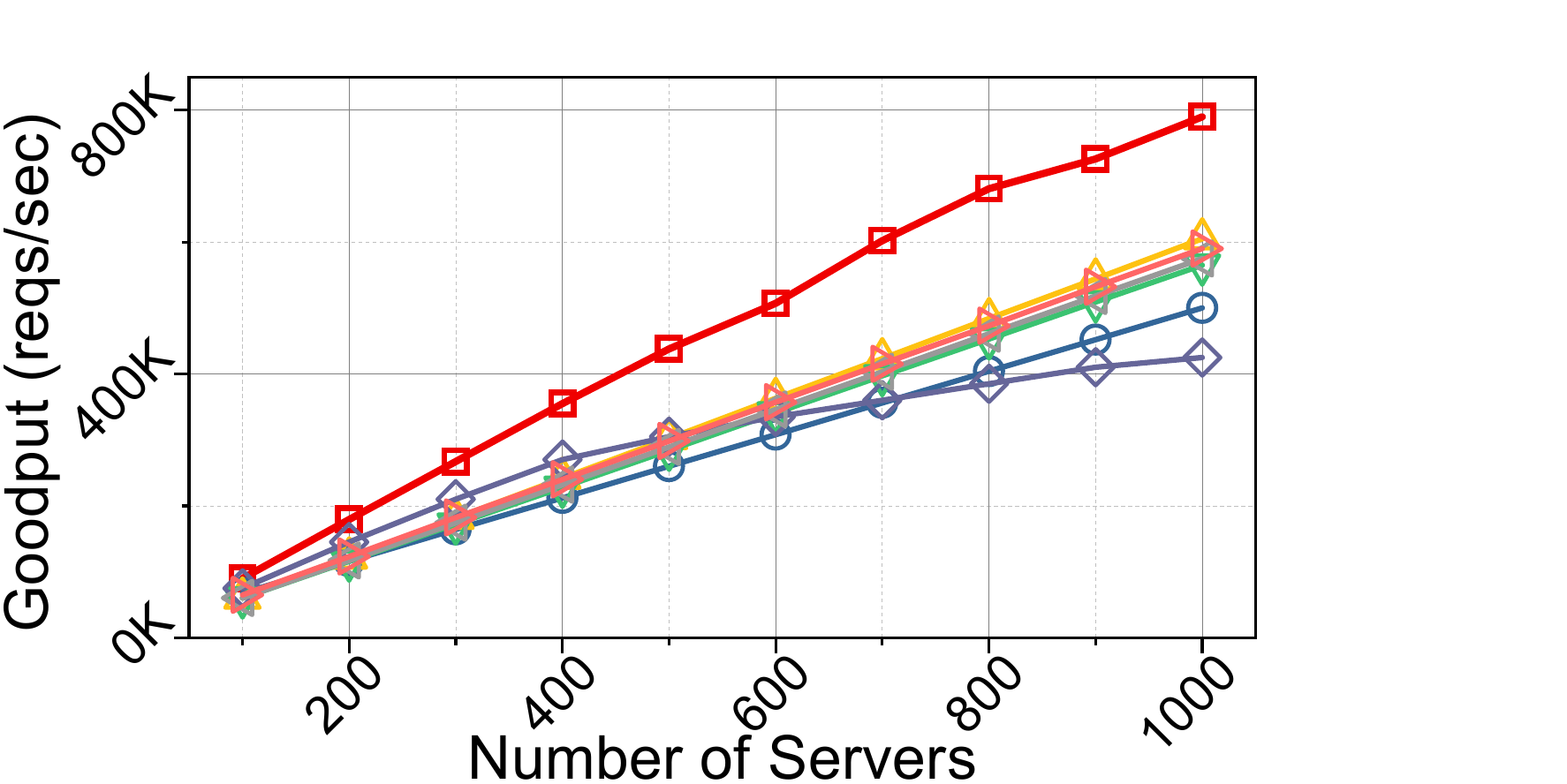}
        \captionsetup{font=footnotesize, justification=centering}
        \vspace{-0.5cm}
        \caption{<1 GPU latency-sensitive.}
        \label{fig:simulation-overall-lesslatency}
    \end{subfigure}
    ~~
    \begin{subfigure}[b]{0.2\textwidth}
        \centering
        \includegraphics[width=\textwidth,trim=0cm 0cm 5.5cm 0.7cm, clip]{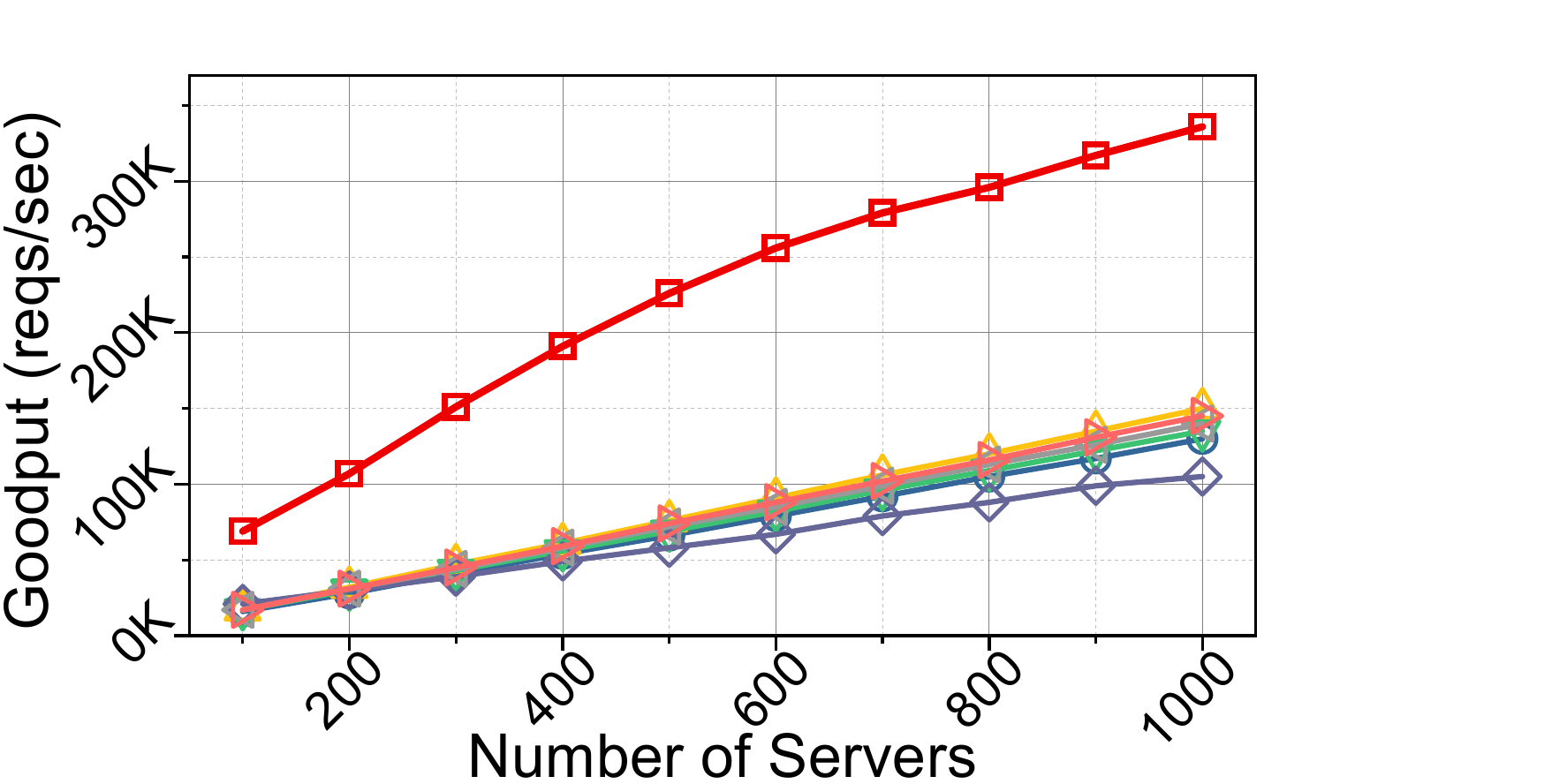}
        \captionsetup{font=footnotesize, justification=centering}
        \vspace{-0.5cm}
        \caption{>1 GPU frequency-sensitive.}
        \label{fig:simulation-overall-morefrequency}
    \end{subfigure}
    ~~
    \begin{subfigure}[b]{0.2\textwidth}
        \centering
        \includegraphics[width=\textwidth,trim=0cm 0cm 5.5cm 0.7cm, clip]{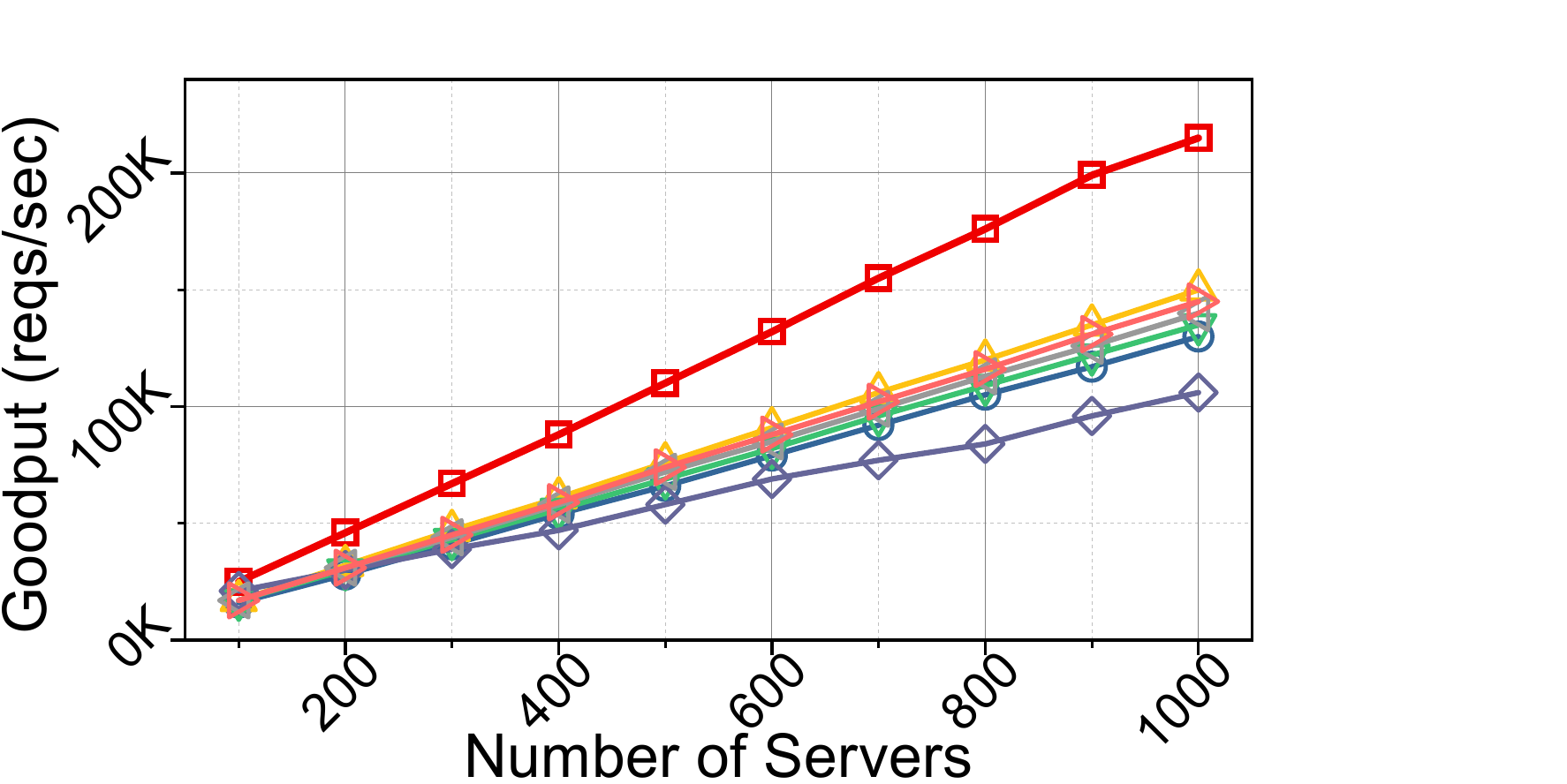}
        \captionsetup{font=footnotesize, justification=centering}
        \vspace{-0.5cm}
        \caption{>1 GPU latency-sensitive.}
        \label{fig:simulation-overall-morelatency}
    \end{subfigure}

    \caption{Large-scale goodput of comparison methods under different workloads.}
    \label{fig:evaluation-simulation overall}
\end{figure*}
\begin{figure*}[thbp]
    \centering
    \begin{subfigure}[b]{0.2\textwidth}
        \centering
        \includegraphics[width=\textwidth,trim=0.72cm 0cm 4.78cm 0.5cm, clip]{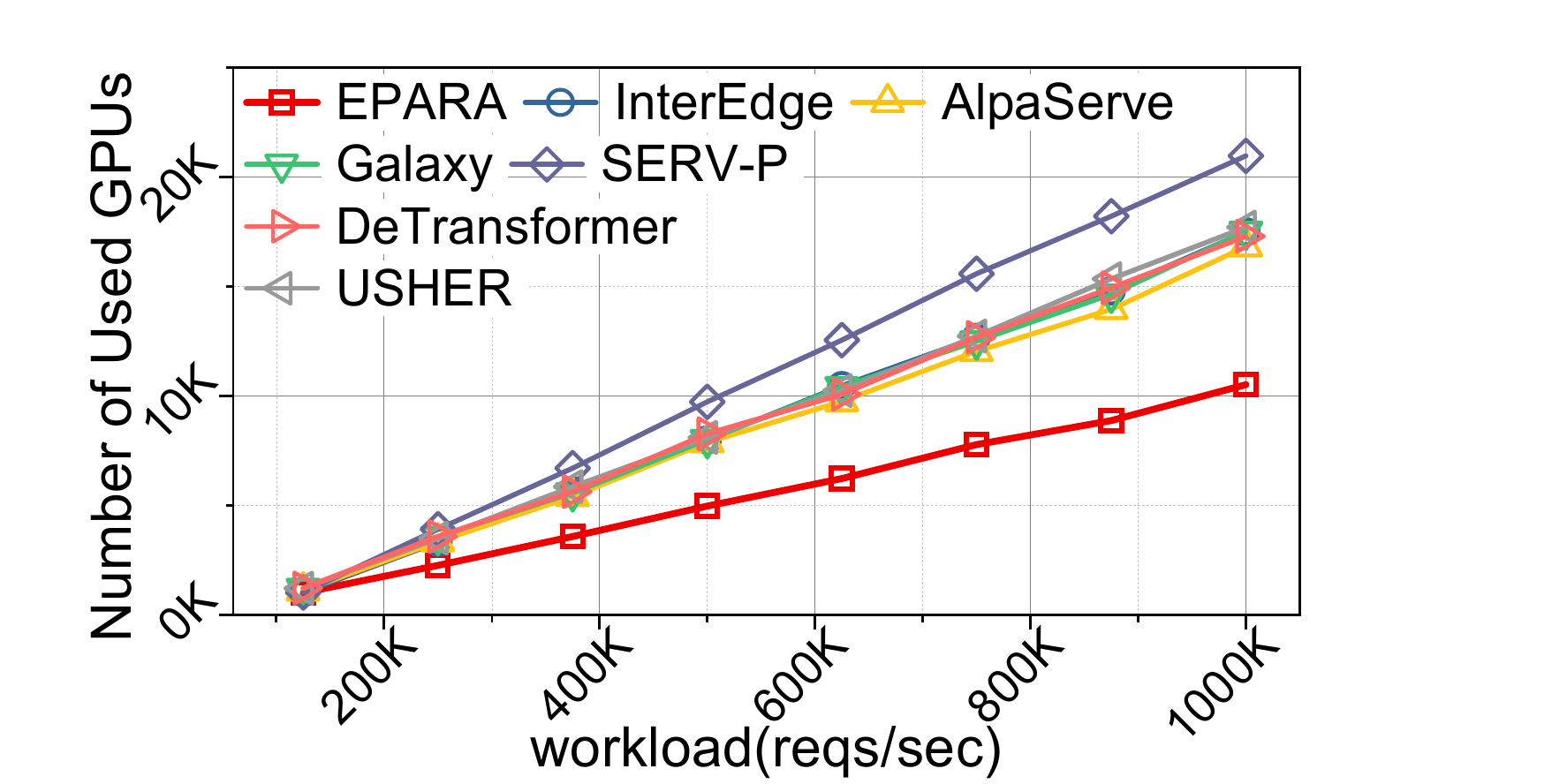}
        \captionsetup{font=footnotesize, justification=centering}
        \vspace{-0.5cm}
        \caption{Mixed request.}
        \label{fig:simulation-overall-mix}
    \end{subfigure}
    ~~
    \begin{subfigure}[b]{0.2\textwidth}
        \centering
        \includegraphics[width=\textwidth,trim=1.2cm 0cm 4.3cm 0.5cm, clip]{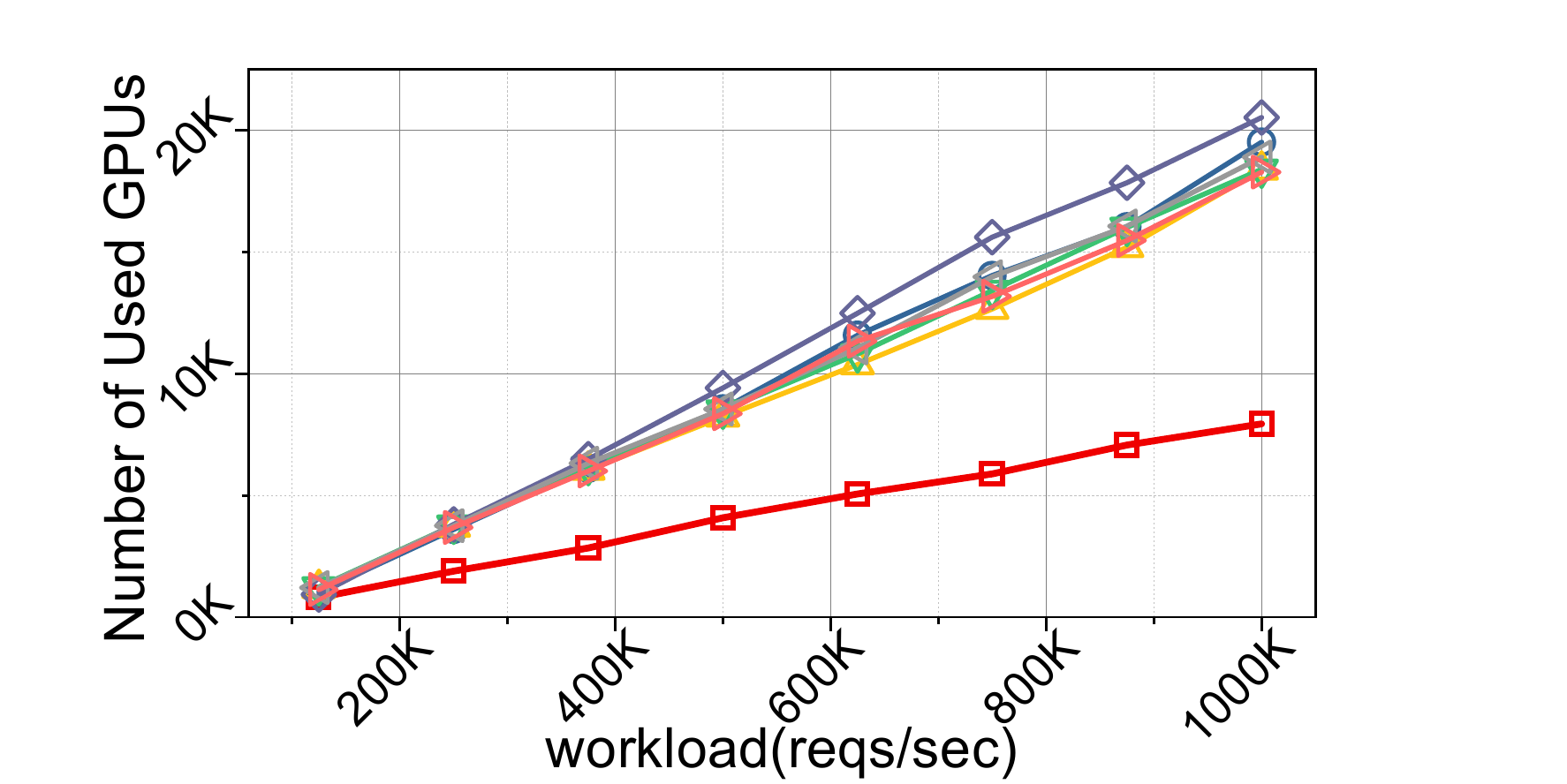}
        \captionsetup{font=footnotesize, justification=centering}
        \vspace{-0.5cm}
        \caption{<1 GPU frequency-sensitive.}
        \label{fig:simulation-overall-lessfrequency}
    \end{subfigure}
    ~~
    \begin{subfigure}[b]{0.2\textwidth}
        \centering
        \includegraphics[width=\textwidth,trim=1.2cm 0cm 4.3cm 0.5cm, clip]{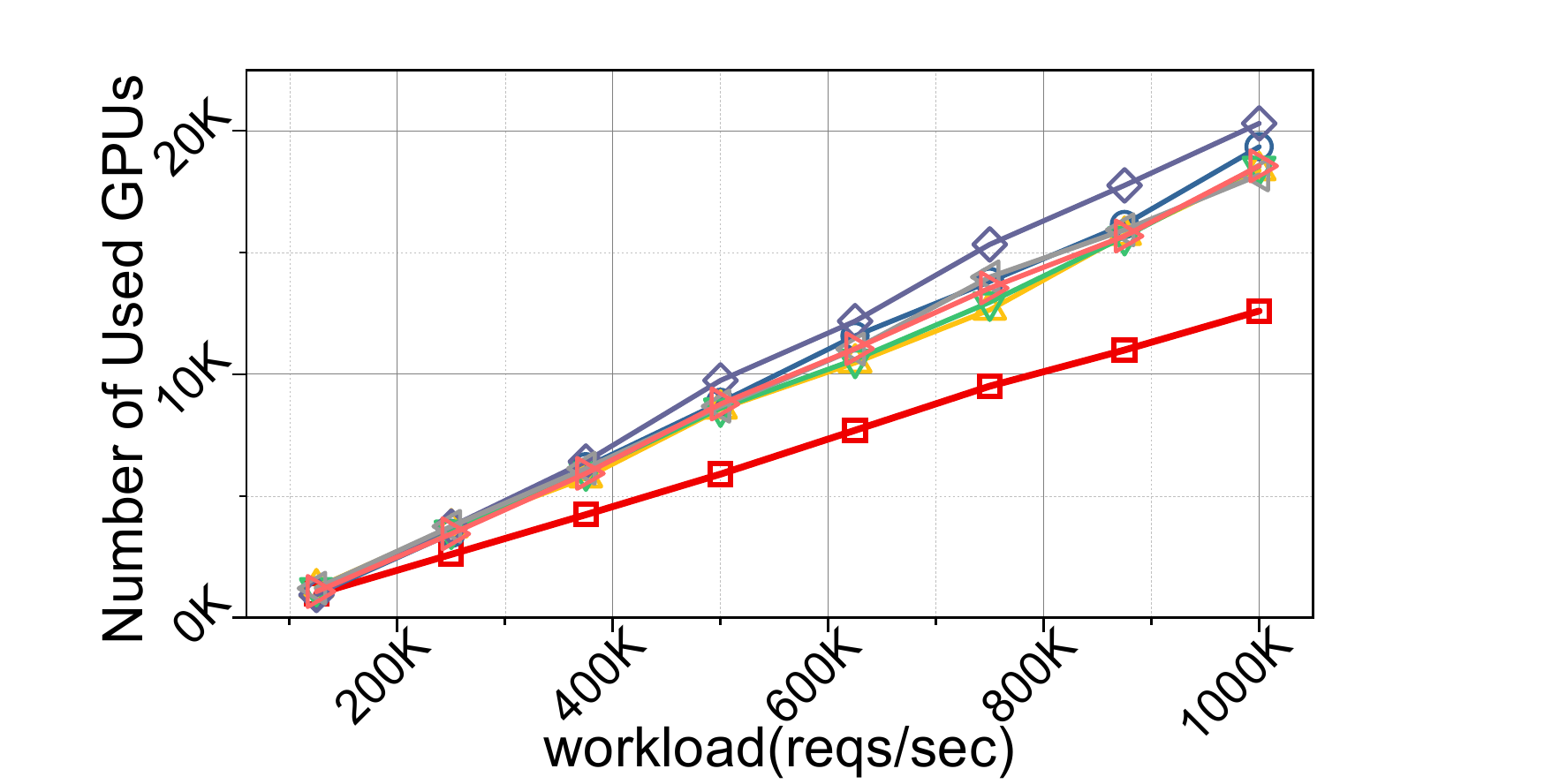}
        \captionsetup{font=footnotesize, justification=centering}
        \vspace{-0.5cm}
        \caption{<1 GPU latency-sensitive.}
        \label{fig:simulation-overall-lesslatency}
    \end{subfigure}
    ~~
    \begin{subfigure}[b]{0.2\textwidth}
        \centering
        \includegraphics[width=\textwidth,trim=1.5cm 0cm 4cm 0.5cm, clip]{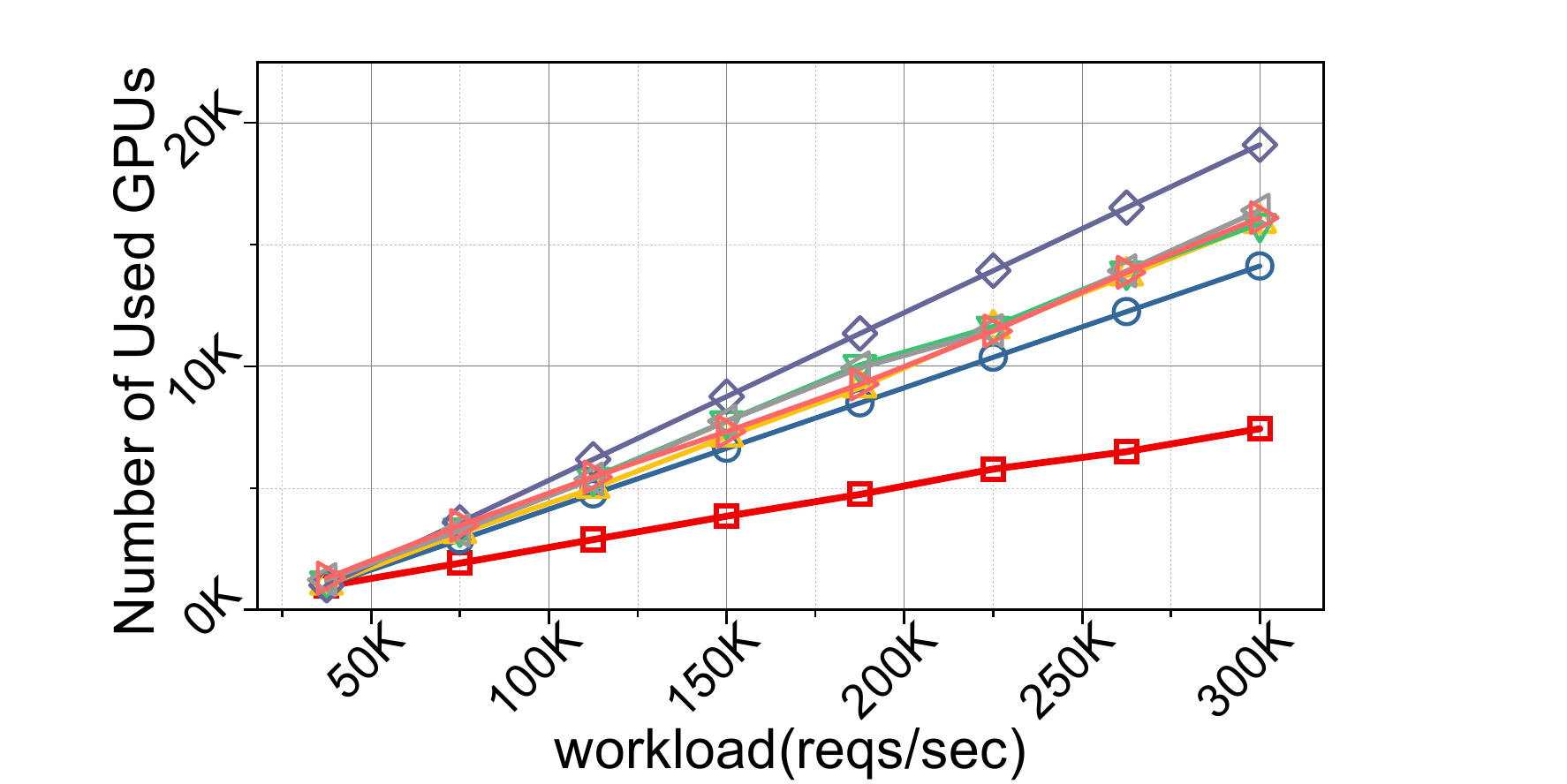}
        \captionsetup{font=footnotesize, justification=centering}
        \vspace{-0.5cm}
        \caption{>1 GPU frequency-sensitive.}
        \label{fig:simulation-overall-morefrequency}
    \end{subfigure}
    ~~
    \begin{subfigure}[b]{0.2\textwidth}
        \centering
        \includegraphics[width=\textwidth,trim=1.37cm 0cm 4.13cm 0.5cm, clip]{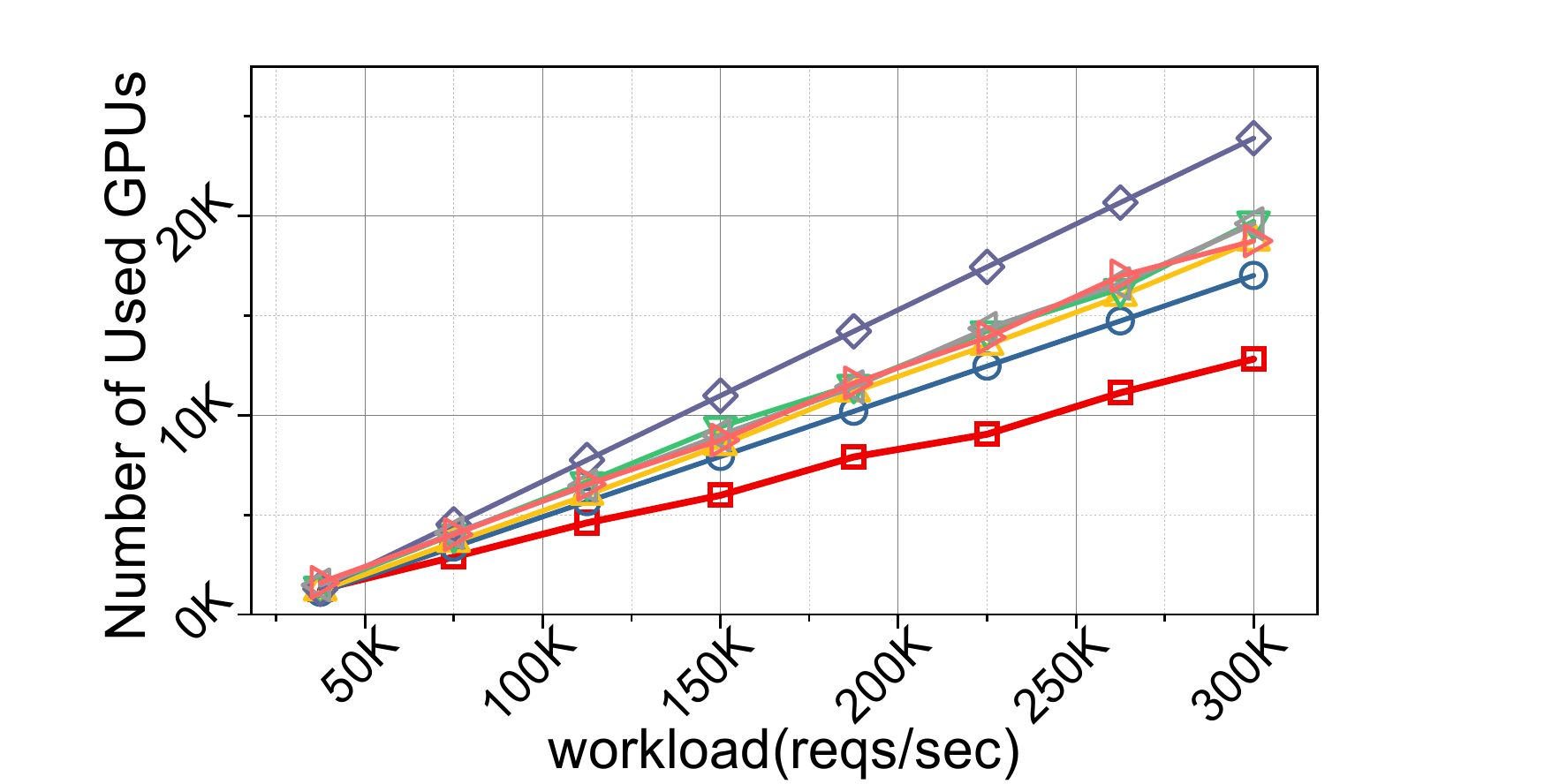}
        \captionsetup{font=footnotesize, justification=centering}
        \vspace{-0.5cm}
        \caption{>1 GPU latency-sensitive.}
        \label{fig:simulation-overall-morelatency}
    \end{subfigure}

    \caption{Used GPUs of comparison methods in non-fixed clusters under different workloads.}
    \label{fig:evaluation-simulation non-fixed}
\end{figure*}

\noindent{\bf Non-fixed cluster performance:}
We consider a scenario where each server possesses the capability to scale up GPU resources, which represents a common practice in cloud computing environments. Through empirical evaluation, we quantify the GPU requirements of different methods for handling user requests at identical scales. Fig.~\ref{fig:evaluation-simulation non-fixed} shows the number of used GPUs in simulation to complete all inference requests within their SLOs. \EPARA requires 1.5$\times$-2.6$\times$ fewer GPUs because \EPARA can schedule tasks across servers and efficiently execute category-based task parallelism, avoiding situations where a single server needs to significantly increase GPUs to handle peak-time requests.

\begin{figure*}[thbp]
    \vspace{2pt}
	\centering
	\abovecaptionskip=0pt
	\belowcaptionskip=-10pt
	\includegraphics[width=\textwidth]{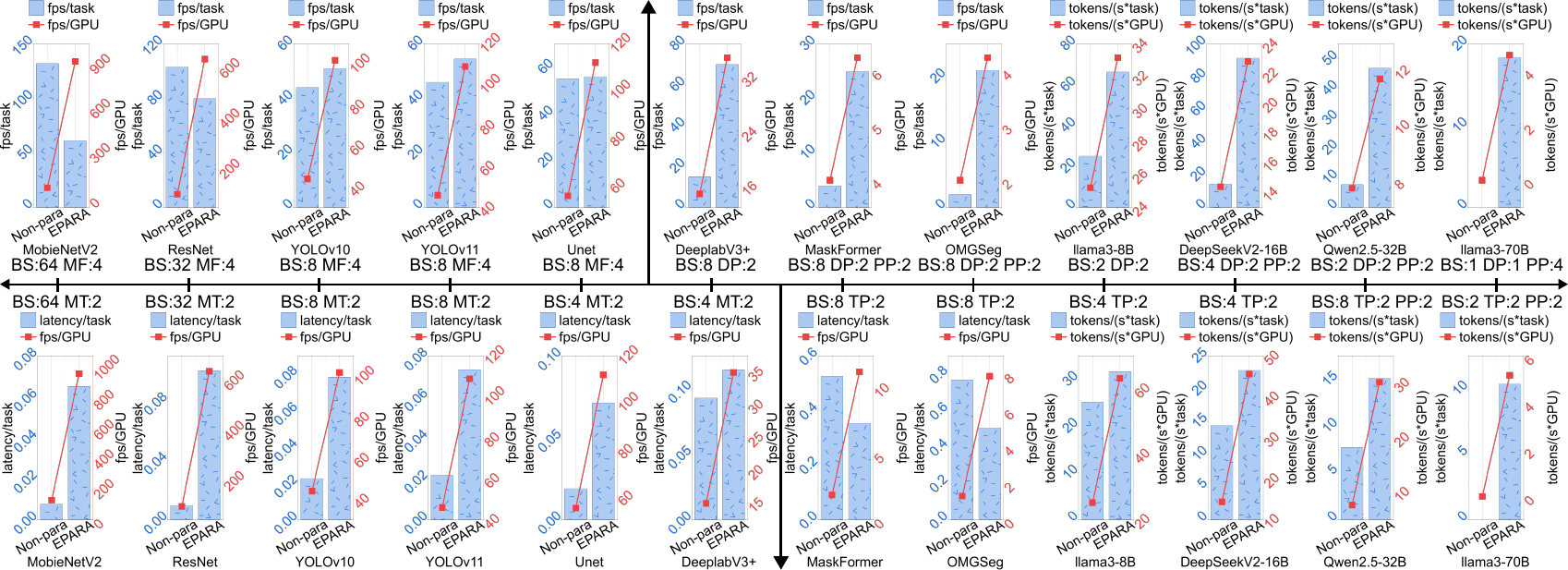}
	\caption{Effect of \EPARAbf parallelism allocator in real testbed.}
	\label{fig:evaluation-deepdive-allocator}
    \vspace{-2pt}
\end{figure*}

\subsection{EPARA Deep Dive}\label{section:evaluation-deepdive}
\subsubsection{Effectiveness of \EPARAit's Design Components.}\label{section:evaluation-deepdive-effect}
~\\
\noindent{\bf Effect of parallelism allocator:}
We validate the effect of task-categorized parallelism allocator discussed in \cref{section:design-parallelism} in real testbed. As demonstrated in Fig.~\ref{fig:evaluation-deepdive-allocator}, we categorize tasks in Table~\ref{table:evaluation-models} into four distinct classes aligning with those defined in Fig.~\ref{fig:design-task categorized}. \EPARA maintains frequency-sensitive SLOs across two dimensions: frame rate stabilization at 60-100 fps for vision tasks, and service frequency regulation at 10-30 tokens/s (with vllm~\cite{kwonEfficientMemoryManagement2023a}) for generative AI workloads. For all task categories, \EPARA consistently enhances per-GPU service processing capacity by up to 12.4$\times$ compared to non-parallelism deployment.
We analyze each of the four categories in detail as follows.
First, for <1 GPU frequency tasks (e.g., MobileNetV2 video), \EPARA allocator effectively constrains performance beyond SLOs through multi-task, multi-frame and batching, thereby achieving 5.9$\times$-12.4$\times$ higher goodput per GPU.
Second, for >1 GPU frequency tasks (e.g., LLaMA3-70B HCI), \EPARA allocator orchestrates DP, MP, batching, multi-task, and multi-frame (here we set MF equal to BS), achieving 1.3$\times$-2.5$\times$ higher goodput per GPU.
Third, for <1 GPU latency tasks, \EPARA allocator ensures that tasks are processed within SLOs through multi-task and batching, achieving 2.3$\times$-9.1$\times$ higher goodput per GPU. 
Fourth, for >1 GPU latency tasks, \EPARA guarantees end-to-end execution latencies within predefined SLOs by implementing MP, BS and multi-task, achieving 2.9$\times$-4.5$\times$ higher goodput per GPU.

\noindent{\bf Effect and overhead of request handling:}
First, we examine the effect of request handling in \EPARA. Assuming that all computations in \EPARA must be completed at the first hop upon receiving a request without offloading. We compare this scenario with standard \EPARA implementation, as illustrated in Fig.~\ref{fig:evaluation-deepdive-effect-offloading}. Request handling enhances performance by 2.2$\times$–2.4$\times$ for <1 GPU tasks and 2.9$\times$–3.1$\times$ for >1 GPU tasks.
Second, regarding handling latency, each handling operation comprises scheduling latency and offloading latency. For scheduling latency, simulation experiments under varying edge server scales (up to 10k nodes) reveal that \EPARA's scheduling latency remains below 20ms. For offloading transmission latency, it mainly depends on task payload and bandwidth. Through direct-connection TCP protocol in our testbed, network transmission latency remains under 5ms when bandwidth exceeds 100Mbps, demonstrating \EPARA's capability to maintain efficient task allocation and request handling without requiring high bandwidth datacenter network.

\noindent{\bf Effect and overhead of service placement:}
In \cref{section:design-placement}, we adopt a submodular function placement strategy. First, its effect is validated in comparison with LRU, LFU, and MFU. An effective placement should align services with requests, thereby maximizing goodput. As demonstrated in Fig.~\ref{fig:evaluation-deepdive-effect-placementgoodput}, \EPARA's placement strategy achieves up to 1.9$\times$ improvements in overall goodput. 
Second, experimental evaluations are conducted to analyze scheduling latency under varying numbers of edge servers. 
Fig.~\ref{fig:evaluation-deepdive-effect-placementlatency} reveals that \EPARA's scheduling latency for a single placement remains under 200ms with less than 10k servers.

\begin{figure*}[thbp]
    \centering
    \begin{subfigure}[b]{0.2\textwidth}
        \centering
        \includegraphics[width=\textwidth,trim=0.65cm 17cm 4.8cm 1.5cm, clip]{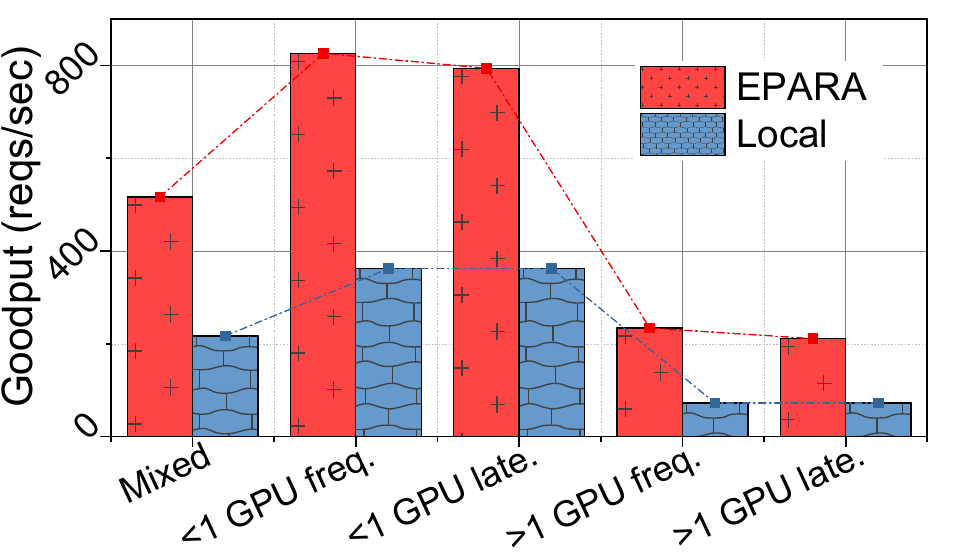}
        \captionsetup{font=footnotesize, justification=centering}
        \vspace{-0.5cm}
        \caption{Effect of offloading.}\label{fig:evaluation-deepdive-effect-offloading}
    \end{subfigure}
    ~~
    \begin{subfigure}[b]{0.2\textwidth}
        \centering
        \includegraphics[width=\textwidth,trim=0cm 0.1cm 6.6cm 0.8cm, clip]{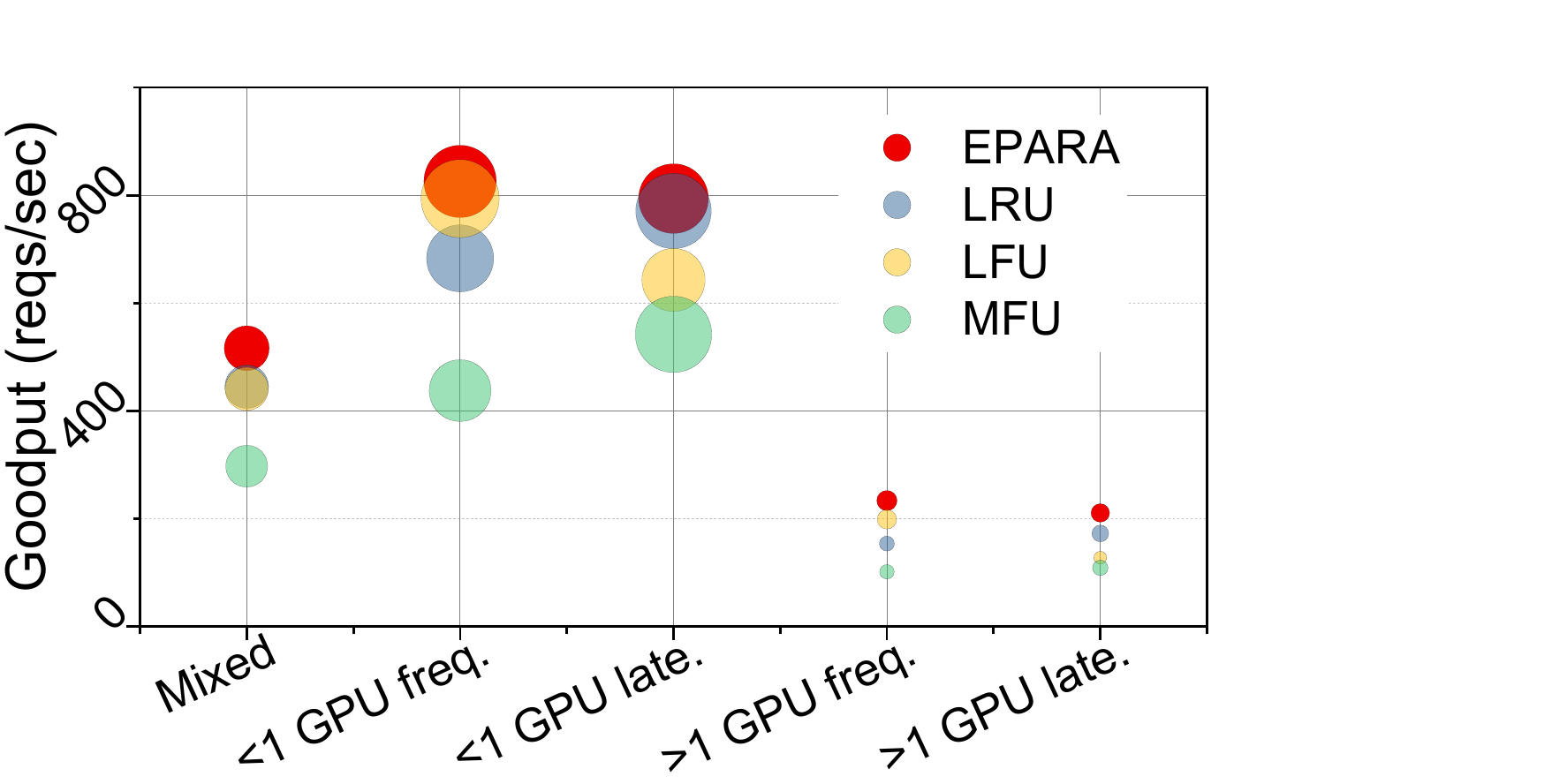}
        \captionsetup{font=footnotesize, justification=centering}
        \vspace{-0.5cm}
        \caption{Effect of placement.}\label{fig:evaluation-deepdive-effect-placementgoodput}
    \end{subfigure}
    ~~
    \begin{subfigure}[b]{0.2\textwidth}
        \centering
        \includegraphics[width=\textwidth,trim=1.6cm 0cm 4cm 0cm, clip]{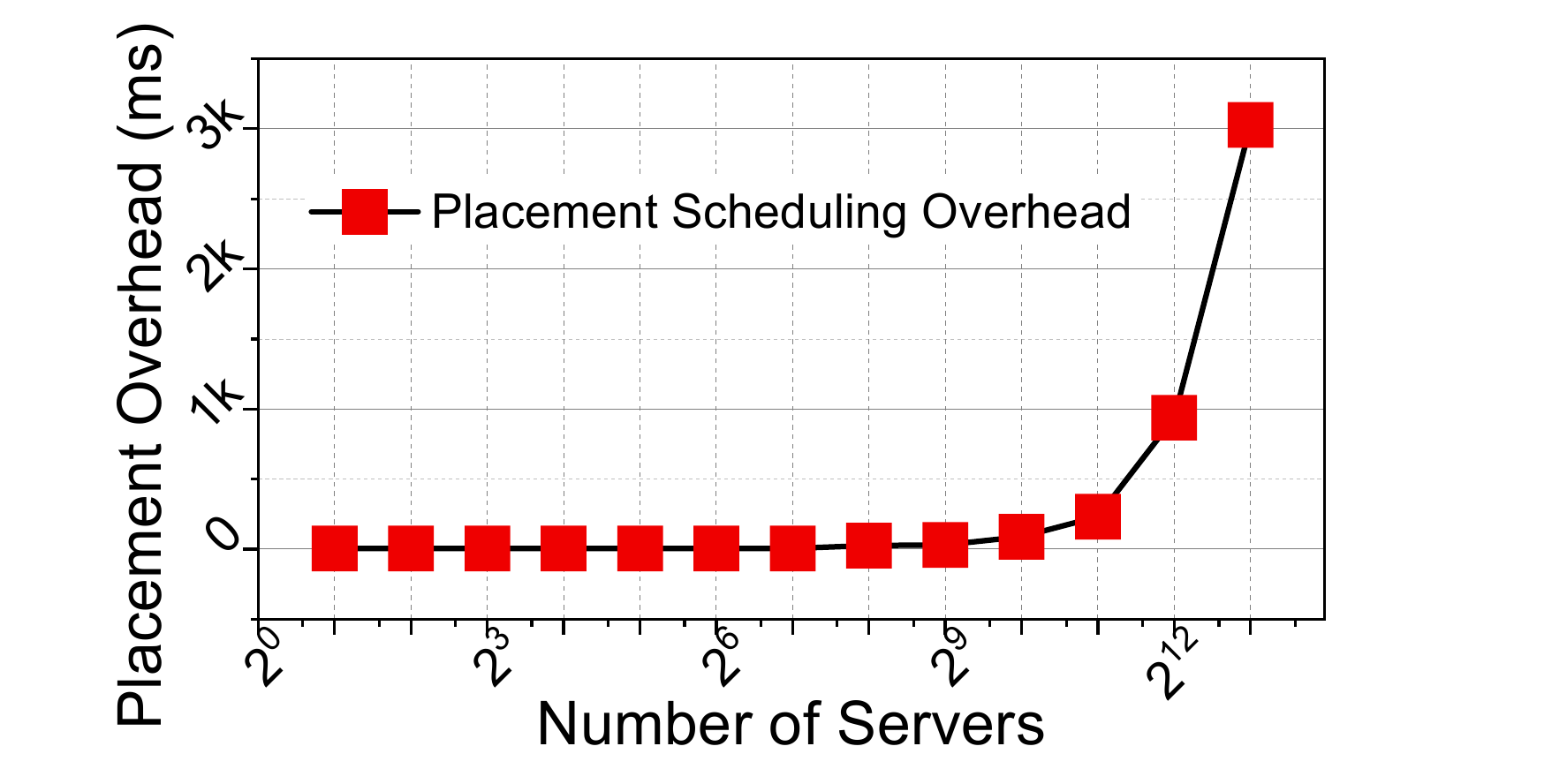}
        \captionsetup{font=footnotesize, justification=centering}
        \vspace{-0.5cm}
        \caption{Overhead of placement.}\label{fig:evaluation-deepdive-effect-placementlatency}
    \end{subfigure}
    ~~
    \begin{subfigure}[b]{0.2\textwidth}
        \centering
        \includegraphics[width=\textwidth,trim=0cm 0cm 1.3cm 0.1cm, clip]{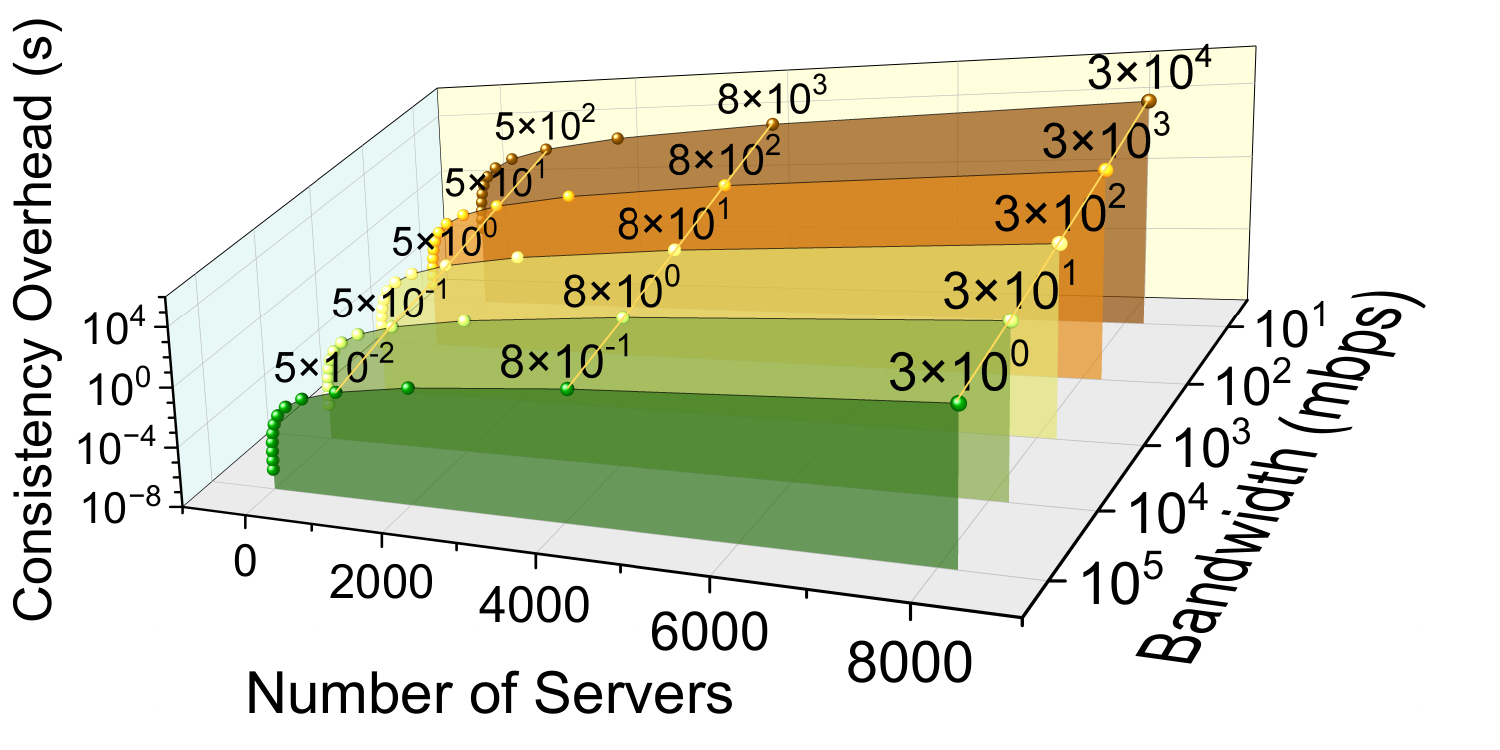}
        \captionsetup{font=footnotesize, justification=centering}
        \vspace{-0.5cm}
        \caption{Overhead of synchronization.}\label{fig:evaluation-deepdive-effect-synchronizationoverhead}
    \end{subfigure}
    ~~
    \begin{subfigure}[b]{0.2\textwidth}
        \centering
        \includegraphics[width=\textwidth,trim=0.35cm 17cm 5.05cm 1.5cm, clip]{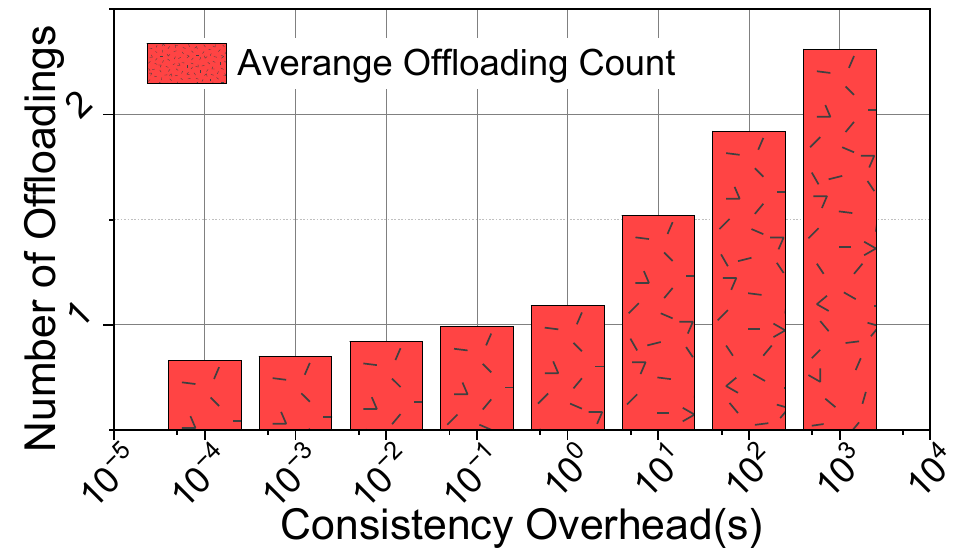}
        \captionsetup{font=footnotesize, justification=centering}
        \vspace{-0.5cm}
        \caption{Impact of synchronization.}\label{fig:evaluation-deepdive-effect-synchronizationoffload}
    \end{subfigure}
    \caption{Effect of \EPARAbf handling-synchronization-placement system.}\label{fig:evaluation-deepdive-epara system}
\end{figure*}

\begin{figure*}[!t]
    \centering
    \begin{subfigure}[b]{0.2\textwidth}
        \centering
        \includegraphics[width=\textwidth,trim=0cm 0.5cm 6cm 0.3cm, clip]{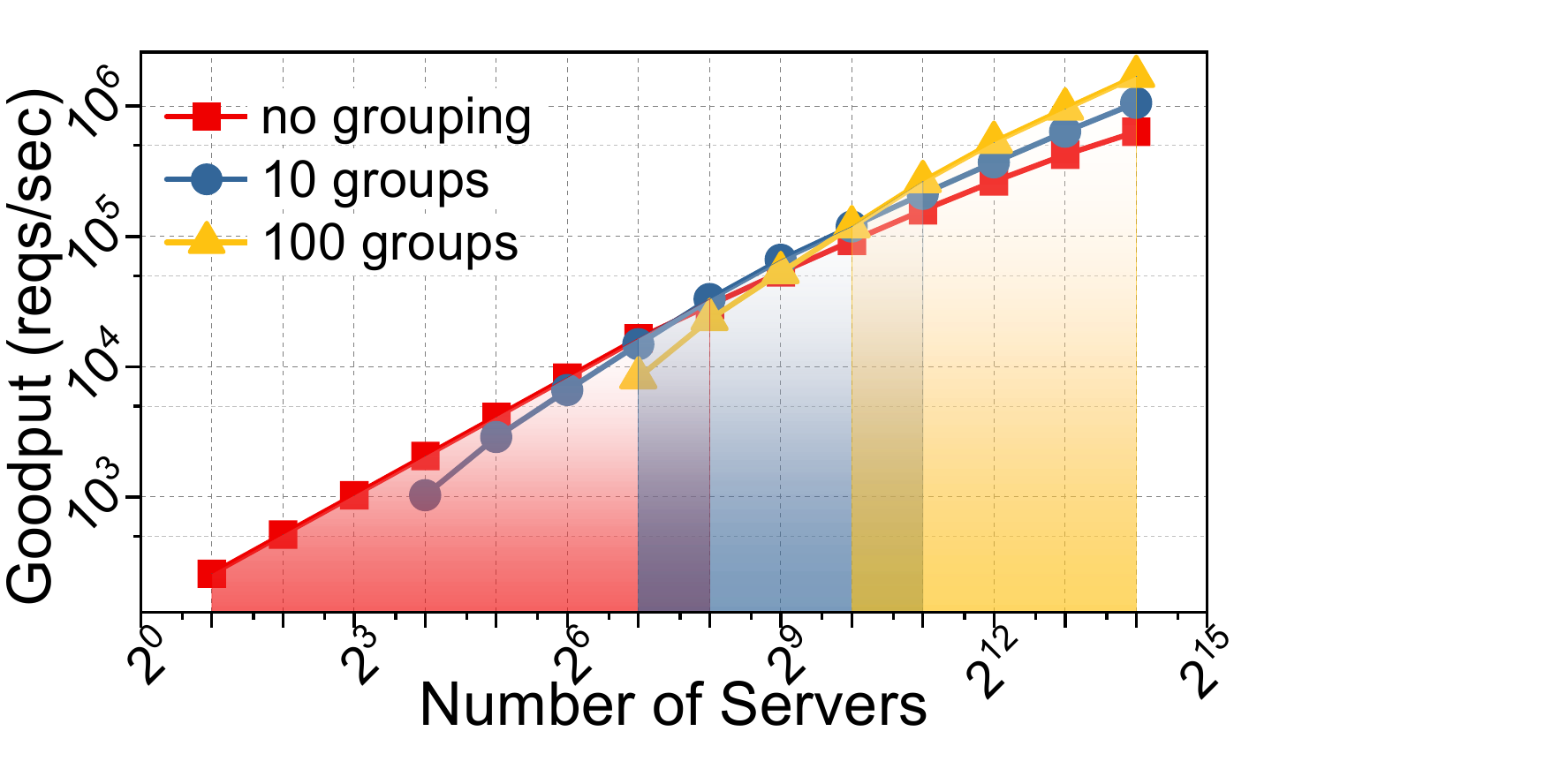}
        \captionsetup{font=footnotesize, justification=centering}
        \vspace{-0.5cm}
        \caption{Goodput of abundant servers.}
        \label{fig:evaluation-deepdive-extreme-servergoodput}
    \end{subfigure}
    ~~
    \begin{subfigure}[b]{0.2\textwidth}
        \centering
        \includegraphics[width=\textwidth,trim=0cm 0.5cm 6cm 0cm, clip]{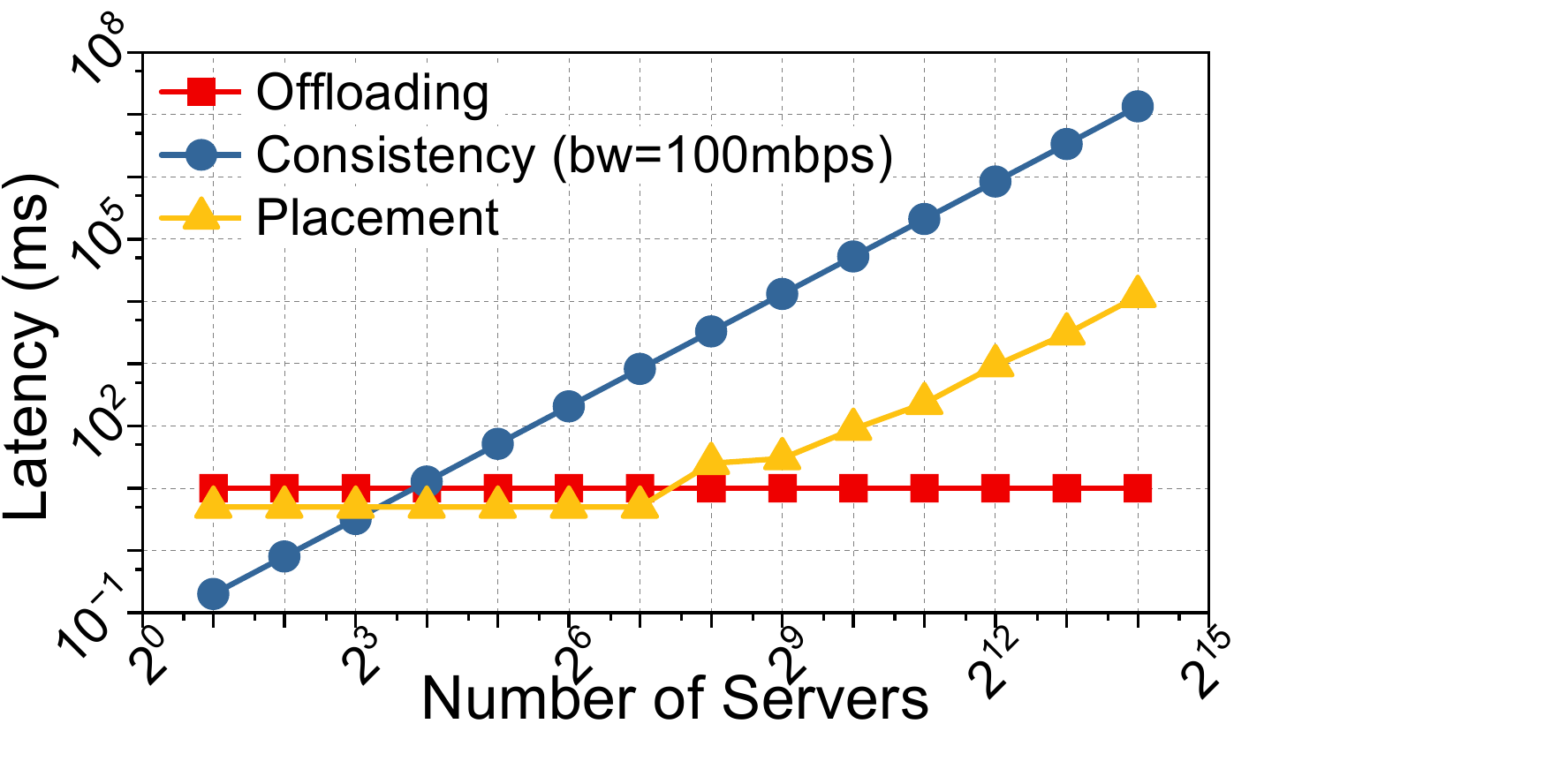}
        \captionsetup{font=footnotesize, justification=centering}
        \vspace{-0.5cm}
        \caption{Latency of abundant servers.}
        \label{fig:evaluation-deepdive-extreme-serverlatency}
    \end{subfigure}
    ~~
    \begin{subfigure}[b]{0.2\textwidth}
        \centering
        \includegraphics[width=\textwidth,trim=0.7cm 17.4cm 4.75cm 1.1cm, clip]{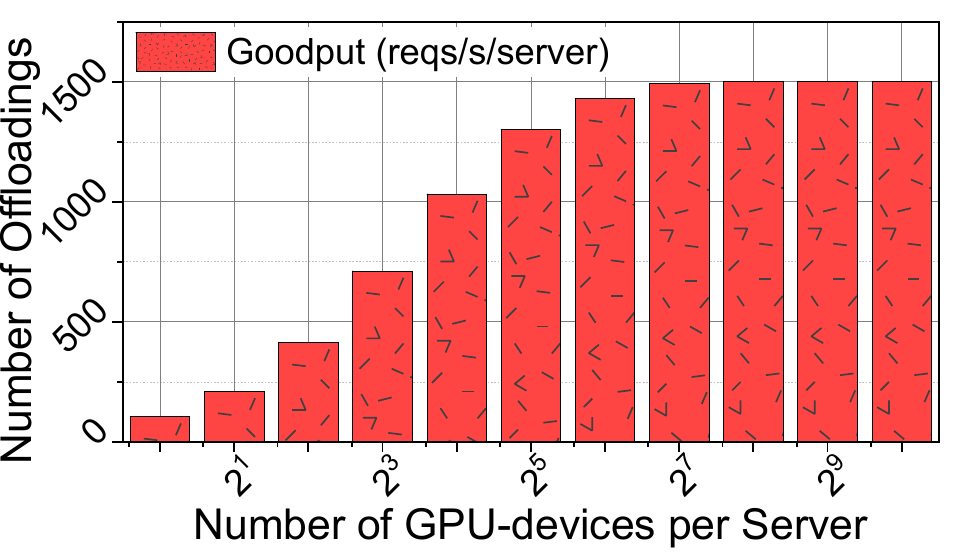}
        \captionsetup{font=footnotesize, justification=centering}
        \vspace{-0.5cm}
        \caption{Goodput of abundant devices.}
        \label{fig:evaluation-deepdive-extreme-devicegoodput}
    \end{subfigure}
    ~~
    \begin{subfigure}[b]{0.2\textwidth}
        \centering
        \includegraphics[width=\textwidth,trim=0cm 0.3cm 5.8cm 0.2cm, clip]{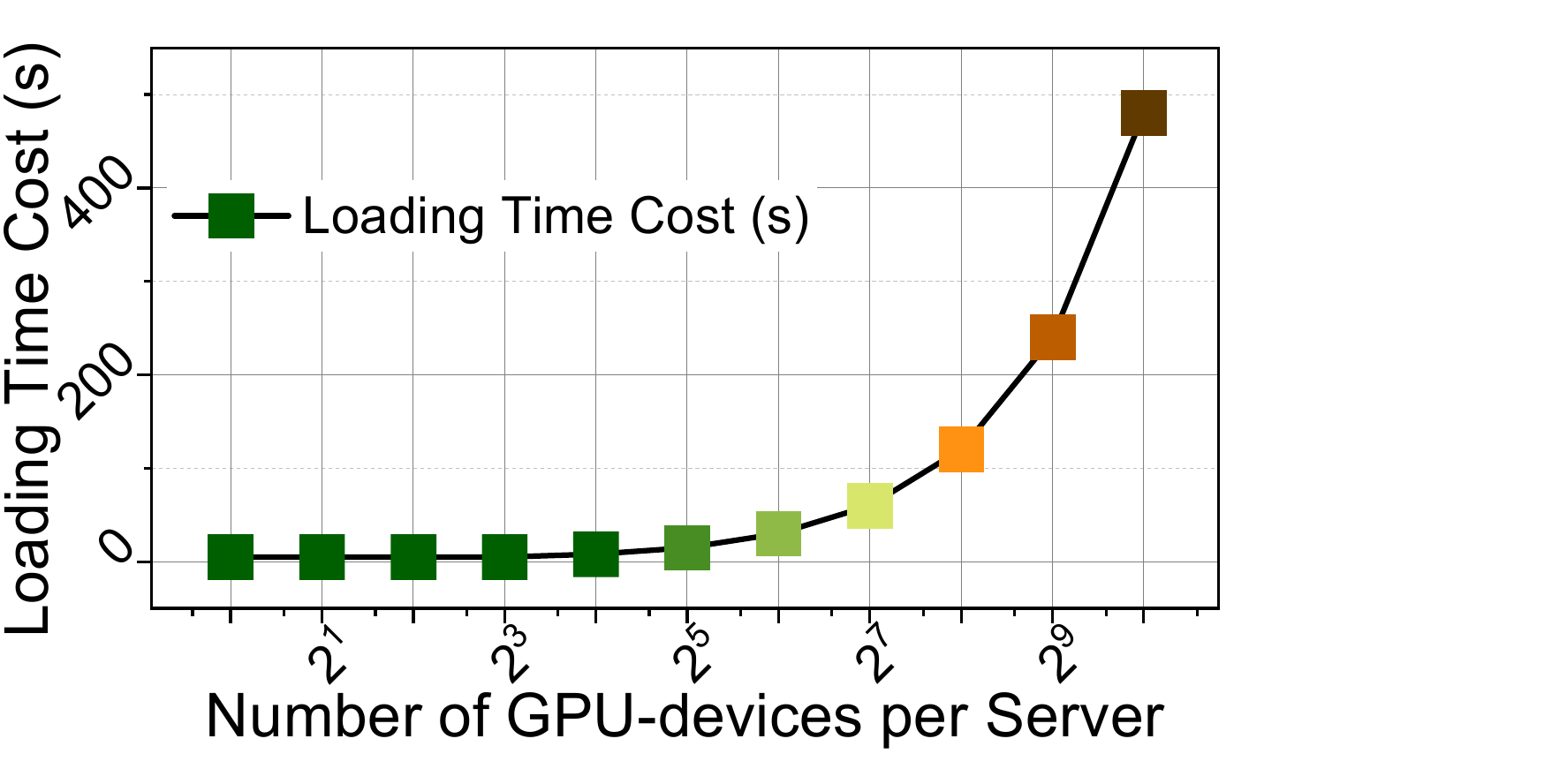}
        \captionsetup{font=footnotesize, justification=centering}
        \vspace{-0.5cm}
        \caption{Loading latency of devices.}
        \label{fig:evaluation-deepdive-extreme-devicelatency}
    \end{subfigure}
    ~~
    \begin{subfigure}[b]{0.2\textwidth}
        \centering
        \includegraphics[width=\textwidth,trim=0.68cm 17.2cm 4.65cm 1cm, clip]{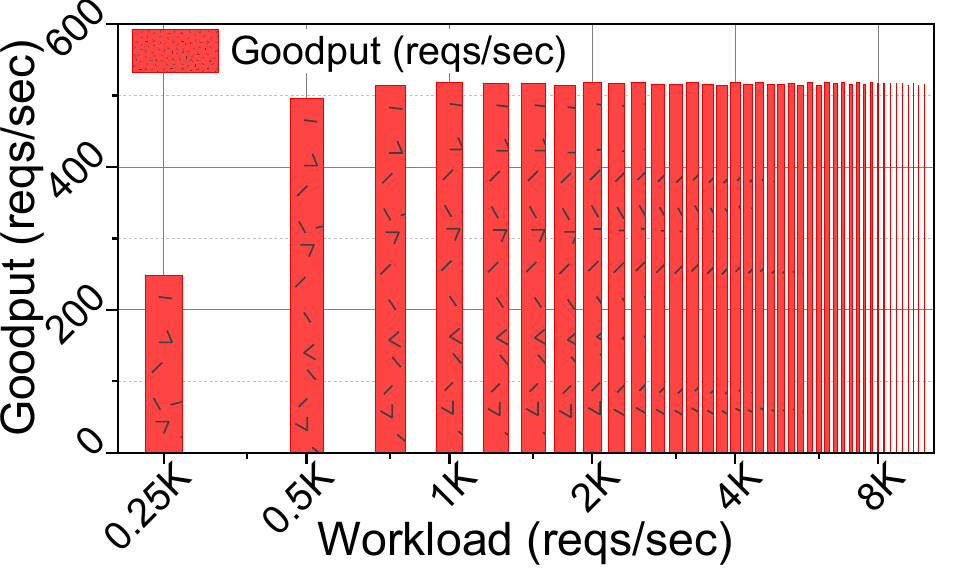}
        \captionsetup{font=footnotesize, justification=centering}
        \vspace{-0.5cm}
        \caption{GPU-sparse system.}
        \label{fig:evaluation-deepdive-extreme-gpusparse}
    \end{subfigure}
    \caption{Handling some extreme cases in \EPARAbf.}
    \label{fig:evaluation-deepdive-extreme cases}
\end{figure*}

\begin{figure}[!t]
    \vspace{-0.5pt}
    \centering
    \begin{subfigure}[thbp]{0.24\textwidth}
        \centering
        \includegraphics[width=\textwidth,trim=0cm 0.10cm 3.5cm 0.5cm, clip]{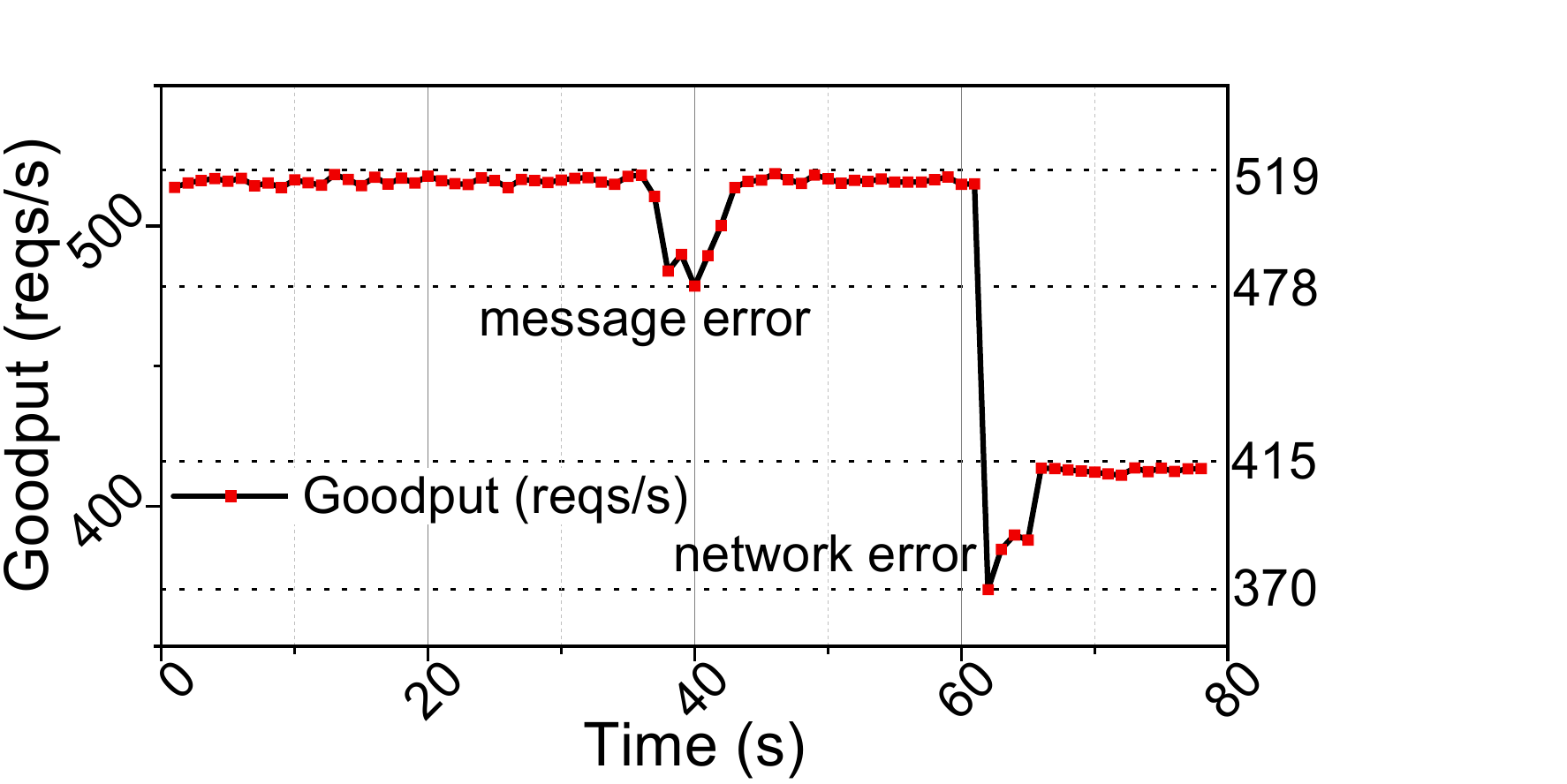}
        \captionsetup{font=footnotesize, justification=centering}
        \vspace{-0.6cm}
        \caption{Synchronization error.}
        \label{fig:evaluation-deepdive-sensitivity-synchronization}
    \end{subfigure}
    ~~
    \centering
    \begin{subfigure}[thbp]{0.24\textwidth}
        \centering
		\includegraphics[width=\textwidth,trim=0cm 0.25cm 3.5cm 0.5cm, clip]{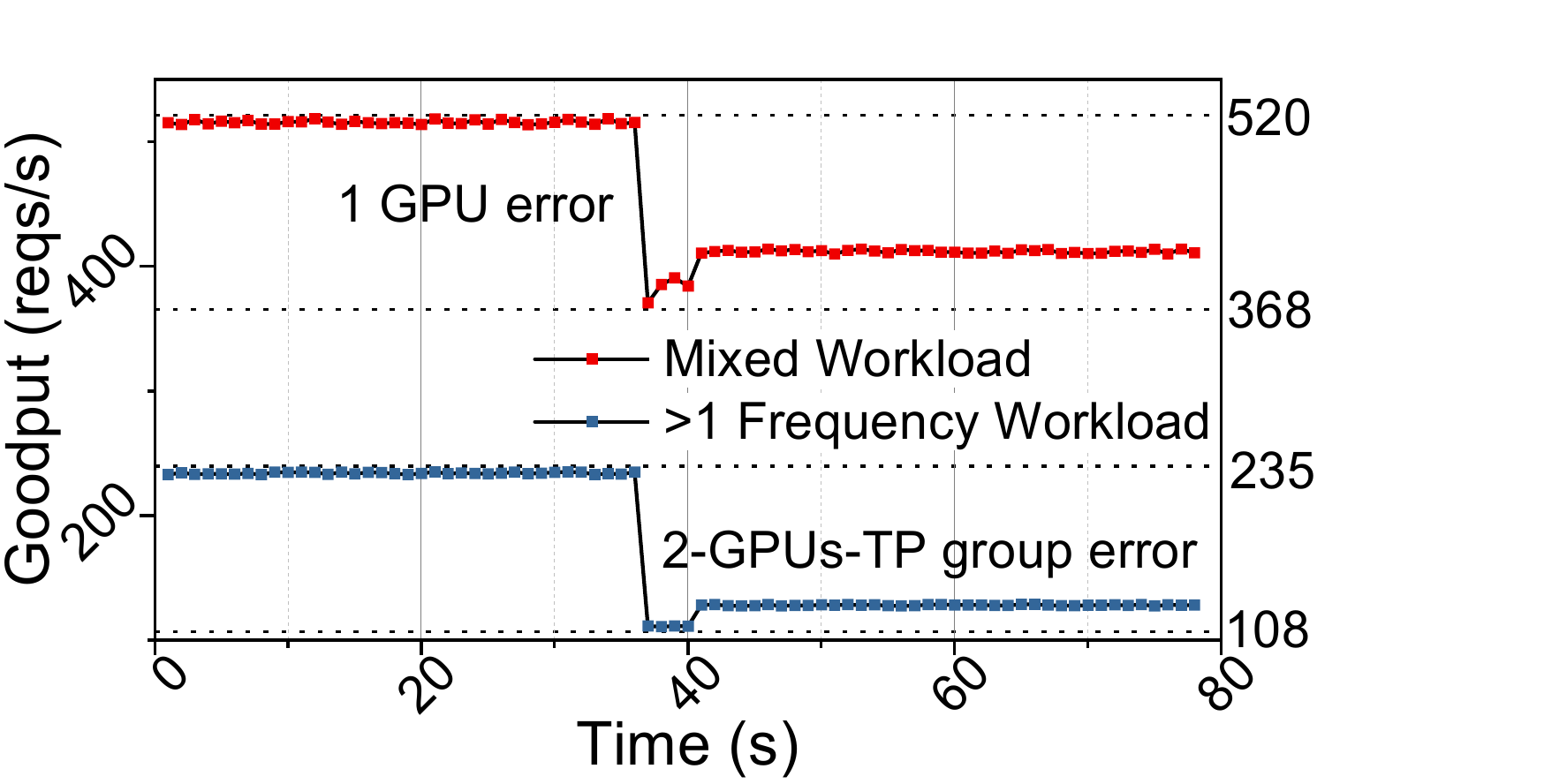}
        \captionsetup{font=footnotesize, justification=centering}
        \vspace{-0.6cm}
        \caption{Hardware error.}
        \label{fig:evaluation-deepdive-sensitivity-hardware}
    \end{subfigure}
    \caption{Sensitivity and error handling in real testbed.}\label{fig:evaluation-deepdive-sensitivity}
    \vspace{10pt}
\end{figure}

\noindent{\bf Overhead and impact of information synchronization:}
First, regarding the overhead of information synchronization in \cref{section:design-decentre}, we experiment with the synchronization delay under various bandwidth and server count conditions. On a small scale, we conduct tests using testbed servers, and on a large scale, we simulate with bandwidth and data volume. 
As shown in Fig.~\ref{fig:evaluation-deepdive-effect-synchronizationoverhead}, the information synchronization delay of \EPARA is always within 10s when (bandwith, servers count) equal to (50mbps, 100) or (500mbps, 1000). 
Second, concerning the request handling effect discussed in~\cref{section:design-handling}, \EPARA predetermines service placement during the offloading process. In this context, offloading precision primarily depends on the effectiveness of scheduling strategies and the delay of information sync.
Prolonged sync intervals may lead to information obsolescence, thereby increasing offloading instances. 
Testbed experimental verification demonstrates that \EPARA maintains an average offloading count below 1 when sync overhead remains shorter than 100ms, as illustrated in Fig.~\ref{fig:evaluation-deepdive-effect-synchronizationoffload}. Notably, offloading count exhibits a notable upward trend as sync overhead increases.

\subsubsection{Handling Some Extreme Cases}\label{section:
evaluation-deepdive-extreme} 
~\\
%
\noindent{\bf Scalability for abundant servers:}
We evaluate \EPARA's performance in large-scale server deployments through simulation experiments. First, regarding system service capacity, as shown in Fig.~\ref{fig:evaluation-deepdive-extreme-servergoodput}, we observe that \EPARA falls short in maintaining linear service capacity improvement when the number of edge servers continues growing beyond a certain threshold, and we investigate the underlying causes through latency analysis. Second, in terms of latency, Fig.~\ref{fig:evaluation-deepdive-extreme-serverlatency} demonstrates that while request handling maintains satisfactory latency performance, both information synchronization and service placement exhibit significantly increased delays. Since service placement and information synchronization do not directly impact request handling workflows, we posit that moderate latency growth remains acceptable within experimental parameters, whereas excessive delays would compromise offloading precision due to information obsolescence. Third, as illustrated in Fig.~\ref{fig:evaluation-deepdive-extreme-servergoodput}, we identify that a straightforward grouping strategy ensuring 100-500 servers per information exchange group effectively enables \EPARA's scalability in large-server deployments.

\noindent{\bf Scalability of device-saturated system:}
We validate the performance of \EPARA when edge servers manage excessive GPU-enabled edge devices registered for computational offloading. In the testbed configuration, a single edge device is configured to send multiple registration requests to the target edge server for experimental purposes. To address the potential network resource contention caused by device overload, bandwidth constraints are implemented for edge device management. 
First, regarding the overall throughput of edge servers and devices, as shown in Fig.~\ref{fig:evaluation-deepdive-extreme-devicegoodput}, both registration time utilization and GPU utilization rates demonstrate decreases due to increased task deployment latency, while the system maintain throughput stability and service availability. 
Second, concerning the task allocation speed to edge devices, Fig.~\ref{fig:evaluation-deepdive-extreme-devicelatency} illustrates the latency distribution from registration request submission to task assignment reception of edge devices. When concurrent device registrations exceed critical thresholds, service deployment enter queuing states, resulting in loading latency escalation.
For device-saturated networks, potential improvements may involve optimizing inter-device collaboration~\cite{zhangEdgeShardEfficientLLM2024}. However, this would introduce direct device-to-device communication and additional scenarios such as device mobility, which fall outside the scope of \EPARA's current framework.

\noindent{\bf Working with GPU-sparse edge system:}
We evaluate \EPARA's performance through testbed experiments when GPU resources are constrained in edge clouds, subjecting the system to request loads 10 times exceeding its computational capacity. As shown in Fig.~\ref{fig:evaluation-deepdive-extreme-gpusparse} regarding service capacity assessment, \EPARA still successfully fulfills the maximum feasible number of user requests without exhibiting throughput degradation.

\subsubsection{Sensitivity and Error Handling.}
~\\
\noindent{\bf Handling wrong synchronization message:}
To address synchronization errors, including undetected information error and detected information loss, we evaluate \EPARA's resilience mechanisms. First, for undetected information errors (e.g., silent data errors), \EPARA passively resolves them with automatic correction during subsequent synchronization cycles. As shown in Fig.~\ref{fig:evaluation-deepdive-sensitivity-synchronization}, such errors only marginally increase offloading counts within the affected synchronization cycle, with negligible impact on overall system throughput. 
Second, for detected information loss (e.g., unresponsiveness), \EPARA actively bypasses faulty servers and propagates synchronization requests to subsequent nodes. The faulty server is flagged as unavailable until manual intervention is performed. Fig.~\ref{fig:evaluation-deepdive-sensitivity-synchronization} demonstrates \EPARA's capability to isolate such faults while maintaining serving continuity.

\noindent{\bf Handling server error:}
To address hardware (e.g., GPU malfunctions) or software errors when serving tasks, \EPARA employs a fault containment strategy. Specifically, when a GPU failure occurs, \EPARA immediately terminates the service of faulty GPU and any GPUs operating in parallel with it. During subsequent service offloading, edge servers controlling the faulty GPU and its associated GPUs automatically exclude their resources from service placement until manual intervention is performed. As illustrated in Fig.~\ref{fig:evaluation-deepdive-sensitivity-hardware}, this mechanism can prevent fault propagation throughout the system during error-handling workflows.

\subsubsection{Case Study 2: Segmentation in \EPARAit}\label{appendix-implement-case_study-segementation}
~\\
\noindent Since we propose a task-categorized parallelism framework supporting heterogeneous AI services with request/service-level allocation, a case study of semantic/instance segmentation is provided to illustrate \EPARA's operation process.  

\textit{Why measuring segmentation models?}  
First, their applications incorporate both latency-sensitive (image) and frequency-sensitive (video, e.g., autonomous driving) scenarios. Second, their computational demands range from lightweight (UNet~\cite{ronnebergerUNetConvolutionalNetworks2015}) to heavy (OMG-Seg~\cite{liOMGSegOneModel2024}), spanning all four categories defined in \cref{section:design-parallelism}. 
Implementations are summarized in Table~\ref{table:appendix-case_study-segementation}.

For settings, we simulate heterogeneous device operators by deploying four categories of services on four servers, each equipped with an Tesla P100 GPU. The services include 1080P image segmentation and 60fps 1080P video segmentation.

Following the methodology in \cref{section:implement-pretraining}, \EPARA conducts adaptive deployment strategies for segmentation models.  
First, We focus on MP and BS for these segmentation models. 
For UNet, DeepLabV3+, SCTNet, MaskFormer, and OMGSeg in latency-sensitive scenarios, \EPARA adopts BS8, BS4, BS4, TP2+BS8, and TP2+BS4, respectively, to reduce inference latency. For UNet, DeepLabV3+, and SCTNet in frequency-sensitive scenarios, \EPARA adopts BS8, BS4, and BS4, respectively, to ensure frame rates.
Second, \EPARA measures the MT of these tasks, and keeps MT equal to 1.
Third, \EPARA focuses on DP and MF strategies for video applications.
For Unet, \EPARA adopts MF4 to ensure the tolerable inter-frame latency. For DeepLabV3+ and SCTNet, \EPARA adopts MF4+DP2 for both services to achieve approximately 60 fps.

We perform segmentation evaluations on \EPARA in a real-world cluster. Results demonstrate that \EPARA better meets the SLOs targets and improves average GPU goodput, as illustrated in Fig.~\ref{fig:implementation-case_study}.

\begin{table}[!t]
    \centering
    \setlength{\abovecaptionskip}{6pt}
    \setlength{\belowcaptionskip}{-10pt}
    \resizebox{1\linewidth}{!} {
    \begin{tabular}{|c||c|c|}
    \hline 
    Category & latency-sensitive (picture) & frequency-sensitive (video) \\
    \hline
    \hline
    $\leq$ 1 GPU & Unet, DeeplabV3+, SCTNet  & Unet \\
    \hline
    $\geq$ 1 GPU & MaskFormer, OMGSeg & DeeplabV3+, SCTNet\\
    \hline
    \end{tabular}
    }
    \vspace{10pt}
    \caption{Segmentation models of case study 2 in \EPARAbf.}\label{table:appendix-case_study-segementation}
    \vspace{-15pt}
\end{table}
\begin{figure}[!t]
    \begin{subfigure}[b]{0.24\textwidth}
        \centering
        \includegraphics[width=0.49\textwidth,trim=0.5cm 14.3cm 12.2cm 1cm, clip]{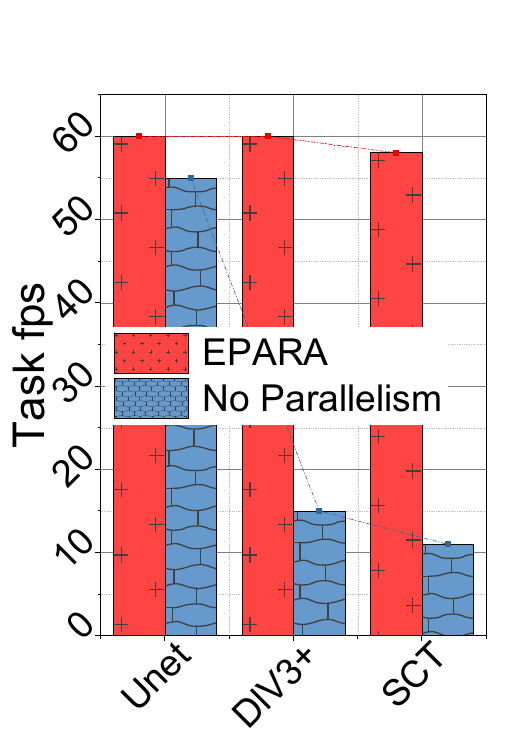}
        \includegraphics[width=0.49\textwidth,trim=0.5cm 14.3cm 12.2cm 1cm, clip]{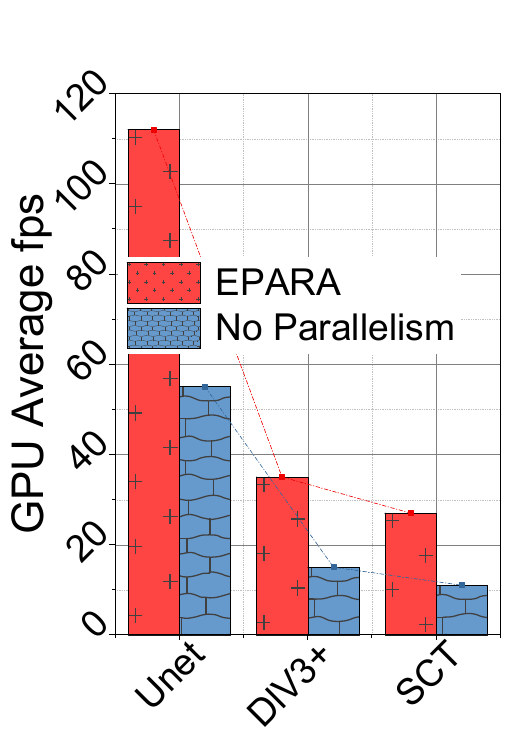}
        \captionsetup{font=footnotesize, justification=centering}
        \vspace{-0.5cm}
        \caption{Video segmentation.}\label{fig:implementation-video}
    \end{subfigure}
    ~~
    \centering
    \begin{subfigure}[b]{0.24\textwidth}
        \centering
        \includegraphics[width=0.49\textwidth,trim=0.5cm 14.3cm 12.2cm 1cm, clip]{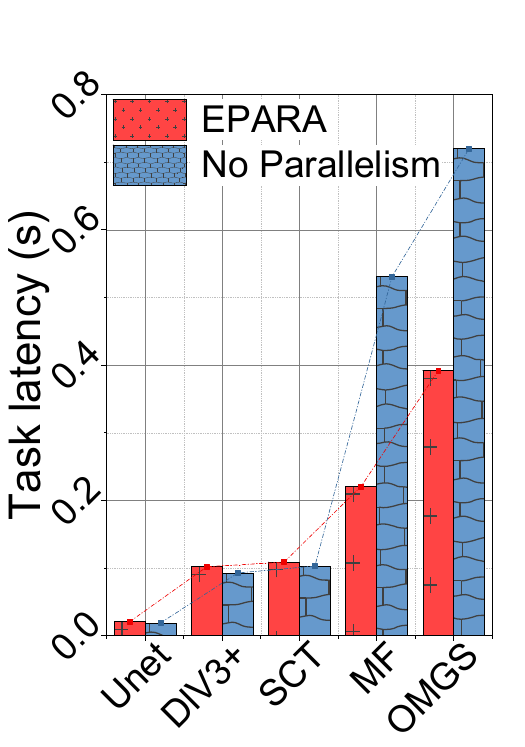}
        \includegraphics[width=0.49\textwidth,trim=0.5cm 14.3cm 12.2cm 1cm, clip]{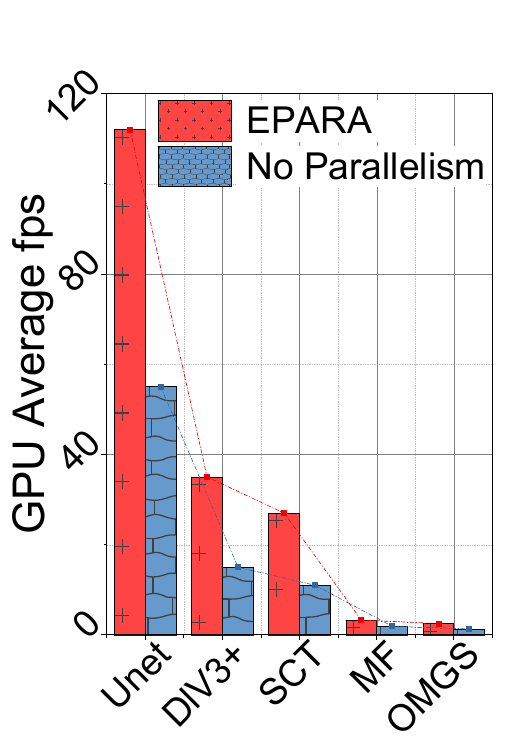}
        \captionsetup{font=footnotesize, justification=centering}
        \vspace{-0.5cm}
        \caption{Picture segmentation.}\label{fig:implementation-picture}
    \end{subfigure}
    \caption{Case study 2: segmentation in \EPARAbf.}\label{fig:implementation-case_study}
    \vspace{0.5cm}
\end{figure}
\section{Related Work}

\noindent Here we discuss the relevant schemes and summarize the differences between some of them and \EPARA as shown in Table~\ref{table:related}.
\begin{table}[!t]
    \begin{center}
    \resizebox{1\linewidth}{!} {
    \begin{tabular}{|p{0em}|c||c|c|c|c|} 
    \hline
    \multicolumn{2}{|c||}{\multirow{2}{*}{\textbf{Schemes}}}  & \multicolumn{2}{c|}{\textbf{Allocation Level}}& \multicolumn{2}{c|}{\textbf{Application}} \\
    \cline{3-6}

    \multicolumn{2}{|c||}{} & \textit{Request}& \textit{Service}& \textit{\textit{Mode}}& \textit{Scene} \\
    \hline
    \hline
    \multirow{5}{*}{\hspace{-0.15cm}\centering\smash{\rotatebox[origin=c]{90}{\textbf{Edge}}}}& InterEdge~\cite{brownArchitectureEdgeNetworking2024}& No & No & Distr. & Serving \\
    \cline{2-6} 
    &Galaxy~\cite{yeGalaxyResourceEfficientCollaborative2024a} & No & MP+ & Cent. & AI serving \\
    \cline{2-6} 
    &DeTransformer~\cite{weiCommunicationEfficientModelParallelism2024} & No & MP+ & Cent. & AI serving \\
    \cline{2-6} 
    &EdgeShared~\cite{zhangEdgeShardEfficientLLM2024} & No & MP+ & Cent. & LLM serving \\
    \cline{2-6} 
    &SERV-P~\cite{farhadiServicePlacementRequest2021} & No & No & Cent. & Serving \\
    \hline
    \multirow{5}{*}{\hspace{-0.10cm}\centering\smash{\rotatebox[origin=c]{90}{\textbf{Datacenter}}}}& AlpaServe\cite{liAlpaServeStatisticalMultiplexing2023}& No& MP+& Cent.& AI serving\\ 
    \cline{2-6} 
    & USHER~\cite{shubhaUSHERHolisticInterference2024} & No & MP+ &Cent. & AI serving\\
    \cline{2-6} 
    & SHEPHERD~\cite{zhangSHEPHERDServingDNNs2023}& No & MP & Cent.& AI serving \\
    \cline{2-6} 
    & Nexus~\cite{shenNexusGPUCluster2019}& Queue & MP & Cent. & Video AI serv.\\
    \cline{2-6} 
    & CASSINI~\cite{rajasekaranCASSININetworkAwareJob2024} & Network & 3D & Cent. & AI training\\
    \hline
    \multicolumn{2}{|c||}{\EPARAbf (ours)}& DP+MF & MP+ & Mixed & AI serving\\
    \hline
    \end{tabular}
    }
    \end{center}
    \vspace{10pt}
    \caption{Summary of prior works and comparison to \EPARAbf .}\label{table:related}
    \vspace{-10pt}
\end{table}

\noindent{\bf Edge AI inference:}
Galaxy~\cite{yeGalaxyResourceEfficientCollaborative2024a}  DeTransformer~\cite{weiCommunicationEfficientModelParallelism2024} are two edge devices systems for transformer-based AI with MP. 
CLIO~\cite{huangCLIOEnablingAutomatic2020} focuses on automatic model offloading for edge devices. 
Jupiter~\cite{shengyuanyeJupiterFastResourceEfficient2024} is an edge-devices LLMs inference system concentrating on prefill and decode phases.
EdgeShared~\cite{zhangEdgeShardEfficientLLM2024} is an LLM inference system, focusing on MP through edge devices. These works mainly discuss AI serving through edge devices under centralized management.

\noindent{\bf Edge network task scheduling:} 
InterEdge~\cite{brownArchitectureEdgeNetworking2024} is a novel architecture for edge networking. 
Linklab~\cite{dongLinkLab20Multitenant2023} is a multi-tenant edge computing system. 
TORS~\cite{liMultiHopTaskOffloading2025} is a multi-hop task offloading system for edge networks.
KubeEdge~\cite{kubeedge} is a Kubernetes-based edge computing system for universal tasks.

\noindent{\bf Edge AI training:}
Delta and Asteroid~\cite{gongDeltaCloudassistedData2024,yeAsteroidResourceEfficientHybrid2024} are edge AI training systems.
They focus on edge model training, which is concerned with data enrichment and collaboration between devices and cloud services. AdaInf~\cite{shubhaAdaInfDataDrift2023} concentrates on adaptive model training and parameter adjustment in edge clouds.

\noindent{\bf Edge tasks adaptation:} 
AdaptiveNet and EdgeML~\cite{wenAdaptiveNetPostdeploymentNeural2023,edgeml04} focus on AI models adaptation for edge devices.
AccuMO, Gemel, and Hairpin~\cite{kongAccuMOAccuracyCentricMultitask2023,padmanabhanGemelModelMerging2023,mengHairpinRethinkingPacket2024} are dedicated to video AI application adaptation.
REMIX and NeuLens\cite{jiangFlexibleHighresolutionObject2021,houNeuLensSpatialbasedDynamic2022} aim to the joint configuration of video AI models and applications.

\noindent{\bf ML serving system:} AlpaServe, USHER, SHEPHERD, and gpulet~\cite{liAlpaServeStatisticalMultiplexing2023,shubhaUSHERHolisticInterference2024,zhangSHEPHERDServingDNNs2023,choiServingHeterogeneousMachine2022} are machine learning serving systems in datacenter or cloud computing scenarios. DistServe, Parrot, Sarathi-Serve, FastServe~\cite{zhongDistServeDisaggregatingPrefill2024,linParrotEfficientServing2024,agrawalTamingThroughputLatencyTradeoff2024,wuFastDistributedInference2024} are LLMs serving systems. Nexus~\cite{shenNexusGPUCluster2019} focuses on video AI serving with queue management in cloud computing.

\section{Conclusion}
\noindent \EPARA is an edge AI inference framework that aims to enhance serving capabilities of edge computing systems. \EPARA achieves task-resource allocation by categorizing tasks based on latency/frequency sensitivity and GPU resource consumption. \EPARA includes three core components: task-resource allocation, request handling, and service placement. We implement a \EPARA prototype with adaptive deployment strategies and management protocols, and conduct two case studies. We evaluate \EPARA through real-world testbed experiments, large-scale simulations, and deepdive analyses. \EPARA proves to be a viable solution that fulfills our objectives.



\clearpage

\bibliographystyle{num-ACM-Reference-Format}
\bibliography{refs_ccf}


\begin{thebibliography}{87}


\ifx \showCODEN    \undefined \def \showCODEN     #1{\unskip}     \fi
\ifx \showDOI      \undefined \def \showDOI       #1{#1}\fi
\ifx \showISBNx    \undefined \def \showISBNx     #1{\unskip}     \fi
\ifx \showISBNxiii \undefined \def \showISBNxiii  #1{\unskip}     \fi
\ifx \showISSN     \undefined \def \showISSN      #1{\unskip}     \fi
\ifx \showLCCN     \undefined \def \showLCCN      #1{\unskip}     \fi
\ifx \shownote     \undefined \def \shownote      #1{#1}          \fi
\ifx \showarticletitle \undefined \def \showarticletitle #1{#1}   \fi
\ifx \showURL      \undefined \def \showURL       {\relax}        \fi
\providecommand\bibfield[2]{#2}
\providecommand\bibinfo[2]{#2}
\providecommand\natexlab[1]{#1}
\providecommand\showeprint[2][]{arXiv:#2}

\bibitem[Agrawal et~al\mbox{.}(2024)]%
        {agrawalTamingThroughputLatencyTradeoff2024}
\bibfield{author}{\bibinfo{person}{Amey Agrawal}, \bibinfo{person}{Nitin Kedia}, \bibinfo{person}{Ashish Panwar}, \bibinfo{person}{Jayashree Mohan}, \bibinfo{person}{Nipun Kwatra}, {et~al\mbox{.}}} \bibinfo{year}{2024}\natexlab{}.
\newblock \showarticletitle{Taming \{\vphantom\}{{Throughput-Latency}}\vphantom\{\} {{Tradeoff}} in \{\vphantom\}{{LLM}}\vphantom\{\} {{Inference}} with \{\vphantom\}{{Sarathi-Serve}}\vphantom\{\}}. In \bibinfo{booktitle}{\emph{OSDI}}. \bibinfo{pages}{117--134}.
\newblock
\showISBNx{978-1-939133-40-3}


\bibitem[Ashkboos et~al\mbox{.}(2024)]%
        {ashkboosQuaRotOutlierFree4Bit2024}
\bibfield{author}{\bibinfo{person}{Saleh Ashkboos}, \bibinfo{person}{Amirkeivan Mohtashami}, \bibinfo{person}{Maximilian~L. Croci}, \bibinfo{person}{Bo Li}, \bibinfo{person}{Pashmina Cameron}, {et~al\mbox{.}}} \bibinfo{year}{2024}\natexlab{}.
\newblock \bibinfo{title}{{{QuaRot}}: {{Outlier-Free}} 4-{{Bit Inference}} in {{Rotated LLMs}}}.
\newblock
\urldef\tempurl%
\url{http://arxiv.org/abs/2404.00456}
\showURL{%
\tempurl}


\bibitem[{AUTOSAR}(2025)]%
        {autosarAUTOSARAutomotiveOpen2025}
\bibfield{author}{\bibinfo{person}{{AUTOSAR}}.} \bibinfo{year}{2025}\natexlab{}.
\newblock \bibinfo{title}{{{AUTOSAR}} ({{Automotive Open System Architecture}})}.
\newblock
\urldef\tempurl%
\url{https://www.autosar.org/}
\showURL{%
\tempurl}


\bibitem[Brown et~al\mbox{.}(2024)]%
        {brownArchitectureEdgeNetworking2024}
\bibfield{author}{\bibinfo{person}{Lloyd Brown}, \bibinfo{person}{Emily Marx}, \bibinfo{person}{Dev Bali}, \bibinfo{person}{Emmanuel Amaro}, \bibinfo{person}{Debnil Sur}, {et~al\mbox{.}}} \bibinfo{year}{2024}\natexlab{}.
\newblock \showarticletitle{An {{Architecture For Edge Networking Services}}}. In \bibinfo{booktitle}{\emph{SIGCOMM}}. \bibinfo{pages}{645--660}.
\newblock
\showISBNx{979-8-4007-0614-1}


\bibitem[Brown et~al\mbox{.}(2020)]%
        {brownLanguageModelsAre2020}
\bibfield{author}{\bibinfo{person}{Tom~B. Brown}, \bibinfo{person}{Benjamin Mann}, \bibinfo{person}{Nick Ryder}, \bibinfo{person}{Melanie Subbiah}, \bibinfo{person}{Jared Kaplan}, {et~al\mbox{.}}} \bibinfo{year}{2020}\natexlab{}.
\newblock \bibinfo{title}{Language {{Models}} Are {{Few-Shot Learners}}}.
\newblock
\urldef\tempurl%
\url{http://arxiv.org/abs/2005.14165}
\showURL{%
\tempurl}


\bibitem[Chen and Ran(2019)]%
        {chenDeepLearningEdge2019}
\bibfield{author}{\bibinfo{person}{Jiasi Chen} {and} \bibinfo{person}{Xukan Ran}.} \bibinfo{year}{2019}\natexlab{}.
\newblock \showarticletitle{Deep {{Learning With Edge Computing}}: {{A Review}}}.
\newblock \bibinfo{journal}{\emph{Proc. IEEE}} \bibinfo{volume}{107}, \bibinfo{number}{8} (\bibinfo{year}{2019}), \bibinfo{pages}{1655--1674}.
\newblock
\showISSN{0018-9219, 1558-2256}


\bibitem[Chen et~al\mbox{.}(2018)]%
        {chenEncoderDecoderAtrousSeparable2018}
\bibfield{author}{\bibinfo{person}{Liang-Chieh Chen}, \bibinfo{person}{Yukun Zhu}, \bibinfo{person}{George Papandreou}, \bibinfo{person}{Florian Schroff}, {and} \bibinfo{person}{Hartwig Adam}.} \bibinfo{year}{2018}\natexlab{}.
\newblock \showarticletitle{Encoder-{{Decoder}} with {{Atrous Separable Convolution}} for {{Semantic Image Segmentation}}}. In \bibinfo{booktitle}{\emph{ECCV}}. \bibinfo{pages}{801--818}.
\newblock


\bibitem[Chen et~al\mbox{.}(2024)]%
        {chenMAGISMemoryOptimization2024}
\bibfield{author}{\bibinfo{person}{Renze Chen}, \bibinfo{person}{Zijian Ding}, \bibinfo{person}{Size Zheng}, \bibinfo{person}{Chengrui Zhang}, \bibinfo{person}{Jingwen Leng}, {et~al\mbox{.}}} \bibinfo{year}{2024}\natexlab{}.
\newblock \showarticletitle{{{MAGIS}}: {{Memory Optimization}} via {{Coordinated Graph Transformation}} and {{Scheduling}} for {{DNN}}}. In \bibinfo{booktitle}{\emph{ASPLOS}}, Vol.~\bibinfo{volume}{3}. \bibinfo{pages}{607--621}.
\newblock
\showISBNx{979-8-4007-0386-7}


\bibitem[Cheng et~al\mbox{.}(2022)]%
        {chengMaskedAttentionMaskTransformer2022}
\bibfield{author}{\bibinfo{person}{Bowen Cheng}, \bibinfo{person}{Ishan Misra}, \bibinfo{person}{Alexander~G. Schwing}, \bibinfo{person}{Alexander Kirillov}, {and} \bibinfo{person}{Rohit Girdhar}.} \bibinfo{year}{2022}\natexlab{}.
\newblock \showarticletitle{Masked-{{Attention Mask Transformer}} for {{Universal Image Segmentation}}}. In \bibinfo{booktitle}{\emph{CVPR}}. \bibinfo{pages}{1290--1299}.
\newblock


\bibitem[Cheng et~al\mbox{.}(2021)]%
        {NEURIPS2021_950a4152}
\bibfield{author}{\bibinfo{person}{Bowen Cheng}, \bibinfo{person}{Alex Schwing}, {and} \bibinfo{person}{Alexander Kirillov}.} \bibinfo{year}{2021}\natexlab{}.
\newblock \showarticletitle{Per-Pixel Classification Is Not All You Need for Semantic Segmentation}. In \bibinfo{booktitle}{\emph{NeurIPS}}, Vol.~\bibinfo{volume}{34}. \bibinfo{pages}{17864--17875}.
\newblock


\bibitem[Choi et~al\mbox{.}(2022)]%
        {choiServingHeterogeneousMachine2022}
\bibfield{author}{\bibinfo{person}{Seungbeom Choi}, \bibinfo{person}{Sunho Lee}, \bibinfo{person}{Yeonjae Kim}, \bibinfo{person}{Jongse Park}, \bibinfo{person}{Youngjin Kwon}, {et~al\mbox{.}}} \bibinfo{year}{2022}\natexlab{}.
\newblock \showarticletitle{Serving {{Heterogeneous Machine Learning Models}} on \{\vphantom\}{{Multi-GPU}}\vphantom\{\} {{Servers}} with \{\vphantom\}{{Spatio-Temporal}}\vphantom\{\} {{Sharing}}}. In \bibinfo{booktitle}{\emph{USENIX ATC}}. \bibinfo{pages}{199--216}.
\newblock


\bibitem[Chu et~al\mbox{.}(2013)]%
        {chuRFC6928Increasing2013}
\bibfield{author}{\bibinfo{person}{J. Chu}, \bibinfo{person}{N. Dukkipati}, \bibinfo{person}{Y. Cheng}, {and} \bibinfo{person}{M. Mathis}.} \bibinfo{year}{2013}\natexlab{}.
\newblock \bibinfo{booktitle}{\emph{{{RFC}} 6928: {{Increasing TCP}}'s {{Initial Window}}}}.
\newblock


\bibitem[{DeepSeek-AI} et~al\mbox{.}(2024a)]%
        {deepseek-aiDeepSeekV2StrongEconomical2024}
\bibfield{author}{\bibinfo{person}{{DeepSeek-AI}}, \bibinfo{person}{Aixin Liu}, \bibinfo{person}{Bei Feng}, \bibinfo{person}{Bin Wang}, \bibinfo{person}{Bingxuan Wang}, {et~al\mbox{.}}} \bibinfo{year}{2024}\natexlab{a}.
\newblock \bibinfo{title}{{{DeepSeek-V2}}: {{A Strong}}, {{Economical}}, and {{Efficient Mixture-of-Experts Language Model}}}.
\newblock
\urldef\tempurl%
\url{http://arxiv.org/abs/2405.04434}
\showURL{%
\tempurl}


\bibitem[{DeepSeek-AI} et~al\mbox{.}(2024b)]%
        {deepseek-aiDeepSeekV3TechnicalReport2024}
\bibfield{author}{\bibinfo{person}{{DeepSeek-AI}}, \bibinfo{person}{Aixin Liu}, \bibinfo{person}{Bei Feng}, \bibinfo{person}{Bing Xue}, \bibinfo{person}{Bingxuan Wang}, {et~al\mbox{.}}} \bibinfo{year}{2024}\natexlab{b}.
\newblock \bibinfo{title}{{{DeepSeek-V3 Technical Report}}}.
\newblock
\urldef\tempurl%
\url{http://arxiv.org/abs/2412.19437}
\showURL{%
\tempurl}


\bibitem[Dennis et~al\mbox{.}(2025)]%
        {edgeml04}
\bibfield{author}{\bibinfo{person}{Don~Kurian Dennis}, \bibinfo{person}{Yash Gaurkar}, \bibinfo{person}{Sridhar Gopinath}, \bibinfo{person}{Sachin Goyal}, \bibinfo{person}{Chirag Gupta}, {et~al\mbox{.}}} \bibinfo{year}{2025}\natexlab{}.
\newblock \bibinfo{title}{{{EdgeML}}: {{Machine Learning}} for Resource-Constrained Edge Devices}.
\newblock
\urldef\tempurl%
\url{https://github.com/Microsoft/EdgeML}
\showURL{%
\tempurl}


\bibitem[Devlin et~al\mbox{.}(2019)]%
        {devlinBERTPretrainingDeep2019}
\bibfield{author}{\bibinfo{person}{Jacob Devlin}, \bibinfo{person}{Ming-Wei Chang}, \bibinfo{person}{Kenton Lee}, {and} \bibinfo{person}{Kristina Toutanova}.} \bibinfo{year}{2019}\natexlab{}.
\newblock \showarticletitle{{{BERT}}: {{Pre-training}} of {{Deep Bidirectional Transformers}} for {{Language Understanding}}}. In \bibinfo{booktitle}{\emph{NAACL}}. \bibinfo{pages}{4171--4186}.
\newblock


\bibitem[Devoto et~al\mbox{.}(2025)]%
        {devotoAdaptiveSemanticToken2025}
\bibfield{author}{\bibinfo{person}{Alessio Devoto}, \bibinfo{person}{Jary Pomponi}, \bibinfo{person}{Mattia Merluzzi}, \bibinfo{person}{Paolo~Di Lorenzo}, {and} \bibinfo{person}{Simone Scardapane}.} \bibinfo{year}{2025}\natexlab{}.
\newblock \bibinfo{title}{Adaptive {{Semantic Token Communication}} for {{Transformer-based Edge Inference}}}.
\newblock
\urldef\tempurl%
\url{http://arxiv.org/abs/2505.17604}
\showURL{%
\tempurl}


\bibitem[Dong et~al\mbox{.}(2023)]%
        {dongLinkLab20Multitenant2023}
\bibfield{author}{\bibinfo{person}{Wei Dong}, \bibinfo{person}{Borui Li}, \bibinfo{person}{Haoyu Li}, \bibinfo{person}{Hao Wu}, \bibinfo{person}{Kaijie Gong}, {et~al\mbox{.}}} \bibinfo{year}{2023}\natexlab{}.
\newblock \showarticletitle{\{\vphantom\}{{LinkLab}}\vphantom\{\} 2.0: {{A Multi-tenant Programmable}} \{\vphantom\}{{IoT}}\vphantom\{\} {{Testbed}} for {{Experimentation}} with \{\vphantom\}{{Edge-Cloud}}\vphantom\{\} {{Integration}}}. In \bibinfo{booktitle}{\emph{NSDI}}. \bibinfo{pages}{1683--1699}.
\newblock
\showISBNx{978-1-939133-33-5}


\bibitem[Farhadi et~al\mbox{.}(2021)]%
        {farhadiServicePlacementRequest2021}
\bibfield{author}{\bibinfo{person}{Vajiheh Farhadi}, \bibinfo{person}{Fidan Mehmeti}, \bibinfo{person}{Ting He}, \bibinfo{person}{Thomas F.~La Porta}, \bibinfo{person}{Hana Khamfroush}, {et~al\mbox{.}}} \bibinfo{year}{2021}\natexlab{}.
\newblock \showarticletitle{Service {{Placement}} and {{Request Scheduling}} for {{Data-Intensive Applications}} in {{Edge Clouds}}}.
\newblock \bibinfo{journal}{\emph{IEEE/ACM Transactions on Networking}} \bibinfo{volume}{29}, \bibinfo{number}{2} (\bibinfo{year}{2021}), \bibinfo{pages}{779--792}.
\newblock
\showISSN{1063-6692, 1558-2566}


\bibitem[Fisher et~al\mbox{.}(1978)]%
        {fisherAnalysisApproximationsMaximizing1978}
\bibfield{author}{\bibinfo{person}{M.~L. Fisher}, \bibinfo{person}{G.~L. Nemhauser}, {and} \bibinfo{person}{L.~A. Wolsey}.} \bibinfo{year}{1978}\natexlab{}.
\newblock \showarticletitle{An Analysis of Approximations for Maximizing Submodular Set Functions---{{II}}}.
\newblock In \bibinfo{booktitle}{\emph{Polyhedral {{Combinatorics}}}}. Vol.~\bibinfo{volume}{8}. \bibinfo{pages}{73--87}.
\newblock
\showISBNx{978-3-642-00789-7 978-3-642-00790-3}


\bibitem[Gibiansky(2017)]%
        {gibiansky2017bringing}
\bibfield{author}{\bibinfo{person}{Andrew Gibiansky}.} \bibinfo{year}{2017}\natexlab{}.
\newblock \showarticletitle{Bringing {{HPC}} Techniques to Deep Learning}.
\newblock \bibinfo{journal}{\emph{Baidu Research, Tech. Rep.}} (\bibinfo{year}{2017}).
\newblock


\bibitem[Gong et~al\mbox{.}(2024)]%
        {gongDeltaCloudassistedData2024}
\bibfield{author}{\bibinfo{person}{Chen Gong}, \bibinfo{person}{Zhenzhe Zheng}, \bibinfo{person}{Fan Wu}, \bibinfo{person}{Xiaofeng Jia}, {and} \bibinfo{person}{Guihai Chen}.} \bibinfo{year}{2024}\natexlab{}.
\newblock \showarticletitle{Delta: {{A Cloud-assisted Data Enrichment Framework}} for {{On-Device Continual Learning}}}. In \bibinfo{booktitle}{\emph{MobiCom}}. \bibinfo{pages}{1408--1423}.
\newblock


\bibitem[Grattafiori et~al\mbox{.}(2024)]%
        {grattafioriLlama3Herd2024}
\bibfield{author}{\bibinfo{person}{Aaron Grattafiori}, \bibinfo{person}{Abhimanyu Dubey}, \bibinfo{person}{Abhinav Jauhri}, \bibinfo{person}{Abhinav Pandey}, \bibinfo{person}{Abhishek Kadian}, {et~al\mbox{.}}} \bibinfo{year}{2024}\natexlab{}.
\newblock \bibinfo{title}{The {{Llama}} 3 {{Herd}} of {{Models}}}.
\newblock
\urldef\tempurl%
\url{http://arxiv.org/abs/2407.21783}
\showURL{%
\tempurl}


\bibitem[Hazelwood et~al\mbox{.}(2018)]%
        {hazelwoodAppliedMachineLearning2018}
\bibfield{author}{\bibinfo{person}{Kim Hazelwood}, \bibinfo{person}{Sarah Bird}, \bibinfo{person}{David Brooks}, \bibinfo{person}{Soumith Chintala}, \bibinfo{person}{Utku Diril}, {et~al\mbox{.}}} \bibinfo{year}{2018}\natexlab{}.
\newblock \showarticletitle{Applied {{Machine Learning}} at {{Facebook}}: {{A Datacenter Infrastructure Perspective}}}. In \bibinfo{booktitle}{\emph{HPCA}}. \bibinfo{pages}{620--629}.
\newblock
\showISBNx{978-1-5386-3659-6}


\bibitem[He et~al\mbox{.}(2016)]%
        {heDeepResidualLearning2016}
\bibfield{author}{\bibinfo{person}{Kaiming He}, \bibinfo{person}{Xiangyu Zhang}, \bibinfo{person}{Shaoqing Ren}, {and} \bibinfo{person}{Jian Sun}.} \bibinfo{year}{2016}\natexlab{}.
\newblock \showarticletitle{Deep {{Residual Learning}} for {{Image Recognition}}}. In \bibinfo{booktitle}{\emph{CVPR}}. \bibinfo{pages}{770--778}.
\newblock


\bibitem[Hong et~al\mbox{.}(2013)]%
        {hongAchievingHighUtilization2013}
\bibfield{author}{\bibinfo{person}{Chi-Yao Hong}, \bibinfo{person}{Srikanth Kandula}, \bibinfo{person}{Ratul Mahajan}, \bibinfo{person}{Ming Zhang}, \bibinfo{person}{Vijay Gill}, {et~al\mbox{.}}} \bibinfo{year}{2013}\natexlab{}.
\newblock \showarticletitle{Achieving High Utilization with Software-Driven {{WAN}}}. In \bibinfo{booktitle}{\emph{SIGCOMM}}. \bibinfo{pages}{15--26}.
\newblock
\showISBNx{978-1-4503-2056-6}


\bibitem[Hou et~al\mbox{.}(2022)]%
        {houNeuLensSpatialbasedDynamic2022}
\bibfield{author}{\bibinfo{person}{Xueyu Hou}, \bibinfo{person}{Yongjie Guan}, {and} \bibinfo{person}{Tao Han}.} \bibinfo{year}{2022}\natexlab{}.
\newblock \showarticletitle{{{NeuLens}}: Spatial-Based Dynamic Acceleration of Convolutional Neural Networks on Edge}. In \bibinfo{booktitle}{\emph{MobiCom}}. \bibinfo{pages}{186--199}.
\newblock
\showISBNx{978-1-4503-9181-8}


\bibitem[Huang et~al\mbox{.}(2020)]%
        {huangCLIOEnablingAutomatic2020}
\bibfield{author}{\bibinfo{person}{Jin Huang}, \bibinfo{person}{Colin Samplawski}, \bibinfo{person}{Deepak Ganesan}, \bibinfo{person}{Benjamin Marlin}, {and} \bibinfo{person}{Heesung Kwon}.} \bibinfo{year}{2020}\natexlab{}.
\newblock \showarticletitle{{{CLIO}}: Enabling Automatic Compilation of Deep Learning Pipelines across {{IoT}} and Cloud}. In \bibinfo{booktitle}{\emph{MobiCom}}. \bibinfo{pages}{1--12}.
\newblock
\showISBNx{978-1-4503-7085-1}


\bibitem[Ioannidis and Yeh(2018)]%
        {ioannidisAdaptiveCachingNetworks2018}
\bibfield{author}{\bibinfo{person}{Stratis Ioannidis} {and} \bibinfo{person}{Edmund Yeh}.} \bibinfo{year}{2018}\natexlab{}.
\newblock \showarticletitle{Adaptive {{Caching Networks With Optimality Guarantees}}}.
\newblock \bibinfo{journal}{\emph{IEEE/ACM Transactions on Networking}} \bibinfo{volume}{26}, \bibinfo{number}{2} (\bibinfo{year}{2018}), \bibinfo{pages}{737--750}.
\newblock
\showISSN{1558-2566}


\bibitem[Jiang et~al\mbox{.}(2021)]%
        {jiangFlexibleHighresolutionObject2021}
\bibfield{author}{\bibinfo{person}{Shiqi Jiang}, \bibinfo{person}{Zhiqi Lin}, \bibinfo{person}{Yuanchun Li}, \bibinfo{person}{Yuanchao Shu}, {and} \bibinfo{person}{Yunxin Liu}.} \bibinfo{year}{2021}\natexlab{}.
\newblock \showarticletitle{Flexible High-Resolution Object Detection on Edge Devices with Tunable Latency}. In \bibinfo{booktitle}{\emph{MobiCom}}. \bibinfo{pages}{559--572}.
\newblock
\showISBNx{978-1-4503-8342-4}


\bibitem[Kannan et~al\mbox{.}(2024)]%
        {kannanSMARTLLMSmartMultiAgent2024}
\bibfield{author}{\bibinfo{person}{Shyam~Sundar Kannan}, \bibinfo{person}{Vishnunandan L.~N. Venkatesh}, {and} \bibinfo{person}{Byung-Cheol Min}.} \bibinfo{year}{2024}\natexlab{}.
\newblock \showarticletitle{{{SMART-LLM}}: {{Smart Multi-Agent Robot Task Planning}} Using {{Large Language Models}}}. In \bibinfo{booktitle}{\emph{IROS}}. \bibinfo{pages}{12140--12147}.
\newblock
\showISSN{2153-0866}


\bibitem[Karli et~al\mbox{.}(2024)]%
        {karliAlchemistLLMAidedEndUser2024}
\bibfield{author}{\bibinfo{person}{Ulas~Berk Karli}, \bibinfo{person}{Juo-Tung Chen}, \bibinfo{person}{Victor~Nikhil Antony}, {and} \bibinfo{person}{Chien-Ming Huang}.} \bibinfo{year}{2024}\natexlab{}.
\newblock \showarticletitle{Alchemist: {{LLM-Aided End-User Development}} of {{Robot Applications}}}. In \bibinfo{booktitle}{\emph{Proceedings of the 2024 {{ACM}}/{{IEEE International Conference}} on {{Human-Robot Interaction}}}} \emph{(\bibinfo{series}{{{HRI}} '24})}. \bibinfo{pages}{361--370}.
\newblock
\showISBNx{979-8-4007-0322-5}


\bibitem[Khanam and Hussain(2024)]%
        {khanamYOLOv11OverviewKey2024}
\bibfield{author}{\bibinfo{person}{Rahima Khanam} {and} \bibinfo{person}{Muhammad Hussain}.} \bibinfo{year}{2024}\natexlab{}.
\newblock \bibinfo{title}{{{YOLOv11}}: {{An Overview}} of the {{Key Architectural Enhancements}}}.
\newblock
\urldef\tempurl%
\url{http://arxiv.org/abs/2410.17725}
\showURL{%
\tempurl}


\bibitem[Khare et~al\mbox{.}(2025)]%
        {khareSuperServeFineGrainedInference2025}
\bibfield{author}{\bibinfo{person}{Alind Khare}, \bibinfo{person}{Dhruv Garg}, \bibinfo{person}{Sukrit Kalra}, \bibinfo{person}{Snigdha Grandhi}, \bibinfo{person}{Ion Stoica}, {et~al\mbox{.}}} \bibinfo{year}{2025}\natexlab{}.
\newblock \showarticletitle{\{\vphantom\}{{SuperServe}}\vphantom\{\}: \{\vphantom\}{{Fine-Grained}}\vphantom\{\} {{Inference Serving}} for {{Unpredictable Workloads}}}. In \bibinfo{booktitle}{\emph{NSDI}}. \bibinfo{pages}{739--758}.
\newblock
\showISBNx{978-1-939133-46-5}


\bibitem[Kong et~al\mbox{.}(2023)]%
        {kongAccuMOAccuracyCentricMultitask2023}
\bibfield{author}{\bibinfo{person}{Z.~Jonny Kong}, \bibinfo{person}{Qiang Xu}, \bibinfo{person}{Jiayi Meng}, {and} \bibinfo{person}{Y.~Charlie Hu}.} \bibinfo{year}{2023}\natexlab{}.
\newblock \showarticletitle{{{AccuMO}}: {{Accuracy-Centric Multitask Offloading}} in {{Edge-Assisted Mobile Augmented Reality}}}. In \bibinfo{booktitle}{\emph{MobiCom}}. \bibinfo{pages}{1--16}.
\newblock
\showISBNx{978-1-4503-9990-6}


\bibitem[{KubeEdge}(2025)]%
        {kubeedge}
\bibfield{author}{\bibinfo{person}{{KubeEdge}}.} \bibinfo{year}{2025}\natexlab{}.
\newblock \bibinfo{title}{{{KubeEdge}}: {{Kubernetes}} Native Edge Computing Framework (Project under {{CNCF}})}.
\newblock
\urldef\tempurl%
\url{https://github.com/kubeedge/kubeedge}
\showURL{%
\tempurl}


\bibitem[Kwon et~al\mbox{.}(2023)]%
        {kwonEfficientMemoryManagement2023a}
\bibfield{author}{\bibinfo{person}{Woosuk Kwon}, \bibinfo{person}{Zhuohan Li}, \bibinfo{person}{Siyuan Zhuang}, \bibinfo{person}{Ying Sheng}, \bibinfo{person}{Lianmin Zheng}, {et~al\mbox{.}}} \bibinfo{year}{2023}\natexlab{}.
\newblock \bibinfo{title}{Efficient {{Memory Management}} for {{Large Language Model Serving}} with {{PagedAttention}}}.
\newblock
\urldef\tempurl%
\url{http://arxiv.org/abs/2309.06180}
\showURL{%
\tempurl}


\bibitem[Leitao et~al\mbox{.}(2007)]%
        {leitaoHyParViewMembershipProtocol2007}
\bibfield{author}{\bibinfo{person}{Joao Leitao}, \bibinfo{person}{Jose Pereira}, {and} \bibinfo{person}{Luis Rodrigues}.} \bibinfo{year}{2007}\natexlab{}.
\newblock \showarticletitle{{{HyParView}}: {{A Membership Protocol}} for {{Reliable Gossip-Based Broadcast}}}. In \bibinfo{booktitle}{\emph{37th {{Annual IEEE}}/{{IFIP International Conference}} on {{Dependable Systems}} and {{Networks}} ({{DSN}}'07)}}. \bibinfo{pages}{419--429}.
\newblock
\showISSN{2158-3927}


\bibitem[Li et~al\mbox{.}(2014)]%
        {liCommunicationEfficientDistributed2014}
\bibfield{author}{\bibinfo{person}{Mu Li}, \bibinfo{person}{David~G. Andersen}, \bibinfo{person}{Alexander Smola}, {and} \bibinfo{person}{Kai Yu}.} \bibinfo{year}{2014}\natexlab{}.
\newblock \showarticletitle{Communication {{Efficient Distributed Machine Learning}} with the {{Parameter Server}}}. In \bibinfo{booktitle}{\emph{NeurIPS}}, Vol.~\bibinfo{volume}{27}.
\newblock


\bibitem[Li et~al\mbox{.}(2025)]%
        {liMultiHopTaskOffloading2025}
\bibfield{author}{\bibinfo{person}{Ting Li}, \bibinfo{person}{Yinlong Liu}, \bibinfo{person}{Tao Ouyang}, \bibinfo{person}{Hangsheng Zhang}, \bibinfo{person}{Kai Yang}, {et~al\mbox{.}}} \bibinfo{year}{2025}\natexlab{}.
\newblock \showarticletitle{Multi-{{Hop Task Offloading}} and {{Relay Selection}} for {{IoT Devices}} in {{Mobile Edge Computing}}}.
\newblock \bibinfo{journal}{\emph{IEEE Transactions on Mobile Computing}} \bibinfo{volume}{24}, \bibinfo{number}{1} (\bibinfo{year}{2025}), \bibinfo{pages}{466--481}.
\newblock
\showISSN{1536-1233, 1558-0660, 2161-9875}


\bibitem[Li et~al\mbox{.}(2024)]%
        {liOMGSegOneModel2024}
\bibfield{author}{\bibinfo{person}{Xiangtai Li}, \bibinfo{person}{Haobo Yuan}, \bibinfo{person}{Wei Li}, \bibinfo{person}{Henghui Ding}, \bibinfo{person}{Size Wu}, {et~al\mbox{.}}} \bibinfo{year}{2024}\natexlab{}.
\newblock \showarticletitle{{{OMG-Seg}}: {{Is One Model Good Enough For All Segmentation}}?}. In \bibinfo{booktitle}{\emph{CVPR}}. \bibinfo{pages}{27948--27959}.
\newblock


\bibitem[Li et~al\mbox{.}(2021)]%
        {liOnlineCachingNetworks2021}
\bibfield{author}{\bibinfo{person}{Yuanyuan Li}, \bibinfo{person}{Tareq Si~Salem}, \bibinfo{person}{Giovanni Neglia}, {and} \bibinfo{person}{Stratis Ioannidis}.} \bibinfo{year}{2021}\natexlab{}.
\newblock \showarticletitle{Online {{Caching Networks}} with {{Adversarial Guarantees}}}.
\newblock \bibinfo{journal}{\emph{Proceedings of the ACM on Measurement and Analysis of Computing Systems}} \bibinfo{volume}{5}, \bibinfo{number}{3} (\bibinfo{year}{2021}), \bibinfo{pages}{1--39}.
\newblock
\showISSN{2476-1249}


\bibitem[Li et~al\mbox{.}(2023)]%
        {liAlpaServeStatisticalMultiplexing2023}
\bibfield{author}{\bibinfo{person}{Zhuohan Li}, \bibinfo{person}{Lianmin Zheng}, \bibinfo{person}{Yinmin Zhong}, \bibinfo{person}{Vincent Liu}, \bibinfo{person}{Ying Sheng}, {et~al\mbox{.}}} \bibinfo{year}{2023}\natexlab{}.
\newblock \showarticletitle{\{\vphantom\}{{AlpaServe}}\vphantom\{\}: {{Statistical Multiplexing}} with {{Model Parallelism}} for {{Deep Learning Serving}}}. In \bibinfo{booktitle}{\emph{OSDI}}. \bibinfo{pages}{663--679}.
\newblock
\showISBNx{978-1-939133-34-2}


\bibitem[Lin et~al\mbox{.}(2024a)]%
        {linOpenSoraPlanOpenSource2024}
\bibfield{author}{\bibinfo{person}{Bin Lin}, \bibinfo{person}{Yunyang Ge}, \bibinfo{person}{Xinhua Cheng}, \bibinfo{person}{Zongjian Li}, \bibinfo{person}{Bin Zhu}, {et~al\mbox{.}}} \bibinfo{year}{2024}\natexlab{a}.
\newblock \bibinfo{title}{Open-{{Sora Plan}}: {{Open-Source Large Video Generation Model}}}.
\newblock
\urldef\tempurl%
\url{http://arxiv.org/abs/2412.00131}
\showURL{%
\tempurl}


\bibitem[Lin et~al\mbox{.}(2024b)]%
        {linParrotEfficientServing2024}
\bibfield{author}{\bibinfo{person}{Chaofan Lin}, \bibinfo{person}{Zhenhua Han}, \bibinfo{person}{Chengruidong Zhang}, \bibinfo{person}{Yuqing Yang}, \bibinfo{person}{Fan Yang}, {et~al\mbox{.}}} \bibinfo{year}{2024}\natexlab{b}.
\newblock \showarticletitle{Parrot: {{Efficient Serving}} of \{\vphantom\}{{LLM-based}}\vphantom\{\} {{Applications}} with {{Semantic Variable}}}. In \bibinfo{booktitle}{\emph{OSDI}}. \bibinfo{pages}{929--945}.
\newblock
\showISBNx{978-1-939133-40-3}


\bibitem[Meng et~al\mbox{.}(2024)]%
        {mengHairpinRethinkingPacket2024}
\bibfield{author}{\bibinfo{person}{Zili Meng}, \bibinfo{person}{Xiao Kong}, \bibinfo{person}{Jing Chen}, \bibinfo{person}{Bo Wang}, \bibinfo{person}{Mingwei Xu}, {et~al\mbox{.}}} \bibinfo{year}{2024}\natexlab{}.
\newblock \showarticletitle{Hairpin: {{Rethinking Packet Loss Recovery}} in {{Edge-based Interactive Video Streaming}}}. In \bibinfo{booktitle}{\emph{NSDI}}. \bibinfo{pages}{907--926}.
\newblock
\showISBNx{978-1-939133-39-7}


\bibitem[Miao et~al\mbox{.}(2024)]%
        {miaoSpecInferAcceleratingLarge2024}
\bibfield{author}{\bibinfo{person}{Xupeng Miao}, \bibinfo{person}{Gabriele Oliaro}, \bibinfo{person}{Zhihao Zhang}, \bibinfo{person}{Xinhao Cheng}, \bibinfo{person}{Zeyu Wang}, {et~al\mbox{.}}} \bibinfo{year}{2024}\natexlab{}.
\newblock \showarticletitle{{{SpecInfer}}: {{Accelerating Large Language Model Serving}} with {{Tree-based Speculative Inference}} and {{Verification}}}. In \bibinfo{booktitle}{\emph{ASPLOS}}, Vol.~\bibinfo{volume}{3}. \bibinfo{pages}{932--949}.
\newblock
\showISBNx{979-8-4007-0386-7}


\bibitem[{NVIDIA}(2025a)]%
        {nvidiaNVIDIAJetsonNano2025}
\bibfield{author}{\bibinfo{person}{{NVIDIA}}.} \bibinfo{year}{2025}\natexlab{a}.
\newblock \bibinfo{title}{{Nvidia Jetson Nano}}.
\newblock
\urldef\tempurl%
\url{https://www.nvidia.cn/autonomous-machines/embedded-systems/jetson-nano/product-development/}
\showURL{%
\tempurl}


\bibitem[{NVIDIA}(2025b)]%
        {nvidiaNVIDIAMultiProcessService2025}
\bibfield{author}{\bibinfo{person}{{NVIDIA}}.} \bibinfo{year}{2025}\natexlab{b}.
\newblock \bibinfo{title}{Nvidia {{Multi-Process Service}} ({{MPS}})}.
\newblock
\urldef\tempurl%
\url{https://docs.nvidia.com/deploy/mps/}
\showURL{%
\tempurl}


\bibitem[{NVIDIA}(2025c)]%
        {nvidiaNVIDIANsightSystems2025}
\bibfield{author}{\bibinfo{person}{{NVIDIA}}.} \bibinfo{year}{2025}\natexlab{c}.
\newblock \bibinfo{title}{Nvidia {{Nsight Systems}}}.
\newblock
\urldef\tempurl%
\url{https://developer.nvidia.com/nsight-systems}
\showURL{%
\tempurl}


\bibitem[{OpenStack}(2025)]%
        {openstackOpenStackDocsScheduling2025}
\bibfield{author}{\bibinfo{person}{{OpenStack}}.} \bibinfo{year}{2025}\natexlab{}.
\newblock \bibinfo{title}{{{OpenStack Docs}}: {{Scheduling}}}.
\newblock
\urldef\tempurl%
\url{https://docs.openstack.org/mitaka/config-reference/compute/scheduler.html}
\showURL{%
\tempurl}


\bibitem[Padmanabhan et~al\mbox{.}(2023)]%
        {padmanabhanGemelModelMerging2023}
\bibfield{author}{\bibinfo{person}{Arthi Padmanabhan}, \bibinfo{person}{Neil Agarwal}, \bibinfo{person}{Anand Iyer}, \bibinfo{person}{Ganesh Ananthanarayanan}, \bibinfo{person}{Yuanchao Shu}, {et~al\mbox{.}}} \bibinfo{year}{2023}\natexlab{}.
\newblock \showarticletitle{Gemel: {{Model Merging}} for \{\vphantom\}{{Memory-Efficient}}\vphantom\{\}, \{\vphantom\}{{Real-Time}}\vphantom\{\} {{Video Analytics}} at the {{Edge}}}. In \bibinfo{booktitle}{\emph{NSDI}}. \bibinfo{pages}{973--994}.
\newblock
\showISBNx{978-1-939133-33-5}


\bibitem[Park et~al\mbox{.}(2018)]%
        {parkDeepLearningInference2018}
\bibfield{author}{\bibinfo{person}{Jongsoo Park}, \bibinfo{person}{Maxim Naumov}, \bibinfo{person}{Protonu Basu}, \bibinfo{person}{Summer Deng}, \bibinfo{person}{Aravind Kalaiah}, {et~al\mbox{.}}} \bibinfo{year}{2018}\natexlab{}.
\newblock \bibinfo{title}{Deep {{Learning Inference}} in {{Facebook Data Centers}}: {{Characterization}}, {{Performance Optimizations}} and {{Hardware Implications}}}.
\newblock
\urldef\tempurl%
\url{http://arxiv.org/abs/1811.09886}
\showURL{%
\tempurl}


\bibitem[Paszke et~al\mbox{.}(2019)]%
        {paszkePyTorchImperativeStyle2019}
\bibfield{author}{\bibinfo{person}{Adam Paszke}, \bibinfo{person}{Sam Gross}, \bibinfo{person}{Francisco Massa}, \bibinfo{person}{Adam Lerer}, \bibinfo{person}{James Bradbury}, {et~al\mbox{.}}} \bibinfo{year}{2019}\natexlab{}.
\newblock \showarticletitle{{{PyTorch}}: {{An Imperative Style}}, {{High-Performance Deep Learning Library}}}. In \bibinfo{booktitle}{\emph{NeurIPS}}, Vol.~\bibinfo{volume}{32}.
\newblock


\bibitem[Patel et~al\mbox{.}(2024)]%
        {patelSplitwiseEfficientGenerative2024}
\bibfield{author}{\bibinfo{person}{Pratyush Patel}, \bibinfo{person}{Esha Choukse}, \bibinfo{person}{Chaojie Zhang}, \bibinfo{person}{Aashaka Shah}, \bibinfo{person}{{\'I}{\~n}igo Goiri}, {et~al\mbox{.}}} \bibinfo{year}{2024}\natexlab{}.
\newblock \showarticletitle{Splitwise: {{Efficient Generative LLM Inference Using Phase Splitting}}}. In \bibinfo{booktitle}{\emph{ISCA}}. \bibinfo{pages}{118--132}.
\newblock


\bibitem[Qwen et~al\mbox{.}(2025)]%
        {qwenQwen25TechnicalReport2025}
\bibfield{author}{\bibinfo{person}{Qwen}, \bibinfo{person}{An Yang}, \bibinfo{person}{Baosong Yang}, \bibinfo{person}{Beichen Zhang}, \bibinfo{person}{Binyuan Hui}, {et~al\mbox{.}}} \bibinfo{year}{2025}\natexlab{}.
\newblock \bibinfo{title}{Qwen2.5 {{Technical Report}}}.
\newblock
\urldef\tempurl%
\url{http://arxiv.org/abs/2412.15115}
\showURL{%
\tempurl}


\bibitem[{Rajagopalan} et~al\mbox{.}(2025)]%
        {rajagopalanRunningDistilledDeepSeek2025}
\bibfield{author}{\bibinfo{person}{{Rajagopalan}}, \bibinfo{person}{{Vivek Pradeep}}, \bibinfo{person}{{Logan Iyer}}, {and} \bibinfo{person}{{Raji}}.} \bibinfo{year}{2025}\natexlab{}.
\newblock \bibinfo{title}{Running {{Distilled DeepSeek R1}} Models Locally on {{Copilot}}+ {{PCs}}, Powered by {{Windows Copilot Runtime}}}.
\newblock
\urldef\tempurl%
\url{https://blogs.windows.com/windowsdeveloper/2025/01/29/running-distilled-deepseek-r1-models-locally-on-copilot-pcs-powered-by-windows-copilot-runtime/}
\showURL{%
\tempurl}


\bibitem[Rajasekaran et~al\mbox{.}(2024)]%
        {rajasekaranCASSININetworkAwareJob2024}
\bibfield{author}{\bibinfo{person}{Sudarsanan Rajasekaran}, \bibinfo{person}{Manya Ghobadi}, {and} \bibinfo{person}{Aditya Akella}.} \bibinfo{year}{2024}\natexlab{}.
\newblock \showarticletitle{\{\vphantom\}{{CASSINI}}\vphantom\{\}: \{\vphantom\}{{Network-Aware}}\vphantom\{\} {{Job Scheduling}} in {{Machine Learning Clusters}}}. In \bibinfo{booktitle}{\emph{NSDI}}. \bibinfo{pages}{1403--1420}.
\newblock
\showISBNx{978-1-939133-39-7}


\bibitem[Rajbhandari et~al\mbox{.}(2020)]%
        {rajbhandariZeROMemoryOptimizations2020}
\bibfield{author}{\bibinfo{person}{Samyam Rajbhandari}, \bibinfo{person}{Jeff Rasley}, \bibinfo{person}{Olatunji Ruwase}, {and} \bibinfo{person}{Yuxiong He}.} \bibinfo{year}{2020}\natexlab{}.
\newblock \showarticletitle{{{ZeRO}}: {{Memory}} Optimizations {{Toward Training Trillion Parameter Models}}}. In \bibinfo{booktitle}{\emph{SC}}. \bibinfo{pages}{1--16}.
\newblock
\showISBNx{978-1-7281-9998-6}


\bibitem[Ronneberger et~al\mbox{.}(2015)]%
        {ronnebergerUNetConvolutionalNetworks2015}
\bibfield{author}{\bibinfo{person}{Olaf Ronneberger}, \bibinfo{person}{Philipp Fischer}, {and} \bibinfo{person}{Thomas Brox}.} \bibinfo{year}{2015}\natexlab{}.
\newblock \showarticletitle{U-{{Net}}: {{Convolutional Networks}} for {{Biomedical Image Segmentation}}}. In \bibinfo{booktitle}{\emph{MICCAI}}. \bibinfo{pages}{234--241}.
\newblock
\showISBNx{978-3-319-24574-4}


\bibitem[Sandler et~al\mbox{.}(2018)]%
        {sandlerMobileNetV2InvertedResiduals2018}
\bibfield{author}{\bibinfo{person}{Mark Sandler}, \bibinfo{person}{Andrew Howard}, \bibinfo{person}{Menglong Zhu}, \bibinfo{person}{Andrey Zhmoginov}, {and} \bibinfo{person}{Liang-Chieh Chen}.} \bibinfo{year}{2018}\natexlab{}.
\newblock \showarticletitle{{{MobileNetV2}}: {{Inverted Residuals}} and {{Linear Bottlenecks}}}. In \bibinfo{booktitle}{\emph{CVPR}}. \bibinfo{pages}{4510--4520}.
\newblock


\bibitem[Shen et~al\mbox{.}(2019)]%
        {shenNexusGPUCluster2019}
\bibfield{author}{\bibinfo{person}{Haichen Shen}, \bibinfo{person}{Lequn Chen}, \bibinfo{person}{Yuchen Jin}, \bibinfo{person}{Liangyu Zhao}, \bibinfo{person}{Bingyu Kong}, {et~al\mbox{.}}} \bibinfo{year}{2019}\natexlab{}.
\newblock \showarticletitle{Nexus: A {{GPU}} Cluster Engine for Accelerating {{DNN-based}} Video Analysis}. In \bibinfo{booktitle}{\emph{SOSP}}. \bibinfo{pages}{322--337}.
\newblock
\showISBNx{978-1-4503-6873-5}


\bibitem[Shoeybi et~al\mbox{.}(2020)]%
        {shoeybiMegatronLMTrainingMultiBillion2020}
\bibfield{author}{\bibinfo{person}{Mohammad Shoeybi}, \bibinfo{person}{Mostofa Patwary}, \bibinfo{person}{Raul Puri}, \bibinfo{person}{Patrick LeGresley}, \bibinfo{person}{Jared Casper}, {et~al\mbox{.}}} \bibinfo{year}{2020}\natexlab{}.
\newblock \bibinfo{title}{Megatron-{{LM}}: {{Training Multi-Billion Parameter Language Models Using Model Parallelism}}}.
\newblock
\urldef\tempurl%
\url{http://arxiv.org/abs/1909.08053}
\showURL{%
\tempurl}


\bibitem[Shubha and Shen(2023)]%
        {shubhaAdaInfDataDrift2023}
\bibfield{author}{\bibinfo{person}{Sudipta~Saha Shubha} {and} \bibinfo{person}{Haiying Shen}.} \bibinfo{year}{2023}\natexlab{}.
\newblock \showarticletitle{{{AdaInf}}: {{Data Drift Adaptive Scheduling}} for {{Accurate}} and {{SLO-guaranteed Multiple-Model Inference Serving}} at {{Edge Servers}}}. In \bibinfo{booktitle}{\emph{SIGCOMM}}. \bibinfo{pages}{473--485}.
\newblock
\showISBNx{979-8-4007-0236-5}


\bibitem[Shubha et~al\mbox{.}(2024)]%
        {shubhaUSHERHolisticInterference2024}
\bibfield{author}{\bibinfo{person}{Sudipta~Saha Shubha}, \bibinfo{person}{Haiying Shen}, {and} \bibinfo{person}{Anand Iyer}.} \bibinfo{year}{2024}\natexlab{}.
\newblock \showarticletitle{\{\vphantom\}{{USHER}}\vphantom\{\}: {{Holistic Interference Avoidance}} for {{Resource Optimized}} \{\vphantom\}{{ML}}\vphantom\{\} {{Inference}}}. In \bibinfo{booktitle}{\emph{OSDI}}. \bibinfo{pages}{947--964}.
\newblock
\showISBNx{978-1-939133-40-3}


\bibitem[Simonyan and Zisserman(2015)]%
        {simonyanVeryDeepConvolutional2015}
\bibfield{author}{\bibinfo{person}{Karen Simonyan} {and} \bibinfo{person}{Andrew Zisserman}.} \bibinfo{year}{2015}\natexlab{}.
\newblock \bibinfo{title}{Very {{Deep Convolutional Networks}} for {{Large-Scale Image Recognition}}}.
\newblock
\urldef\tempurl%
\url{http://arxiv.org/abs/1409.1556}
\showURL{%
\tempurl}


\bibitem[Soifer et~al\mbox{.}(2019)]%
        {soiferDeepLearningInference2019}
\bibfield{author}{\bibinfo{person}{Jonathan Soifer}, \bibinfo{person}{Jason Li}, \bibinfo{person}{Mingqin Li}, \bibinfo{person}{Jeffrey Zhu}, \bibinfo{person}{Yingnan Li}, {et~al\mbox{.}}} \bibinfo{year}{2019}\natexlab{}.
\newblock \showarticletitle{Deep {{Learning Inference Service}} at {{Microsoft}}}. In \bibinfo{booktitle}{\emph{2019 {{USENIX Conference}} on {{Operational Machine Learning}} ({{OpML}} 19)}}. \bibinfo{pages}{15--17}.
\newblock
\showISBNx{978-1-939133-00-7}


\bibitem[Wan et~al\mbox{.}(2024)]%
        {wanBiCCBilateralCongestion2024}
\bibfield{author}{\bibinfo{person}{Zirui Wan}, \bibinfo{person}{Jiao Zhang}, \bibinfo{person}{Mingxuan Yu}, \bibinfo{person}{Junwei Liu}, \bibinfo{person}{Jun Yao}, {et~al\mbox{.}}} \bibinfo{year}{2024}\natexlab{}.
\newblock \showarticletitle{{{BiCC}}: {{Bilateral Congestion Control}} in {{Cross-datacenter RDMA Networks}}}. In \bibinfo{booktitle}{\emph{INFOCOM}}. \bibinfo{pages}{1381--1390}.
\newblock
\showISSN{2641-9874}


\bibitem[Wang et~al\mbox{.}(2024)]%
        {NEURIPS2024_c34ddd05}
\bibfield{author}{\bibinfo{person}{Ao Wang}, \bibinfo{person}{Hui Chen}, \bibinfo{person}{Lihao Liu}, \bibinfo{person}{Kai CHEN}, \bibinfo{person}{Zijia Lin}, {et~al\mbox{.}}} \bibinfo{year}{2024}\natexlab{}.
\newblock \showarticletitle{{{YOLOv10}}: {{Real-time}} End-to-End Object Detection}. In \bibinfo{booktitle}{\emph{NeurIPS}}, Vol.~\bibinfo{volume}{37}. \bibinfo{pages}{107984--108011}.
\newblock


\bibitem[Wang et~al\mbox{.}(2017)]%
        {wangPipeCNNOpenCLbasedOpensource2017}
\bibfield{author}{\bibinfo{person}{Dong Wang}, \bibinfo{person}{Ke Xu}, {and} \bibinfo{person}{Diankun Jiang}.} \bibinfo{year}{2017}\natexlab{}.
\newblock \showarticletitle{{{PipeCNN}}: {{An OpenCL-based}} Open-Source {{FPGA}} Accelerator for Convolution Neural Networks}. In \bibinfo{booktitle}{\emph{2017 {{International Conference}} on {{Field Programmable Technology}} ({{ICFPT}})}}. \bibinfo{pages}{279--282}.
\newblock


\bibitem[Wang et~al\mbox{.}(2022)]%
        {wangDependentTaskOffloading2022}
\bibfield{author}{\bibinfo{person}{Jin Wang}, \bibinfo{person}{Jia Hu}, \bibinfo{person}{Geyong Min}, \bibinfo{person}{Wenhan Zhan}, \bibinfo{person}{Albert~Y. Zomaya}, {et~al\mbox{.}}} \bibinfo{year}{2022}\natexlab{}.
\newblock \showarticletitle{Dependent {{Task Offloading}} for {{Edge Computing}} Based on {{Deep Reinforcement Learning}}}.
\newblock \bibinfo{journal}{\emph{IEEE Trans. Comput.}} \bibinfo{volume}{71}, \bibinfo{number}{10} (\bibinfo{year}{2022}), \bibinfo{pages}{2449--2461}.
\newblock
\showISSN{0018-9340, 1557-9956, 2326-3814}


\bibitem[Wang et~al\mbox{.}(2023)]%
        {wangIntelligentLowOverhead2023}
\bibfield{author}{\bibinfo{person}{Xianbin Wang}, \bibinfo{person}{Pengyi Jia}, \bibinfo{person}{Xuemin Shen}, {and} \bibinfo{person}{H.~Vincent Poor}.} \bibinfo{year}{2023}\natexlab{}.
\newblock \showarticletitle{Intelligent and {{Low Overhead Network Synchronization}} for {{Large-Scale Industrial IoT Systems}} in the {{6G Era}}}.
\newblock \bibinfo{journal}{\emph{IEEE Network}} \bibinfo{volume}{37}, \bibinfo{number}{3} (\bibinfo{year}{2023}), \bibinfo{pages}{76--84}.
\newblock
\showISSN{1558-156X}


\bibitem[Wei et~al\mbox{.}(2024)]%
        {weiCommunicationEfficientModelParallelism2024}
\bibfield{author}{\bibinfo{person}{Yuanxin Wei}, \bibinfo{person}{Shengyuan Ye}, \bibinfo{person}{Jiazhi Jiang}, \bibinfo{person}{Xu Chen}, \bibinfo{person}{Dan Huang}, {et~al\mbox{.}}} \bibinfo{year}{2024}\natexlab{}.
\newblock \showarticletitle{Communication-{{Efficient Model Parallelism}} for {{Distributed In-Situ Transformer Inference}}}. In \bibinfo{booktitle}{\emph{2024 {{Design}}, {{Automation}} \& {{Test}} in {{Europe Conference}} \& {{Exhibition}} ({{DATE}})}}. \bibinfo{pages}{1--6}.
\newblock
\showISSN{1558-1101}


\bibitem[Wen et~al\mbox{.}(2024)]%
        {wenSurveyIntegratedSensing2024}
\bibfield{author}{\bibinfo{person}{Dingzhu Wen}, \bibinfo{person}{Yong Zhou}, \bibinfo{person}{Xiaoyang Li}, \bibinfo{person}{Yuanming Shi}, \bibinfo{person}{Kaibin Huang}, {et~al\mbox{.}}} \bibinfo{year}{2024}\natexlab{}.
\newblock \showarticletitle{A {{Survey}} on {{Integrated Sensing}}, {{Communication}}, and {{Computation}}}.
\newblock \bibinfo{journal}{\emph{IEEE Communications Surveys \& Tutorials}} (\bibinfo{year}{2024}), \bibinfo{pages}{1--41}.
\newblock
\showISSN{1553-877X}


\bibitem[Wen et~al\mbox{.}(2023)]%
        {wenAdaptiveNetPostdeploymentNeural2023}
\bibfield{author}{\bibinfo{person}{Hao Wen}, \bibinfo{person}{Yuanchun Li}, \bibinfo{person}{Zunshuai Zhang}, \bibinfo{person}{Shiqi Jiang}, \bibinfo{person}{Xiaozhou Ye}, {et~al\mbox{.}}} \bibinfo{year}{2023}\natexlab{}.
\newblock \showarticletitle{{{AdaptiveNet}}: {{Post-deployment Neural Architecture Adaptation}} for {{Diverse Edge Environments}}}. In \bibinfo{booktitle}{\emph{MobiCom}}. \bibinfo{pages}{1--17}.
\newblock
\showISBNx{978-1-4503-9990-6}


\bibitem[Wu et~al\mbox{.}(2024)]%
        {wuFastDistributedInference2024}
\bibfield{author}{\bibinfo{person}{Bingyang Wu}, \bibinfo{person}{Yinmin Zhong}, \bibinfo{person}{Zili Zhang}, \bibinfo{person}{Shengyu Liu}, \bibinfo{person}{Fangyue Liu}, {et~al\mbox{.}}} \bibinfo{year}{2024}\natexlab{}.
\newblock \bibinfo{title}{Fast {{Distributed Inference Serving}} for {{Large Language Models}}}.
\newblock
\urldef\tempurl%
\url{http://arxiv.org/abs/2305.05920}
\showURL{%
\tempurl}


\bibitem[Wu et~al\mbox{.}(2016)]%
        {wuGooglesNeuralMachine2016}
\bibfield{author}{\bibinfo{person}{Yonghui Wu}, \bibinfo{person}{Mike Schuster}, \bibinfo{person}{Zhifeng Chen}, \bibinfo{person}{Quoc~V. Le}, \bibinfo{person}{Mohammad Norouzi}, {et~al\mbox{.}}} \bibinfo{year}{2016}\natexlab{}.
\newblock \bibinfo{title}{Google's {{Neural Machine Translation System}}: {{Bridging}} the {{Gap}} between {{Human}} and {{Machine Translation}}}.
\newblock
\urldef\tempurl%
\url{http://arxiv.org/abs/1609.08144}
\showURL{%
\tempurl}


\bibitem[Xiaoxia(2025)]%
        {xiaoxiaXiaoZhiAIChatbot2025}
\bibfield{author}{\bibinfo{person}{Xiaoxia}.} \bibinfo{year}{2025}\natexlab{}.
\newblock \bibinfo{title}{{{XiaoZhi AI Chatbot}}}.
\newblock
\urldef\tempurl%
\url{https://github.com/78/xiaozhi-esp32}
\showURL{%
\tempurl}


\bibitem[Xu et~al\mbox{.}(2024)]%
        {xuSCTNetSingleBranchCNN2024}
\bibfield{author}{\bibinfo{person}{Zhengze Xu}, \bibinfo{person}{Dongyue Wu}, \bibinfo{person}{Changqian Yu}, \bibinfo{person}{Xiangxiang Chu}, \bibinfo{person}{Nong Sang}, {et~al\mbox{.}}} \bibinfo{year}{2024}\natexlab{}.
\newblock \showarticletitle{Sctnet: {{Single-branch}} Cnn with Transformer Semantic Information for Real-Time Segmentation}. In \bibinfo{booktitle}{\emph{AAAI}}, Vol.~\bibinfo{volume}{38}. \bibinfo{pages}{6378--6386}.
\newblock


\bibitem[Ye et~al\mbox{.}(2024a)]%
        {yeGalaxyResourceEfficientCollaborative2024a}
\bibfield{author}{\bibinfo{person}{Shengyuan Ye}, \bibinfo{person}{Jiangsu Du}, \bibinfo{person}{Liekang Zeng}, \bibinfo{person}{Wenzhong Ou}, \bibinfo{person}{Xiaowen Chu}, {et~al\mbox{.}}} \bibinfo{year}{2024}\natexlab{a}.
\newblock \showarticletitle{Galaxy: {{A Resource-Efficient Collaborative Edge AI System}} for {{In-situ Transformer Inference}}}. In \bibinfo{booktitle}{\emph{INFOCOM}}. \bibinfo{pages}{1001--1010}.
\newblock
\showISBNx{979-8-3503-8350-8}


\bibitem[Ye et~al\mbox{.}(2025)]%
        {shengyuanyeJupiterFastResourceEfficient2024}
\bibfield{author}{\bibinfo{person}{Shengyuan Ye}, \bibinfo{person}{Bei Ouyang}, \bibinfo{person}{Liekang Zeng}, \bibinfo{person}{Tianyi Qian}, \bibinfo{person}{Xiaowen Chu}, {et~al\mbox{.}}} \bibinfo{year}{2025}\natexlab{}.
\newblock \showarticletitle{Jupiter: {{Fast}} and Resource-Efficient Collaborative Inference of Generative Llms on Edge Devices}. In \bibinfo{booktitle}{\emph{INFOCOM}}. \bibinfo{pages}{1--10}.
\newblock


\bibitem[Ye et~al\mbox{.}(2024b)]%
        {yeAsteroidResourceEfficientHybrid2024}
\bibfield{author}{\bibinfo{person}{Shengyuan Ye}, \bibinfo{person}{Liekang Zeng}, \bibinfo{person}{Xiaowen Chu}, \bibinfo{person}{Guoliang Xing}, {and} \bibinfo{person}{Xu Chen}.} \bibinfo{year}{2024}\natexlab{b}.
\newblock \showarticletitle{Asteroid: {{Resource-Efficient Hybrid Pipeline Parallelism}} for {{Collaborative DNN Training}} on {{Heterogeneous Edge Devices}}}. In \bibinfo{booktitle}{\emph{MobiCom}}. \bibinfo{pages}{312--326}.
\newblock
\showISBNx{979-8-4007-0489-5}


\bibitem[Zhang et~al\mbox{.}(2023)]%
        {zhangSHEPHERDServingDNNs2023}
\bibfield{author}{\bibinfo{person}{Hong Zhang}, \bibinfo{person}{Yupeng Tang}, \bibinfo{person}{Anurag Khandelwal}, {and} \bibinfo{person}{Ion Stoica}.} \bibinfo{year}{2023}\natexlab{}.
\newblock \showarticletitle{\{\vphantom\}{{SHEPHERD}}\vphantom\{\}: {{Serving}} \{\vphantom\}{{DNNs}}\vphantom\{\} in the {{Wild}}}. In \bibinfo{booktitle}{\emph{NSDI}}. \bibinfo{pages}{787--808}.
\newblock
\showISBNx{978-1-939133-33-5}


\bibitem[Zhang et~al\mbox{.}(2017)]%
        {zhangResilientDatacenterLoad2017}
\bibfield{author}{\bibinfo{person}{Hong Zhang}, \bibinfo{person}{Junxue Zhang}, \bibinfo{person}{Wei Bai}, \bibinfo{person}{Kai Chen}, {and} \bibinfo{person}{Mosharaf Chowdhury}.} \bibinfo{year}{2017}\natexlab{}.
\newblock \showarticletitle{Resilient {{Datacenter Load Balancing}} in the {{Wild}}}. In \bibinfo{booktitle}{\emph{SIGCOMM}}. \bibinfo{pages}{253--266}.
\newblock
\showISBNx{978-1-4503-4653-5}


\bibitem[Zhang et~al\mbox{.}(2024)]%
        {zhangEdgeShardEfficientLLM2024}
\bibfield{author}{\bibinfo{person}{Mingjin Zhang}, \bibinfo{person}{Jiannong Cao}, \bibinfo{person}{Xiaoming Shen}, {and} \bibinfo{person}{Zeyang Cui}.} \bibinfo{year}{2024}\natexlab{}.
\newblock \bibinfo{title}{{{EdgeShard}}: {{Efficient LLM Inference}} via {{Collaborative Edge Computing}}}.
\newblock
\urldef\tempurl%
\url{http://arxiv.org/abs/2405.14371}
\showURL{%
\tempurl}


\bibitem[Zhang et~al\mbox{.}(2021)]%
        {zhangFasterCheaperServerless2021}
\bibfield{author}{\bibinfo{person}{Yanqi Zhang}, \bibinfo{person}{{\'I}{\~n}igo Goiri}, \bibinfo{person}{Gohar~Irfan Chaudhry}, \bibinfo{person}{Rodrigo Fonseca}, \bibinfo{person}{Sameh Elnikety}, {et~al\mbox{.}}} \bibinfo{year}{2021}\natexlab{}.
\newblock \showarticletitle{Faster and {{Cheaper Serverless Computing}} on {{Harvested Resources}}}. In \bibinfo{booktitle}{\emph{SOSP}}. \bibinfo{pages}{724--739}.
\newblock
\showISBNx{978-1-4503-8709-5}


\bibitem[Zhong et~al\mbox{.}(2024)]%
        {zhongDistServeDisaggregatingPrefill2024}
\bibfield{author}{\bibinfo{person}{Yinmin Zhong}, \bibinfo{person}{Shengyu Liu}, \bibinfo{person}{Junda Chen}, \bibinfo{person}{Jianbo Hu}, \bibinfo{person}{Yibo Zhu}, {et~al\mbox{.}}} \bibinfo{year}{2024}\natexlab{}.
\newblock \showarticletitle{\{\vphantom\}{{DistServe}}\vphantom\{\}: {{Disaggregating Prefill}} and {{Decoding}} for {{Goodput-optimized Large Language Model Serving}}}. In \bibinfo{booktitle}{\emph{OSDI}}. \bibinfo{pages}{193--210}.
\newblock
\showISBNx{978-1-939133-40-3}


\end{thebibliography}



\appendix
\section*{Appendix}
\noindent\emph{Note:} Appendices are supporting material that has not been peer-reviewed.

\section{Submodular Function and Lower Bound}\label{section:appendix-submodular}

\begin{theorem}[\textbf{Submodular Function}]\label{theorem:Submodular Function}
  \EPARAit state-aware placement $\varphi(X)$ is a submodular function.
\end{theorem}

\begin{proof}
  \label{proof:subodular function}
Notice that there is a scheduling variable $y_{tlnm}$ in the model that represents if the request $r_{tln}$ is satisfied on $m$. It is not difficult to find that when we determine two sets of $x_{ln}$ and $r_{tln}$, it must correspond to a certain set of $y_{tlnm}$, that is, $y_{tlnm}$ is a function of $x_{ln}$, $y_{tlnm}=f(x_{ln})$, where $f$ represents request handling strategy (\cref{section:design-handling}) and $y_{tlnm}$ is the shadow scheduling variable. Therefore, $x_{ln}$ completely determines the final output of the scheduling process. 

We make $X=\{x_{ln}\mid l \in L, n \in N\}, Y=\{y_{tlnm}\mid t\in T, l \in L, m \in N, n \in N\}$ and use $\varphi$ to represent the objective function of placement. \EPARA wants to maximize the objective function $\varphi(Y)$, which is equivalent to maximizing $\varphi(f(X))$, where $f$ represents the handling strategy in \cref{section:design-handling}.

It is only necessary to show that the placement process $\varphi(X)$ is submodular when the input to the function is only $X$. The function can be expressed as $H:2^{X}\rightarrow R$, where $X$ is a finite set of all caching possibilities, $2^X$ is the power set of $X$, and $R$ is a set of real numbers. 

We let,
\begin{equation}
\rho_{\Omega}(i)=\varphi(\Omega+i)-\varphi(\Omega). 
\end{equation}

For any $A\subseteq B\subseteq X$, and for any cache $\xi\in X\setminus B$, there is always,
\begin{equation}
\rho_A(\xi) \geq \rho_B(\xi).
\end{equation}

Therefore, the final objective function $\varphi$ is the submodular function of $X$.
\end{proof}

\begin{theorem}[\textbf{Lower Bound Approximation}]\label{theorem:approximation}
    \EPARAit state-aware placement yields a $1/(1+P)$ approximation, where
    \begin{center}
    $P = \left\lceil \frac{\max a_{l}}{\min_{a_{l} > 0} a_{l}} \right\rceil + \left\lceil \frac{\max b_l}{\min_{ b_l > 0} b_l} \right\rceil$, 
    \end{center}
    when each server has at most one GPU that cannot fully utilize computational power or VRAM.
\end{theorem}

\begin{proof}~\\
\textbf{1) We normalize the problem and algorithm in}~\cref{section:design-placement}~\cite{fisherAnalysisApproximationsMaximizing1978}.
\begin{equation}
    \begin{split}
    \max_{\Theta \subseteq X}\{\varphi(\Theta): \Theta \in \bigcap_{p=1}^{P} \mathscr{F}_p,\mathscr{M}_{p}=(X,\mathscr{F}_{p})\ are \ matroids,\\ 
    p=1,\dots,P,\varphi(\Theta)\ Submodular\ and\ nondecreasing\}
    \end{split}
\end{equation}

\noindent The greedy heuristic for nondecreasing set functions on independence systems ($X, \mathscr{F}$). 

Initialization. Let $\Theta^{0}$ = $\emptyset$, $X^{0}=X$ and set $t=1$. 

Iteration k 

Step 0. If $X^{k-1}=\emptyset$, stop with $\Theta^{k-1}$ the greedy solution.

Step 1. Select $i(k)\in X^{k-1}$ with ties settled arbitrarily, for which $\rho_{i(k)}(\Theta^{k-1})=\max_{i\in X^{k-1}}\rho_{i}(\Theta^{k-1})$. 

Step 2a. If $\Theta^{k-1}\cup\{{i(k)}\}\notin{\mathscr{F}}$, set $X^{k-1}=X^{k-1}-\{i(k)\}$ and return to Step 0.

Step 2b. If $\Theta^{k-1}\cup\{{i(k)}\}\in{\mathscr{F}}$, set $\rho_{k-1}=\rho_{i(k)}(\Theta^{k-1})$, $\Theta^{k}=\Theta^{k-1}\cup\{i(k)\}$ and $X^t=X^{t-1}-\{i(k)\}$

Step 3. Set $t\rightarrow t+1$ and continue. 

Finally, we let $\Phi$ denote the optimal value and $\Phi^G$ the value of a greedy solution to this problem.

\noindent\textbf{2) Proposition 1.} 
$U^{k}\subseteq \bigcup_{p=1}^{P}\theta p^{p}(\Theta^k),t=0,1,\dots.$

\textit{Proof.} If \( j \in U^k \), then either \( j \in S^k \subseteq \text{sp}^p(S^k) \) for all \( p \), or \( j \) failed the independence test at Step 2a, which implies \( j \in \text{sp}^p(S^k) \) for some \( p \).

\noindent\textbf{3) Proposition 2.}
$\text{If}\ \sum_{i=0}^{k-1} \sigma_i \leq k$ for $k = 1, \dots, K$, and $\rho_{i-1} \geq \rho_i$, $i = 1, \dots, K-1$ with $\rho_i, \sigma_i \geq 0$, then $\sum_{i=0}^{K-1} \rho_i \sigma_i \leq \sum_{i=0}^{K-1} \rho_i$.

\textit{Proof.} Consider the linear program
\[
V = \max_{\sigma} \left\{
\begin{aligned}
  & \sum_{i=0}^{K-1} \rho_i \sigma_i : \sum_{i=0}^{k-1} \sigma_i \leq k, \\
  & k = 1, \dots, K, \ \sigma_i \geq 0, \ i = 0, \dots, K-1
\end{aligned}
\right\}
\]

with dual
$$
W = \min_{u}\left\{ 
\begin{aligned}
    & \sum_{k=1}^{K} k u_{k-1} : \sum_{k=i}^{K-1} u_k \geq \rho_i, \ i = 0, \dots, K-1, \\
    & \ u_k \geq 0, \ k = 0, \dots, K-1
\end{aligned} 
\right\}.
$$ 

As \( \rho_i \geq \rho_{i+1} \), the solution \( u_i = \rho_i - \rho_{i+1} \), \( i = 0, \dots, K-1 \) (where \( \rho_K = 0 \)) is dual feasible with value 
\[
\sum_{k=1}^K k(\rho_{k-1} - \rho_k) = \sum_{i=0}^{K-1} \rho_i.
\]
By weak linear programming duality,
\[
\sum_{i=0}^{K-1} \rho_i \sigma_i \leq V \leq W \leq \sum_{i=0}^{K-1} \rho_i.
\]

\noindent\textbf{4) Proposition 3.} 
If the greedy heuristic is applied to problem 1), then
$\frac{\Phi - \Phi^G}{\Phi - \varphi(\emptyset)} \leq \frac{P}{P+1}.$
The bound is tight for all $P$.

\textit{Proof.}
Let $T$ and $\Theta$ be optimal and greedy solutions respectively, with $|\Theta| = K$. 
For $k = 1, \dots, K$ let $\theta_{k-1} = |T \cap (U^k - U^{k-1})|$, where $U^{k-1}$ is the set of all elements 
considered during the first $t$ iterations before the addition of a $t$ th element to $\Theta^{t-1}$. 

Next we show 

$\text{(a) }\sum_{j \in T-\Theta} \rho_j(\Theta) \leq \sum_{i=1}^{K} \rho_{i-1} \theta_{i-1}$, and

(b) $\sum_{i=1}^{k} \theta_{i-1} \leq Pk$ for $k = 1, \dots, K$, which implies via Proposition 2 that 
$\sum_{i=1}^{K} \rho_{i-1} s_{i-1} \leq P \sum_{i=1}^{K} \rho_{i-1}$. The result then follows from Proposition 1 since
$$
\begin{aligned}
    \Phi &= \varphi(T) \leq z(\Theta) + \sum_{j \in T-\Theta} \rho_j(\Theta) \leq \Phi^G + P \sum_{i=1}^{K} \rho_{i-1} \\
    & = \Phi^G + P(\Phi^G - \varphi(\emptyset)).
\end{aligned}
$$
$$
\begin{aligned}
\text{(a)}\sum_{j \in T-\Theta} \rho_j(\Theta) \leq \sum_{j \in T} \rho_j(\Theta) &= \sum_{k=1}^{K} \sum_{j \in T \cap (U^t - U^{t-1})} \rho_j(\Theta) \\
&\leq \sum_{k=1}^{K} \rho_{k-1} s_{k-1}
\end{aligned}
$$
as $\rho_j(\Theta) \leq \rho_{k-1}$ for $j \in U^t - U^{t-1}$ by the nature of the greedy heuristic.

(b) From the definition of $s_{k-1}$, and assuming without loss of generality that $U^0 = \emptyset$, we have $\sum_{i=1}^k \theta_{i-1} = |T \cap U^k|$. Now by Proposition 1 $U^k \subseteq \bigcup_{p=1}^P \theta p^p(\Theta^k)$ so that $|T \cap U^t| \leq \sum_{p=1}^P |T \cap \Theta\text{p}^p(S^k)|$. But as $T$ is independent in matroid $(N, \mathscr{F}^p)$ and $r_p(\theta\text{p}^p(\Theta^t)) = k$, we obtain $|T \cap \theta\text{p}^p(\Theta^t)| \leq k$. Therefore $\sum_{i=1}^k \theta_{i-1} \leq P k$.

To show that the bound is tight we exhibit a family of problems with $X = \{1, \ldots, P + 2\}, \ \Phi^G = 1$ and $\Phi = P + 1$. We associate the variable $x_i \in \{0, 1\}$ with element $i$ and represent $\Theta \subseteq N$ by its characteristic vector; $x_i = 1$ if $i \in \Theta$ and $x_i = 0$ if $i \notin \Theta$. Consider the problem
$$
\begin{aligned}
  Z = \max \enspace
  &x_1 + x_2 + \cdots + x_{p+1} + x_{p+2} - x_1 x_{p+2} \\
  \\[-0.7cm]
  &x_1 + x_2  &\leq 1, \mathscr{M}_1 \hspace{0.45cm}\\
  \\[-0.7cm]
  &x_1  \enspace \enspace \enspace\enspace +x_3&\leq 1,\mathscr{M}_2 \hspace{0.45cm}\\
  \\[-0.9cm]
  & \vdots \! &\vdots \hspace{1.26cm} \\
  \\[-0.7cm]
  & x_1 \enspace \enspace \enspace\enspace\enspace\enspace+x_p&\leq 1, \mathscr{M}_{P-1} \hspace{0.08cm} \\
  \\[-0.7cm]
  & x_1  \enspace \enspace \enspace\enspace\enspace\enspace \enspace \enspace\enspace +x_{p+1}&\leq 1, \mathscr{M}_{P}\hspace{0.38cm}\\
  \\[-0.7cm]
  & x_1  \enspace \enspace \enspace\enspace\enspace\enspace \enspace \enspace\enspace\enspace\enspace\enspace\enspace\enspace+x_{p+2}&\leq 1, \mathscr{M}_P\hspace{0.38cm}
\end{aligned}
$$
\[
x_i \in \{0, 1\}, \quad i = 1, \dots, p+2.
\]


It is easy to verify that this quadratic objective function is submodular and nondecreasing. An optimal solution is given by $x_1 = 0$ and $x_i = 1$, $i = 2, \ldots, P + 2$ so that $\Phi = P + 1$. However, the greedy heuristic can select elements $1$ and $P + 2$ in the given order, which yields $\Phi^G = 1$. 

\noindent\textbf{5) final proof.} 
Based on the result from Proposition 3, $\frac{\Phi - \Phi^G}{\Phi - \varphi(\emptyset)} \\ \leq \frac{P}{P+1}$, we can get that the lower bound is $\frac{1}{P+1}$. In our situation, $P = \left\lceil \frac{\max a_{l}}{\min_{a_{l} > 0} a_{l}} \right\rceil + \left\lceil \frac{\max b_l}{\min_{ b_l > 0} b_l} \right\rceil$ is the number of matroids. So the result can be easily deduced. 
\end{proof}

\section{Environment Setting in Testbed}\label{section:appendix-testbed_setting}
\noindent Table~\ref{table:appendix-environment_settings} summarizes the hardware and software environment settings in our testbed experiments.
\begin{table}[!t]
	\centering
	\resizebox{0.8\linewidth}{!} {
	\begin{tabular}{|c||c|c|}
	\hline 
	&Environment&Parameter\\
	\hline
	\hline
	\multirow{4}{*}{\rotatebox{90}{Hardware$\ \ \ \ \ $}}&Mellonax CX6 NIC bandwidth & 100Gb/s \\
	\cline{2-3}
	&AS4610-54T switch bandwidth & 10Gb/s \\
	\cline{2-3}
	&Nvidia Tesla P100 VRAM & 16GB \\
	\cline{2-3}
	&Switch port number & 54 \\
    \cline{2-3}
    &U50 clock frequency & 200MHz \\
    \cline{2-3}
    &Basys 3 clock frequency & 100MHz \\
	\hline
	\multirow{4}{*}{\rotatebox{90}{Software$\ \ \ \ \ \ \ \ \ $}}&Nvidia driver & 550.54.14 \\						
	\cline{2-3}
	&Nvidia cuda version & 11.8 \\						
	\cline{2-3}
	&PyTorch~\cite{paszkePyTorchImperativeStyle2019} version & 2.5.1 \\			
	\cline{2-3}
	&Ubuntu version & 20.04 LTS \\
	\cline{2-3}
    &Xilinx vivado version & 2020.1 \\
    \cline{2-3}
	&\EPARA max offloading count & 5 \\
    \cline{2-3}
	&\EPARA placement mode & offline \\
	\hline 
	\end{tabular}
	}
	\vspace{10pt}
	\caption{Testbed environments.}\label{table:appendix-environment_settings}
	\vspace{-14pt}
\end{table}
\vfill 

\end{document}